\documentclass[letter,11pt]{article}
\usepackage[utf8]{inputenc}
\usepackage[T1]{fontenc}

\usepackage{amsthm,amsmath,bm,bbm}
\usepackage{amssymb,mathtools,nccmath}
\usepackage{fullpage}
\usepackage[dvipsnames]{xcolor}
\usepackage[colorlinks=true, allcolors=blue]{hyperref}
\usepackage{cleveref}

\usepackage{multirow}
\usepackage{caption}
\usepackage{tabularx}
\usepackage{capt-of}
\usepackage[compact]{titlesec}

\usepackage{enumitem}

\usepackage{natbib}
\setcitestyle{authoryear, open={(},close={)}}
\usepackage{minitoc}

\makeatletter
\renewcommand{\paragraph}{%
  \@startsection{paragraph}{4}%
  {\z@}{.5ex \@plus 1ex \@minus .2ex}{-1em}%
  {\normalfont\normalsize\bfseries}%
}
\makeatother

\usepackage{cancel}

\def\Rm{R_M}
\def\Ry{R_Y}

\def\independent{\perp\!\!\!\perp}
\def\E{\mathrm{E}}
\def\P{\mathrm{P}}
\def\odds{\mathrm{odds}}

\allowdisplaybreaks

\newtheorem{lemma}{Lemma}

\title{Self-separated and self-connected models for mediator and outcome missingness in mediation analysis}

\author{Trang Quynh Nguyen, Razieh Nabi, Fan Yang, Grace V. Ringlein, Elizabeth A. Stuart}

\begin{document}

\maketitle

\setlength\abovedisplayshortskip{0pt}
\setlength\belowdisplayshortskip{0pt}
\setlength\abovedisplayskip{3pt}
\setlength\belowdisplayskip{5pt}

\begin{abstract}
    \noindent
    Missing data is a common challenge in studying treatment effects. In the context of mediation analysis, this paper addresses missingness in the mediator and outcome, focusing on identification. We first consider self-separated missingness models where identification is achieved by conditional independence assumptions. This model class is somewhat limited as it is constrained by the need to remove a certain number of connections from the model. We then turn to self-connected missingness models where identification relies on information from shadow variables. This model class turns out to contain substantial variation, allowing models with built-in shadow variables (mediator, outcome or covariates) and models with auxiliary shadow variables at different positions in the causal structure. To improve the practical value of the missingness mechanisms, we allow where possible for dependencies due to unobserved causes of the missingness, a feature often neglected. In this exploration, we review existing models, connect to new models, and develop theory where needed. This results in templates for identification in the mediation setting, generally useful identification techniques, and perhaps most importantly a synthesis and substantial extension of shadow variable theory. Two examples relate the models to practical considerations.

    ~

\end{abstract}

\doparttoc 
\faketableofcontents 
\part{} 

\vspace{-3em}

\section{Introduction}
\label{sec:intro}

Mediation analysis is used in many fields to study the roles certain variables play on the pathway of causal effects. The last two decades have seen major developments in mediation analysis methodology, including the formalization of different estimands, clarification of identifying assumptions, and many estimation methods. Much of this literature assumes data are fully observed. Yet missing data is ubiquitous, and the combination of mediator and outcome missingness poses a greater challenge to mediation analysis than outcome missingness does in non-mediation settings \citep{lee2020NuisanceMediatorsMissing}. How to handle this combined missingness is an important research topic. This paper draws insights from recent advances in missing data theory to extend the range of identified models for mediator and outcome missingness and explore features of these models. In doing so, it helps solidify the knowledge base that serves as foundation for method development and ultimately model and method choices for practice.

The handling of missing data in mediation analysis (and other analysis types) in practice often involves methods that assume the missingness is at random (MAR) (such as nonresponse weighting, multiple imputation, full-information maximum likelihood) or is completely at random (MCAR) (such as complete-case analysis). A general concern, however, is that these assumptions may not hold. It has been repeatedly shown that analyses that assume MAR lead to biased results if the MAR assumption is not satisfied, for example, because the missingness of a variable is caused by the variable itself \citep{zhang2013MethodsMediationAnalysis,dashti2024HandlingMultivariableMissing}. This calls for models that relax MAR, which is the focus of the current work.

\subsection{Motivating examples}
\label{ssec:examples}

\paragraph{School intervention example.}
The Prevention of Alcohol Use in Students (PAS) trial in the Netherlands randomized middle schools to four conditions: student intervention (promoting healthy attitudes and strengthening refusal skills), parent intervention (encouraging parents' rule setting), combined student-and-parent intervention, and control. The combined intervention was effective in reducing the onset and frequency of drinking \citep{koning2009PreventingHeavyAlcohol,koning2011WhyTargetEarly}. Mediation analysis found that student attitudes, perceived self-control, and parental rules mediated these effects \citep{koning2011LongTermEffectsParent,nguyen2016CausalMediationAnalysis}. \cite{nguyen2016CausalMediationAnalysis} reported that there was missingness in the outcome (17.0\%), student-reported mediators (3.5\%) and parent-reported mediators (18.4\%). The mediation analyses opted to use student-reported rather than parent-reported parental rules to reduce missingness, and assumed MAR. It is possible, however, that outcome missingness was related to the outcome itself. For example, some students who were drinking may not have answered the question about drinking behavior; or some were doing poorly at school and therefore may have been absent on the day of data collection. In addition, missingness in parent-reported mediators could be due to unobserved family characteristics or socio-economic factors that are associated with those mediators.

\paragraph{Job training example.}
This example was discussed in \cite{zuo2024MediationAnalysisMediator}. The National Job Corps Study (NJCS) is a randomized trial in the United States to evaluate the Job Corps program, which targets youth disconnected from school. The program was found to improve employment and earnings \citep{schochet2008DoesJobCorps,schochet2021LongRunLaborMarket}. \cite{qin2019MultisiteCausalMediation} conducted a mediation analysis assuming MAR, finding that education attainment (by 30 months post-intervention) mediated the intervention effect on earnings (at 48 months). \cite{zuo2024MediationAnalysisMediator} noted substantial missingness, over 20\% in mediator and outcome each, and over 30\% combined. They were concerned that people who did not obtain an education credential may have been less likely to provide the information, or people with no earnings may have been less willing to report earnings. They reanalyzed data using models that allow either the mediator or the outcome, or both, to be related to its missingness. To accommodate such links, these models disallow certain other connections among the variables, which we will discuss.

As our focus is on the range of models that provide identification rather than inference based on a particular model, we will not reanalyze these data. We instead use these examples loosely with some artificial variations to relate the models to practical considerations.

\subsection{Related work}

The missing not at random (MNAR) literature has seen several important recent strands. In one strand \citep[e.g.,][]{zhou2010BlockConditionalMissingRandom,mohan2013GraphicalModelsInference,bhattacharya2020IdentificationMissingData,nabi2020FullLawIdentification}, it is shown that identification is possible for cases with certain conditional independence assumptions, even though MAR does not hold. To our knowledge, much of this theory has not been used to explore model options for mediation analysis explicitly. In a second strand \citep{ma2003IdentificationGraphicalModels,dhaultfoeuille2010NewInstrumentalMethod,kott2014CalibrationWeightingWhen,wang2014InstrumentalVariableApproach,zhao2015SemiparametricPseudoLikelihoodsGeneralized,miao2024IdentificationSemiparametricEfficiency,miao2016VarietiesDoublyRobust,yang2019CausalInferenceConfounders,li2023SelfcensoringModelMultivariate,zuo2024MediationAnalysisMediator,li2017IdentifiabilityEstimationCausal}, the non-identifiability due to dependence between a variable and its missingness is resolved by tapping into information in a  \textit{shadow variable} -- a special variable that is associated with the variable of interest but conditionally independent of the missingness. For mediation analysis, \citeauthor{zuo2024MediationAnalysisMediator} (\citeyear{zuo2024MediationAnalysisMediator}, cited for the second example above) examine several shadow variable models with mediator and outcome missingness; \cite{li2017IdentifiabilityEstimationCausal} use a shadow variable to handle missingness in the outcome alone. A third strand relies on a different type of special variable that (conditional on fully observed variables) is independent of the variable of interest but is associated with its missingness \citep{sun2018semiparametric,huber2020DirectIndirectEffectsa}. A fourth strand, which leverages process data (e.g., the number of calls required), departs from conditional independences and instead relies on a stable behavior tendency assumption \citep{miao2025StablenessResistanceModel}. Our paper builds on the first two of these literatures to expand the range of identified models for mediator and outcome missingness.

Another literature this work draws from concerns the testability of missingness models, i.e., whether a model implies certain conditions about the observed data distribution, which in principle can be verified against data as a way to falsify the model \citep{potthoff2006CanOneAssess,mohan2013GraphicalModelsInference,nabi2023TestabilityGoodnessFit}. Our discussion of models will cover model testability.

\subsection{Our contributions}

\subsubsection{On missing data in mediation analysis}

Our main contribution to this topic is to expand the range of identified models. We consider two broad classes of missingness models: models where a variable and its missingness do not cause each other and do not share completely unobserved causes (aka \textit{self-separated} missingness models); and models where a variable influences, or shares completely unobserved causes with, its missingness (aka \textit{self-connected} missing models). (This distinction is related, but not equivalent, to terminology such as ``no self-censoring'' \citep{malinsky2022SemiparametricInferenceNonmonotone} or ``self-masking'' \citep{mohan2019HandlingSelfmaskingOther}, which refer to whether a variable influences its own missingness, but not to whether a variable and its missingness are dependent due to unobserved causes.) In the class of self-separated missingness models, we show identified models other than MAR; and that MAR is arguably unnecessarily restrictive as it is dominated by sequentially ignorable missingness (SIM). In the class of self-connected missingness models, we find a wide range of shadow variable models, of which \citeauthor{zuo2024MediationAnalysisMediator}'s models constitute a subset.

A secondary contribution is our attention to settings with different temporal orders. The theory literature cited above includes models where one variable is influenced by another variable yet influences the missingness in that variable. We observe that this may cause confusion about how a variable causes the missingness of a variable in the past. To bring clarity on this matter, we make temporal ordering explicit. With mediator and outcome missingness, the missingness model can take on one of several temporal orders, depending on whether the mediator and outcome each are measured in-time or measured retrospectively. Not surprisingly, the options of identified models vary across those temporal order settings.

Regarding model testability, we find that a large number of the models have testable implications. These results are important as they may help screen out wrong models.

This work focuses on models and identification, and leaves estimation and inference out of its scope. To facilitate the next steps, we express identification results in specific forms that can be used in methods development.

The paper is motivated by causal mediation analysis, but is also relevant to traditional mediation analysis \citep{baron1986ModeratorMediatorVariableDistinction}, which can be considered a special case.

\subsubsection{Identification theory in general}

With regards to theory beyond the mediation context, our main contributions are in synthesizing and extending shadow variable theory. First, we present a simple version of shadow variable theory that ties together existing results, and accommodates missingness in the shadow variable under some conditions. It uses a connection between two odds of missingness -- one conditional on the variable itself and one conditional on the shadow variable -- which works well with certain odds-tilting techniques that we introduce. Second, we extend shadow variable theory to address the case with self-connected missingness in two \textit{main} variables and an \textit{external} shadow variable, finding three models that recover identification. These models apply generally to situations where the joint distribution of two variables is of interest, and can be adapted for more than two variables.

\subsection{Overview of the paper}

Section~\ref{sec:setting} introduces the setting and the key task of identifying the conditional mediator and outcome distributions in the presence of mediator and outcome missingness, then presents several initial considerations about the missingness, including temporal orders and large missingness models used as starting points. Section~\ref{sec:separated} examines self-separated missingness models where identification is achieved based on conditional independence assumptions. Section~\ref{sec:shadow} explores several subclasses of self-connected missingness models where identification relies on shadow variables. Both sections focus on identification and model falsifiability, and develop theory where needed. The two examples are discussed through various model types. Section~\ref{sec:conclusion} concludes with some remarks.

\section{Setting}
\label{sec:setting}

\subsection{Setting and estimands}

Consider the setting with a binary exposure $A$, a mediator $M$, a univariate outcome $Y$, and baseline covariates $X$. Using the potential outcomes framework \citep{rubin1974EstimatingCausalEffects,splawa-neyman1990ApplicationProbabilityTheory} and invoking consistency/SUTVA \citep{rubin1974EstimatingCausalEffects,cole2009ConsistencyStatementCausal}, let $M_a,Y_a$ denote the potential mediator and potential outcome variables that would materialize had exposure been set to $a\in\{0,1\}$; let $Y_{am}$ denote the potential outcome had exposure and mediator been set to $a$ and $m$, respectively, where $m$ is in the support of $M$ conditional on $X,A=a'$, and $a,a'\in\{0,1\}$. 

Throughout, we assume \textit{\underline{s}equential \underline{i}gnorability of exposure \underline{a}ssignment and mediator \underline{a}ssignment} (SIAA), which consists of two components:
\begin{align}
    A\independent(M_{a'},Y_{am})\mid X,\tag{\textsc{siaa}-1}
    \\
    M\independent Y_{am}\mid X,A=a.\tag{\textsc{siaa}-2}
\end{align}
This assumption focuses on the assigned exposure and assigned mediator, thus is weaker than \citeauthor{imai2010GeneralApproachCausal}'s (\citeyear{imai2010GeneralApproachCausal,imai2010IdentificationInferenceSensitivity}) sequential ignorability. We will work with a directed acyclic graph (DAG) model that respects SIAA (see Fig.~\ref{fig:siaa}) where $X$ captures all common causes of $A,M,Y$. This means an $X$ variable can share unobserved causes with at most one of these three variables, which is explicit in the unconditional graph representation of the model, and implicit in the graphs that condition on $X$. In this simple setting, there are no post-exposure mediator-outcome confounders.

\begin{figure*}[t!]
    \caption{A main model that satisfies SIAA, shown in unconditional graph (left) and in graphs conditional on $X$ (middle) and on $X,A$ (right). The model allows at most one of three types of unobserved common causes $U_{XA}$, $U_{XM}$, $U_{XY}$ (between $X$ and respectively $A$, $M$, $Y$).}
    \label{fig:siaa}
    \centering
    \includegraphics[width=.95\textwidth,page=1]{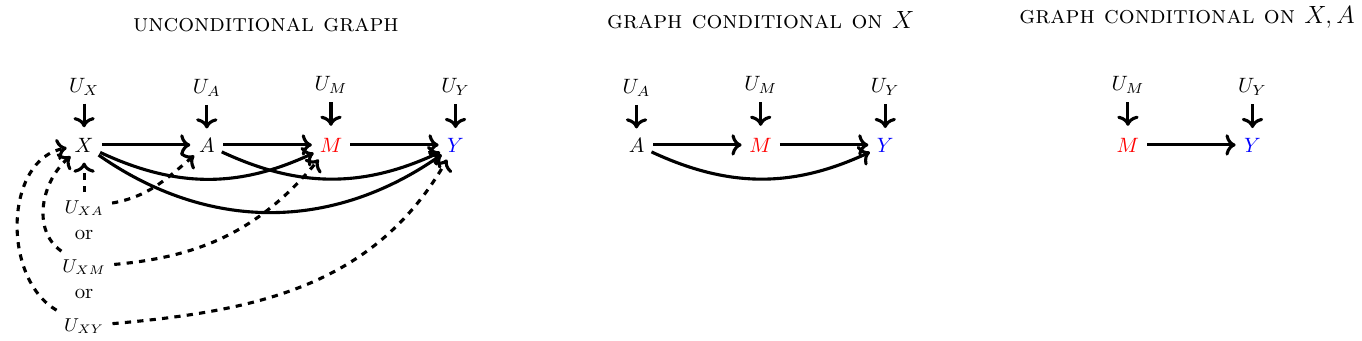}
\end{figure*}

Depending on the research question, causal mediation analyses may target different estimands, including natural (in)direct effects \citep{robins1992IdentifiabilityExchangeabilityDirect,pearl2001DirectIndirectEffects} and a large class of interventional effects \cite[see e.g.,][]{didelez2006DirectIndirectEffects,vanderweele2014EffectDecompositionPresence,nguyen2021ClarifyingCausalMediation,hejazi2023NonparametricCausalMediation}. Rather than targeting specific effect definitions, we will focus on the conditional mediator and outcome densities $\P(M\mid X,A)$ and $\P(Y\mid X,A,M)$. These appear in identification results for various effects -- under relevant assumptions. For example, in the absence of missing data, under SIAA (and relevant positivity assumptions), the potential outcome means contrasted by certain interventional (in)direct effects are identified by statistical parameters of the form
\begin{align*}
    \E(\E\{\E[Y\mid X,A=a,M]\mid X,A=a'\})=
    \iiint y\,{\color{blue}\P(y\mid x,m,a)}\,{\color{red}\P(m\mid x,a')}\,\P(x) \,dy \,dm \,dx.
\end{align*}
(If the cross-world independence assumption $M_{a'}\independent Y_{am}\mid X,A$ is additionally invoked, these also identify the potential outcome means contrasted by the natural (in)direct effects.) Also under SIAA, the potential outcome means contrasted by a controlled direct effect are identified by statistical parameters of the form
\begin{align*}
    \E\{\E[Y\mid X,A=a,M=m]\}=
    \iint y\,{\color{blue}\P(y\mid x,m,a)}\,\P(x) \,dy \,dx.
\end{align*}

\subsection{Mediator and outcome missingness: initial considerations}
\label{sec:preliminaries}

Our central task in this paper is to discuss mediator and outcome missingness models where $\P(M\mid X,A)$ and $\P(Y\mid X,A,M)$ are identified. To simplify language, we will call a model ``identified'' if these two densities are identified, otherwise we call the model ``unidentified.''

Let $R_M$ and $R_Y$ be binary variables indicating whether $M$ and $Y$, respectively, are observed. These variables are \textit{response indicators}, but we will also loosely refer to them by the common term \textit{missing indicators}. For notational convenience, let $R_{MY}:=R_MR_Y$.

\subsubsection{Temporal orders}
\label{subsec:temp-order}

In terms of temporal ordering, we require that $R_M$ comes after $M$ and $R_Y$ comes after $Y$. This allows four settings with different temporal orders, labeled \textit{in-time} ($M,R_M,Y,R_Y$), \textit{delayed} ($M,Y,R_M,R_Y$), \textit{unordered} ($M,Y,\begin{matrix}R_M\\[-.4em]R_Y\end{matrix}$), and \textit{reverse} ($M,Y,R_Y,R_M$). The \textit{in-time} setting is most common, with $M$ and $Y$ measured when or soon after they are realized (soon enough so $R_M$ precedes $Y$). In fact both the examples in section \ref{ssec:examples} belong in this category. The other orders may be relevant though if $M$ and/or $Y$ are measured retrospectively. For example, had the NJCS asked about educational attainment retrospectively at the same time earnings data were collected, $R_M$ and $R_Y$ would be \textit{unordered}. In a different example, \cite{benkeser2021InferenceNaturalMediation} consider the case-cohort sampling design where the mediator is measured (where measuring involves testing stored blood samples) only for individuals who develop the outcome and a subset of the individuals who do not. When there is also outcome missingness, this is a case of the \textit{reverse} setting.

\begin{figure*}[t!]
    \centering
    \caption{Starting missingness models for different temporal orders. Full-form (but not reduced-form) uses left-to-right layout of temporal order.}
    \label{fig:saturated}
    \includegraphics[width=.9\textwidth,page=2]{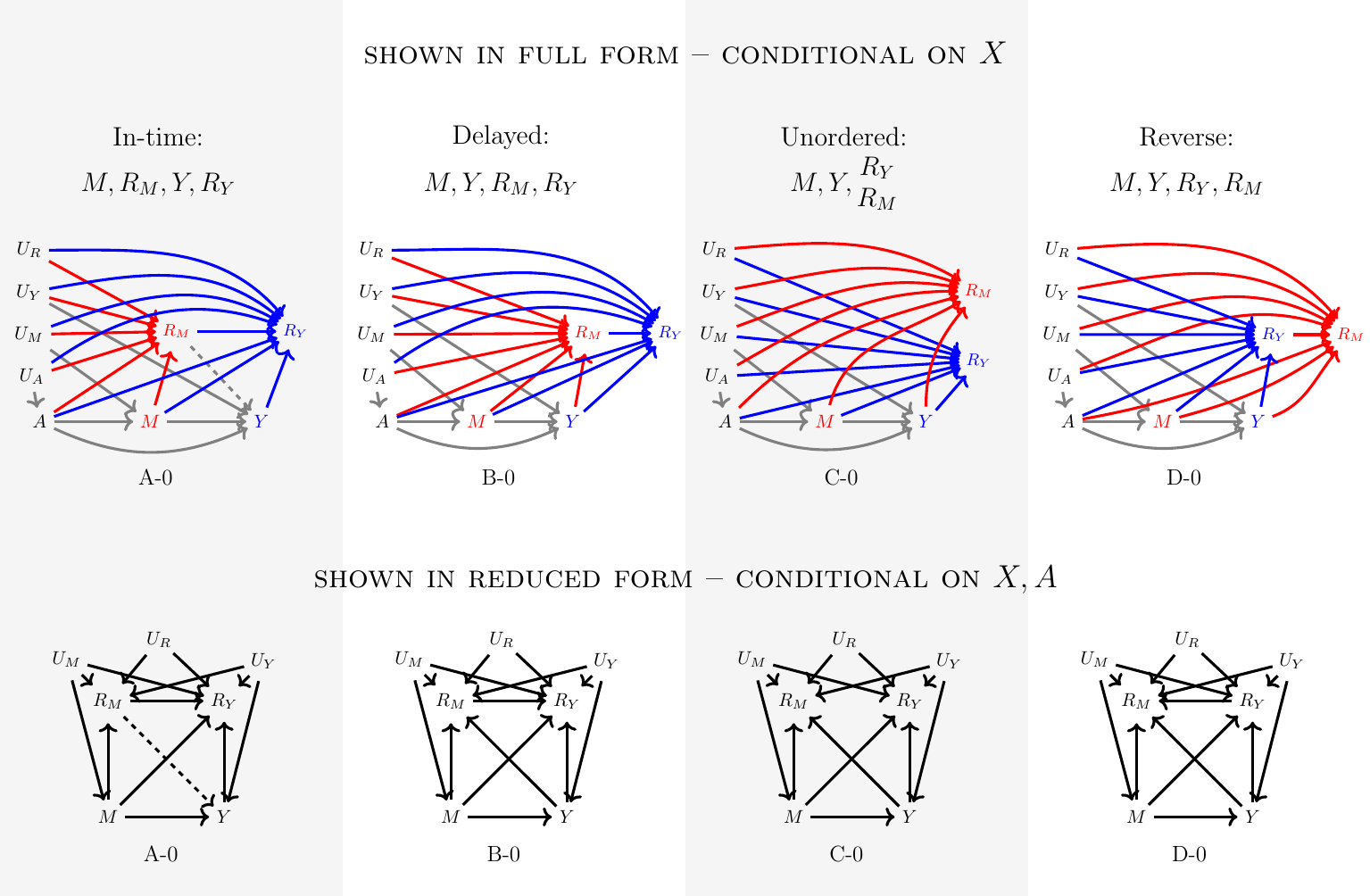}
\end{figure*}

\subsubsection{Starting with a large model}
\label{subsec:saturated}

Rather than immediately discussing identified models, we start with a large missingness generation model, as that will help make explicit what structures are assumed away in the more restrictive models that obtain identification. Our starting model is semi-saturated in the sense that it (i) allows the missingness of $M$ and $Y$ to be caused by variables in the main model and their unobserved causes, and to be causally related, while respecting temporal order; but (ii) excludes missingness causes that are on any causal path depicted by a single arrow in the main model, e.g., variables implicit in the $A\to M$ (or any other) arrow. Such missingness causes deserve separate attention.

Fig.~\ref{fig:saturated} presents this starting model for each of the four temporal orders. The top panel shows the models in what we call \textit{full form}, conditioning on $X$ and using a left-to-right portrayal of temporal order. The bottom panel shows them in \textit{reduced form}, conditioning on all the fully observed variables ($X,A$) and without the left-to-right format (a common representation in the theoretical literature on graphical missingness models). We label these models A-0, B-0, C-0, D-0, with the letter differentiating the settings and 0 indicating these models are our starting points. 

There are connections among these models: If we put temporal orders aside, C-0 is a submodel of B-0 and D-0. If we remove the $R_M\to Y$ path from A-0, it becomes a submodel of B-0.

A note on the dashed $R_M\to Y$ path in A-0: This type of path from a missingness indicator to a variable in the main model is typically excluded in the literature, with the implicit assumption that $Y$ remains the same regardless of whether any other variable is observed or missing. This path may however be relevant in some contexts. An example is where data collection on the mediator takes place on the same day and at the same place as an intervention session (either unintentionally or intentionally to reduce participant burden), then $R_M$ effectively is an indicator of being exposed to the session, which may affect the outcome. Another example is given in \cite{srinivasan2023GraphicalModelsEntangleda}, where showing up for a study visit may affect future outcome through feelings of self-empowerment. Interestingly, if $R_M\to Y$ is present, $U_M\to R_M$ needs to be absent for SIAA to hold (because otherwise there is $M$-$Y$ confounding); this point will be relevant in section~\ref{ssec:shadow-downstream}.

\subsubsection{Response positivity assumptions}

Throughout we assume the following two response positivity assumptions. 
Several models will require additional positivity assumptions, which will be noted then.

\begin{tabularx}{\linewidth}{@{}p{2.5cm} X@{}}
    $R_M$ positivity:
    & $\P(R_M=1\mid x,a,m)>0$ for all $(x,m)$ such that $\P(x,a,m)>0$ 
    \\
    $R_Y$ positivity:
    & $\P(R_{MY}=1\mid x,a,m,y)>0$ for all $(x,m,y)$ such that $\P(x,a,m,y)>0$
\end{tabularx}

\subsubsection{Non-identification when \texorpdfstring{$R_M$}{Rm} shares unobserved causes with \texorpdfstring{$Y$}{Y} or influences \texorpdfstring{$Y$}{Y}}
\label{sssec:nonID}

When there is either a $R_M\leftarrow U_Y\rightarrow Y$ or a $R_M\to Y$ path, $\P(Y\mid X,A,M)$ is generally not identified. This is because this density is a weighted average of $\P(Y\mid X,A,M,R_M=1)$ and $\P(Y\mid X,A,M,R_M=0)$, the latter of which is not identified.
This affects all the four large models in Fig.~\ref{fig:saturated}, as they all include $R_M\leftarrow U_Y\rightarrow Y$, and A-0 also allows $R_M\to Y$. 

An interesting case is the ``disjoint MAR'' model in Fig.~\ref{fig:badMAR}, which is a submodel of A-0 that retains the two problematic paths. (In the reduced form we use a double-headed arrow to abbreviate the path between two variables due to an unobserved cause when that cause is not shared by other variables.) In this model, the mediator is MAR ($R_M\independent M\mid X,A$, with observed $X,A$), which identifies $\P(M\mid X,A)=\P(M\mid X,A,R_M=1)$; the outcome is also MAR ($R_Y\independent Y\mid X,A,R_M$, with observed $X,A,R_M$), which identifies $\P(Y\mid X,A)=\E[\P(Y\mid X,A,R_M)\mid X,A]$. However we do not need $\P(Y\mid X,A)$. What we need is $\P(Y\mid X,A,M)$, which is not identified with this model. (To make this concrete, we provide a numerical example in Appendix~A.) This case shows that not all MAR conditions are helpful. Specifically, it is not helpful to have MAR conditional on the missing indicator ($R_M$) of a variable ($M$) conditional on which the estimand (here $\P(Y\mid X,A,M)$) is defined.

\begin{figure}[h!]
    \centering
    \caption{A ``disjoint MAR'' model}
    \label{fig:badMAR}
    \includegraphics[width=.4\textwidth,page=3]{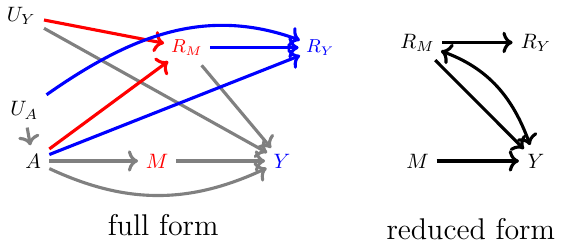}
\end{figure}

Due to this non-identification, we will not see models with $R_M\leftarrow U_Y\rightarrow Y$ or $R_M\to Y$ again until section \ref{ssec:shadow-downstream}, where some remedy is provided by an auxiliary variable.

\section{Self-separated missingness models}
\label{sec:separated}

In this section, we seek models with self-separated missingness that are submodels of the large models above where $\P(M\mid X,A)$ and $\P(Y\mid X,A,M)$ are identified based on certain conditional independence assumptions. We also examine the models' falsifiability, i.e., whether they imply conditions about the observed data distribution, which can in principle be verified.

Specifically, we cover five models (sets of conditional independence assumptions) in two groups (sections \ref{ssec:sim} and \ref{ssec:separation-other}). In between presenting models, we introduce a general identification approach with tweaks for the current setting (section \ref{ssec:approach}) and offer some useful identification techniques (section \ref{ssec:tilt-the-odds}). We later discuss this whole model class (section \ref{ssec:discussion-separated}) and relate back to the two examples (section \ref{ssec:examples-separated}).

The models are labeled S1 to S5 (S for self-\textit{separation}). They apply to one or more temporal settings, so setting-specific models are labeled accordingly, e.g., D-S2 (a submodel of D-0) or AB-S1 (a submodel of both A-0 and B-0, meaning A-S1 and B-S1 are the same model). Fig.~\ref{fig:separated} shows the DAGs for these models.
The presentation here focuses on the key points; identification formulas and testable implications are collected in Table~\ref{tab:separated}; and all proofs are provided in Appendix~B.

\subsection{Sequentially ignorable missingness models}
\label{ssec:sim}

There are two sequentially ignorable missingness (SIM) models.

\subsubsection{Model S1 (SIM type 1)} 

This is the least restrictive model that gives available-case identification (the simplest result): 
\begin{align*}
    \P(M\mid X,A)&=\P(M\mid X,A,R_M=1),
    \\
    \P(Y\mid X,A,M)&=\P(Y\mid X,A,M,R_{MY}=1).
\end{align*}

This model deserves close attention for it is the only option in this self-separation class of models for the \textit{in-time} setting.
In addition, it is also relevant to the \textit{delayed}, \textit{unordered} and \textit{reverse} settings (see AB-S1 and CD-S1 in Fig.~\ref{fig:separated}). The formal model assumption is stated below.

\paragraph*{Assumption S1:}

\begin{align}
    R_M&\independent(M,Y)\mid X,A,\tag{S1-1}\label{S1-1}
    \\
    R_Y&\independent Y\mid X,A,M,R_M.\tag{S1-2}\label{S1-2}
\end{align}

This assumption is about ignorability of $R_M$ followed by ignorability of $R_Y$, hence the \textit{sequentially ignorable missingness} (SIM) label. In settings where it is appropriate to (roughly) consider $R_M$ and $M$ as concurrent and $R_Y$ and $Y$ as concurrent, this statement has an intuitive meaning: missingness in the present is uninformative of the present and the future when conditioning on the past.

There is an alternative, equivalent, statement of assumption S1 as
\begin{align*}
    R_M&\independent M\mid X,A,\tag{S1-1b}\label{S1-1b}
    \\
    (R_M,R_Y)&\independent Y\mid X,A,M,\tag{S1-2b}\label{S1-2b}
\end{align*}
where each component identifies one of the two densities of interest: the first component (which identifies $\P(M\mid X,A)$) is a standard single-variable MAR condition; the second component (which identifies $\P(Y\mid X,A,M)$) is a generalized condition of the same flavor. 

Note that S1 is the two-variable case of \citeauthor{zhou2010BlockConditionalMissingRandom}'s (\citeyear{zhou2010BlockConditionalMissingRandom}) \textit{block-conditional MAR} model. \citeauthor{zhou2010BlockConditionalMissingRandom} constructed this model by first factoring the joint distribution into certain blocks (hence the name of the model) before imposing the same conditional independences.%
\footnote{
\citeauthor{zhou2010BlockConditionalMissingRandom}'s construction of the model, in our current notation:
{\scriptsize\begin{align*}
    \P(M,R_M,Y,R_Y\mid X,A)
    &=\P(M,R_M\mid X,A)\P(Y,R_Y\mid X,A,M,R_M)
    \\
    &=\P(M\mid X,A)\P(R_M\mid X,A,{\color{red}\cancel{M}})
    \P(Y\mid X,A,M,{\color{blue}\cancel{R_M}})\P(R_Y\mid X,A,M,R_M,{\color{red}\cancel{Y}}),
\end{align*}}

\noindent where the first step factors the joint density into $(M,R_M)$, $(Y,R_Y)$ blocks, and the second step imposes conditional independencies.}
Our presentation of the model here centers on the assumption, and is not tied to a specific factorization.

\begin{figure*}[tp!]
    \centering
    \caption{DAGs of self-separated missingness models S1-S5}
    \label{fig:separated}
    \includegraphics[width=.9\textwidth,page=4]{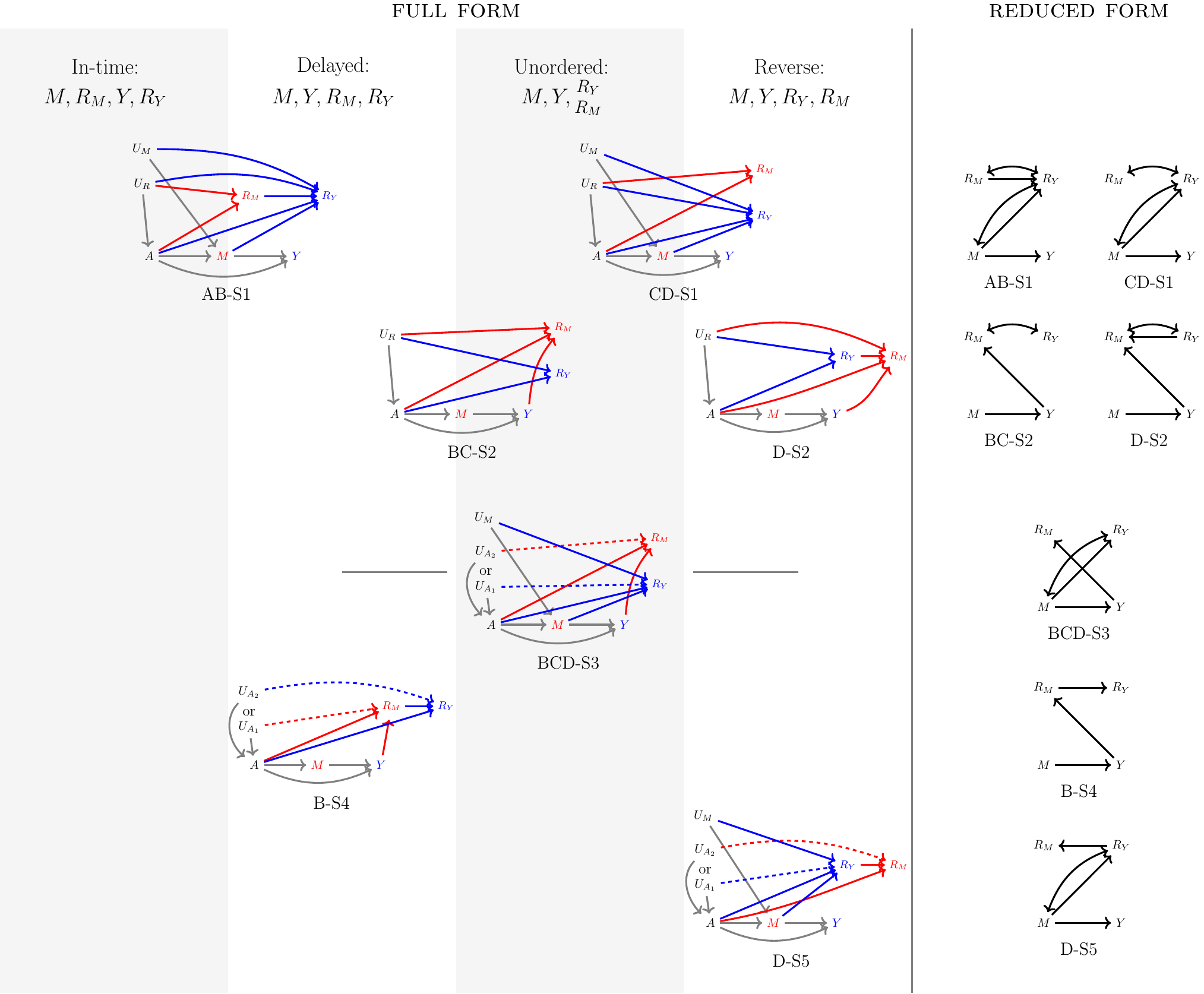}

    \bigskip
    
    \captionof{table}{Identification results and testable implications of self-separated missingness models S1-S5}
    \label{tab:separated}
    \resizebox{\linewidth}{!}{%
\begin{tabular}[b]{cl@{}c}
    Model & Tilting functions & \begin{tabular}{@{}c@{}}
        Testable\\implications
    \end{tabular}
    \\\hline
    \\[-.5em]
    S1
    & $h=1$,~~$k=1$.
    \\[.5em]
    S2
    & $h_b=\E[g_b\mid X,A,M,R_{MY}=1]$,~~$k=g_b/h_b$,~~where~~$g_b=\frac{\P(R_M=1\mid X,A,R_Y=1)}{\P(R_M=1\mid X,A,Y,R_Y=1)}$.
    \\[1em]
    S3
    & $h=\E[g\mid X,A,M,R_{MY}=1]$,~~$k=g/h$,~~where~~$g=\frac{\P(R_M=1\mid X,A)}{\P(R_M=1\mid X,A,Y,R_Y=1)}$.
    & \ref{T3}
    \\[1em]
    S4
    & $h=\E[g\mid X,A,M,R_{MY}=1]$,~~$k=g/h$,~~where
    \\
    &~~$g=\P(R_M=1\mid X,A)\left[1+\text{odds}(R_M=0\mid X,A,Y,R_Y=1)\frac{\P(R_Y=1\mid X,A,Y,R_M=1)}{\P(R_Y=1\mid X,A,Y,R_M=0)}\right]$.
    & \ref{T1}
    \\[1em]
    S5
    & $h=\P(R_M=1\mid X,A)\left[\frac{\P(R_Y=1\mid X,A,M,R_M=1)}{\P(R_M=1\mid X,A,R_Y=1)}+\frac{\P(R_Y=0\mid X,A,M,R_M=1)}{\P(R_M=1\mid X,A,R_Y=0)}\right]$,
    \\
    & $k=1$.
    & \ref{T2}
    \\[-.7em]
    \\\hline
\end{tabular}%
    }
\end{figure*}

There are two weaker versions of S1 that provide the same identification result. The weakest version starts with the second statement and replaces $(R_M,R_Y)$ in \ref{S1-2b} with $R_{MY}$. As $R_{MY}$ is a composite variable, however, it is not clear what this weaker condition means in practical terms. The other version starts with the first statement and replaces $R_M$ in \ref{S1-2} with $R_M=1$; this means outcome self-separation is only needed among those whose mediator is observed. Note that this version of SIM is similar in form to the SIAA assumption of the main model.

Three notes from the DAGs. First, S1 involves not only self-separation but also separation of $R_M$ and $Y$. For the \textit{delayed} (B-S1) and \textit{unordered} (C-S1) settings, this means $Y$ does not influence $R_M$, even though it precedes $R_M$. 

Second, for the \textit{reverse} setting, D-S1 restricts $R_Y$ not to influence $R_M$ (even though $R_Y$ precedes $R_M$) to satisfies assumption \ref{S1-1}. Strictly speaking, a different restriction (separation of $R_Y$ and $M$) would also satisfy \ref{S1-1}, but that model is more restrictive than model D-S2 (introduced shortly). In general, we will skip models that are unnecessarily restrictive.

Third, SIM allows $R_M$ and $R_Y$ to share unobserved causes. In the school intervention example, this means that it is fine if mediator and outcome missingness both results from school absence due to unobserved factors, as long as those factors do not affect $M$ or $Y$.

\subsubsection{Model S2 (SIM type 2)} 

This model, which is relevant  to the \textit{unordered} and \textit{reverse} settings (see C-S2 and D-S2 in Fig.~\ref{fig:separated}), involves self-separation and $R_Y$-$M$ separation. The assumption is a mirror image of S1, where the order of ignorability is first $R_Y$ then $R_M$.

\paragraph*{Assumption S2:}

\begin{align*}
    R_Y&\independent(M,Y)\mid X,A,\tag{S2-1}\label{S2-1}
    \\
    R_M&\independent M\mid X,A,Y,R_Y.\tag{S2-2}\label{S2-2}
\end{align*}

Although the two assumptions mirror each other, S2 is implicitly more restrictive regarding unobserved common causes, due to the ordering of $M,Y$. S1 allows $R_Y$ and $M$ to share unobserved causes, but S2 does not similarly allow $R_M$ and $Y$ to share unobserved causes, i.e., no $R_M\leftarrow U\rightarrow Y$ path (see section~\ref{sssec:nonID}). Adding such a path would make $Y$ a collider, rendering $\P(M\mid X,A,Y)$ in the factorization below unidentified.

Similar to S1, S2 has a second equivalent expression
\begin{align*}
    R_Y&\independent Y\mid X,A,\tag{S2-1b}\label{S2-1b}
    \\
    (R_M,R_Y)&\independent M\mid X,A,Y.\tag{S2-2b}\label{S2-2b}
\end{align*}
Also similar to S1, there are weaker versions of S2, where $R_Y$ in \ref{S2-2} is replaced with $R_Y=1$, and where $(R_M,R_Y)$ in \ref{S2-2b} is replaced with $R_{MY}$.

\paragraph*{Identification.} 
Under S2, one can identify $\P(M,Y\mid X,A)$ via factorization in the $Y$-then-$M$ order,
\begin{align*}
    \P(M,Y\mid X,A)
    &=\P(Y\mid X,A)\P(M\mid X,A,Y)
    \\
    &=\P(Y\mid X,A,R_Y=1)\P(M\mid X,A,Y,R_{MY}=1),
\end{align*}
and then derive the densities of interest, $\P(M\mid X,A)$ and $\P(Y\mid X,A,M)$, from this joint density. Alternatively, a more direct approach can be taken, which will be introduced in section~\ref{ssec:approach}. We will revisit model S2 then.

\subsubsection{A note on MAR}
\label{subsec:mar}

Earlier we mentioned a disjoint MAR model that does not achieve identification. There is a different MAR model that achieves identification, where $M,Y$ are \textit{jointly} MAR, formally 
\begin{align*}
    (R_M,R_Y)\independent(M,Y)\mid X,A.
\end{align*}
This model (see Fig.~\ref{fig:jointMAR}) is a submodel of both SIM models S1 and S2, and thus is unnecessarily restrictive. While S1 and S2 do not place restrictions on the observed data distribution, this MAR model implies two testable conditions \citep{potthoff2006CanOneAssess,mohan2021GraphicalModelsProcessing}:
\begin{align}
    R_Y&\independent M\mid X,A,R_M=1,\tag{T1}\label{T1}
    \\
    R_M&\independent Y\mid X,A,R_Y=1.\tag{T2}\label{T2}
\end{align}
There is thus no reason to adopt it over SIM models.

\begin{figure}[h]
    \centering
    \caption{The jointly MAR model}
    \label{fig:jointMAR}
    \vspace{.5em}
    \includegraphics[width=.35\linewidth,page=5]{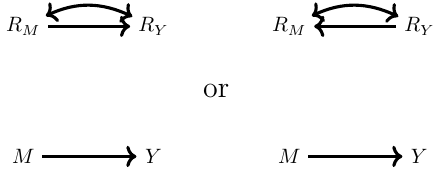}
\end{figure}

\subsection{Tailoring a general approach: Identifying response probabilities}
\label{ssec:approach}

Above we mentioned an alternate identification method for model S2. It is based on the general approach of recovering full-data distributions by identifying response probabilities. Here we tailor this approach to the current setting, to provide templates that are useful not only for S2 but also for other models in the remainder of the paper.

\subsubsection{General approach}

We start by defining \textit{tilting functions} that map from observed-data to full-data conditional densities of $M$ and $Y$ (the latter being of interest). Let 
\begin{align*}
    h(X,A,M)
    &:=\mfrac{\P(M\mid X,A)}{\P(M\mid X,A,R_M=1)},
    \\
    h_b(X,A,M)
    &:=\mfrac{\P(M\mid X,A)}{\P(M\mid X,A,R_{MY}=1)},
    \\
    k(X,A,M,Y)
    &:=\mfrac{\P(Y\mid X,A,M)}{\P(Y\mid X,A,M,R_{MY}=1)}.
\end{align*}
As the denominators above are observed data densities, $\P(M\mid X,A)$ is identified if either function $h$ or $h_b$ is identified, and $\P(Y\mid X,A,M)$ is identified if $k$ is identified. We will thus focus on identifying $h$ or $h_b$, and $k$.

By Bayes' rule, these tilting functions can be re-expressed as ratios of response probabilities,
\begin{align*}
    h(X,A,M)
    &=\mfrac{\P(R_M=1\mid X,A)}{\P(R_M=1\mid X,A,M)}
    \\
    h_b(X,A,M)
    &=\mfrac{\P(R_{MY}=1\mid X,A)}{\P(R_{MY}=1\mid X,A,M)}
    \\
    k(X,A,M,Y)
    &=\mfrac{\P(R_{MY}=1\mid X,A,M)}{\P(R_{MY}=1\mid X,A,M,Y)}.
\end{align*}
This means $\P(M\mid X,A)$ and $\P(Y\mid X,A,M)$ can be identified either by identifying the three response probabilities $\P(R_M=1\mid X,A,M)$, $\P(R_{MY}=1\mid X,A,M)$ and $\P(R_{MY}=1\mid X,A,M,Y)$ (which would identify $h$ and $k$) or by identifying the last two (which would identify $h_b$ and $k$).

\subsubsection{Tailoring to current setting}

For some models (S2 included), to identify two of those response probabilities ($\P(R_M=1\mid X,A,M)$ and $\P(R_{MY}=1\mid X,A,M)$) is complicated. To allow sidestepping them, we derive a different set of expressions for the tilting functions that only involves the third, $\P(R_{MY}=1\mid X,A,M,Y)$.
Specifically, define product functions
\begin{align*}
    g(X,A,M,Y)
    &:=hk=\mfrac{\P(R_M=1\mid X,A)\,\P(R_Y=1\mid X,A,M,R_M=1)}{\color{blue}\P(R_{MY}=1\mid X,A,M,Y)},
    \\
    g_b(X,A,M,Y)
    &:=h_bk=\mfrac{\P(R_{MY}=1\mid X,A)}{\color{blue}\P(R_{MY}=1\mid X,A,M,Y)},
\end{align*}
which are identified if $\color{blue}\P(R_{MY}=1\mid X,A,M,Y)$ is identified. 
The key is that $h$ and $h_b$ are expectations of $g$ and $g_b$ over the outcome distribution given $X,A,M$ among complete cases (details in the Appendix). This obtains new expressions for the tilting functions in terms of $g$ or $g_b$:
\begin{align*}
    h(X,A,M)
    &=\E[g\mid X,A,M,R_{MY}=1]
    \\
    h_b(X,A,M)
    &=\E[g_b\mid X,A,M,R_{MY}=1]
    \\
    k(X,A,M,Y)
    &=g/h=g_b/h_b.
\end{align*}
With these expressions, identification requires identifying only one response probability, $\color{blue}\P(R_{MY}=1\mid X,A,M,Y)$.

\paragraph*{Application to model S2.}

Under S2,
\begin{align*}
    {\color{blue}\P(R_{MY}=1\mid X,A,M,Y)}=
    \P(R_Y=1\mid X,A)\P(R_M=1\mid X,A,Y,R_Y=1),
\end{align*}
which obtains $g_b=\mfrac{\P(R_M=1\mid X,A,R_Y=1)}{\P(R_M=1\mid X,A,Y,R_Y=1)}$ and then the corresponding tilting functions $h_b$ and $k$.

\subsection{Self-separated missingness models other than SIM}
\label{ssec:separation-other}

\subsubsection{Model S3 (with \texorpdfstring{$R_M$}{Rm}-\texorpdfstring{$R_Y$}{Ry} separation)}

This is a well-known two-variable missingness model that appeared in \cite{mohan2013GraphicalModelsInference} as a case where the joint distribution of the variables is identified not by the traditional method of factoring and directly applying conditional independence assumptions, but by identifying the response probability. 
With the separation of $R_M$ and $R_Y$, this model can apply to the \textit{delayed}, \textit{unordered} and \textit{reverse} settings (see BCD-S3 in Fig.~\ref{fig:separated}).

\paragraph*{Assumption S3:}

\begin{align}
    R_M&\independent(M,R_Y)\mid X,A,Y,\tag{S3-1}\label{S3-1}
    \\
    R_Y&\independent(Y,R_M)\mid X,A,M.\tag{S3-2}\label{S3-2}
\end{align}

\paragraph*{Identification.}
Under S3,
\begin{align*}
    \P(R_{MY}=1\mid X,A,M,Y)=
    \P(R_M\!=\!1\mid X,A,Y,R_Y\!=\!1)\P(R_Y\!=\!1\mid X,A,M,R_M\!=\!1).
\end{align*}
This obtains $g=\mfrac{\P(R_M=1\mid X,A)}{\P(R_M=1\mid X,A,Y,R_Y=1)}$, and the corresponding tilting functions $h$ and $k$.

\paragraph*{Testable implication.}
\cite{nabi2023TestabilityGoodnessFit} show that this model is falsifiable. Without going into the complex theory here, the testable implication can be expressed as:
\smallskip
\begin{align}
    \mfrac{\P(R_M=0,R_Y=0\mid X,A)}{\P(R_M=1,R_Y=1\mid X,A)}=
    \E\Big[\mfrac{\P(R_M=0\mid X,A,Y,R_Y=1)}{\P(R_M=1\mid X,A,Y,R_Y=1)}\cdot\mfrac{\P(R_Y=0\mid X,A,M,R_M=1)}{\P(R_Y=1\mid X,A,M,R_M=1)}\mid X,A,R_{MY}=1\Big].\tag{T3}\label{T3}
\end{align}
This expression is deduced (see the Appendix) from \cite{nabi2023TestabilityGoodnessFit} and \cite{malinsky2022SemiparametricInferenceNonmonotone}. A test of this condition is offered in \cite{nabi2023TestabilityGoodnessFit}.

\subsubsection{Models S4 and S5}

These models apply to the \textit{delayed} and \textit{reverse} settings, respectively (see B-S4 and D-S5 in Fig.~\ref{fig:separated}). S4 involves $R_Y$-$M$ separation (like S2); S5 involves $R_M$-$Y$ separation (like S1). In S4 the missing indicators do not share unobserved causes with any other variables; in S5 $R_Y$ can share unobserved causes with $M$. The two model assumptions mirror each other.

\paragraph*{Assumption S4:}

\begin{align}
    &(R_M,R_Y)\independent M\mid X,A,Y,\tag{S4-1}\label{S4-1}
    \\
    &R_Y\independent(M,Y)\mid X,A,R_M,\tag{S4-2}\label{S4-2}
    \\
    &\P(R_Y=1\mid X,A,R_M=0)>0.\tag{S4-3}\label{S4-3}
\end{align}

\paragraph*{Assumption S5:}

\begin{align}
    &(R_M,R_Y)\independent Y\mid X,A,M,\tag{S5-1}\label{S5-1}
    \\
    &R_M\independent(M,Y)\mid X,A,R_Y,\tag{S5-2}\label{S5-2}
    \\
    &\P(R_M=1\mid X,A,R_Y=0)>0.\tag{S5-3}\label{S5-3}
\end{align}

\paragraph*{Identification.}
Under S4, $g=\mfrac{\P(R_M=1\mid X,A)}{\color{teal}\P(R_M=1\mid X,A,Y)}$, so we need to identify the denominator response probability (or its inverse). This is achieved using one of the \textit{odds-tilting} techniques explained shortly in section~\ref{ssec:tilt-the-odds}.

Under S5, $\P(Y\mid X,A,M)$ is identified by complete cases. To identify $\P(M\mid X,A)$, a simple strategy is to first identify $\P(M,R_Y\mid X,A)$ (by tilting $\P(M,R_Y\mid X,A,R_M=1)$) then marginalize over $R_Y$.

\paragraph*{Testable implications.}
Both models are falsifiable: S4 implies \ref{T1} and S5 implies \ref{T2}, which we copy here:
\begin{align}
    R_Y&\independent M\mid X,A,R_M=1,\tag{\ref{T1}}
    \\
    R_M&\independent Y\mid X,A,R_Y=1.\tag{\ref{T2}}
\end{align}

\subsection{Useful techniques: Tilting the odds}
\label{ssec:tilt-the-odds}

There are different ways to identify $\frac{1}{\color{teal}\P(R_M=1\mid X,A,Y)}$ in model S4. We introduce two \textit{odds-tilting} techniques that are generally useful. (The first plays a key role in a later theory on shadow variable missingness in section~\ref{ssec:shadow-missing}.) These are tilting techniques in the sense that they ``tilt'' an identified function toward a function of interest. But unlike the earlier tilting (by $h,h_b,k$) which targets densities of $M,Y$, the tilting here targets \textit{missingness odds}. 
(It is simpler to work with odds than with probabilities, and the inverse response probability is equal to 1 plus the missingness odds, e.g., $\frac{1}{\P(R_M=1\mid X,A,Y)}=1+\frac{\P(R_M=0\mid X,A,Y)}{\P(R_M=1\mid X,A,Y)}$.) To lighten notation, we will often write ``odds'' instead of spelling out a ratio of two probabilities.

Centering on the problem of model S4, the target is $\color{teal}\odds(R_M=0\mid X,A,Y)$. There are two choices for the observed-data odds one may anchor on: $\odds(R_M=0\mid X,A,Y,R_Y=1)$ (conditioning additionally on observing $Y$) and $\odds(R_M=0\mid X,A)$ (removing the conditioning on $Y$). These give two potential odds-tilting functions:
\begin{align}
    \mfrac{\color{teal}\odds(R_M\!=\!0\mid X,A,Y)}{\odds(R_M\!=\!0\mid X,A,Y,R_Y\!=\!1)}&=\mfrac{\P(R_Y\!=\!1\mid X,A,Y,R_M\!=\!1)}{\P(R_Y\!=\!1\mid X,A,Y,R_M\!=\!0)},\label{eq:odds-tilt1}
    \\
    \mfrac{\color{teal}\odds(R_M\!=\!0\mid X,A,Y)}{\odds(R_M\!=\!0\mid X,A)}&=\mfrac{\P(Y\mid X,A,R_M\!=\!0)}{\P(Y\mid X,A,R_M\!=\!1)}.\label{eq:odds-tilt2}
\end{align}
This means if we can identify either the response risk ratio on the RHS of (\ref{eq:odds-tilt1}) or the outcome density ratio on the RHS of (\ref{eq:odds-tilt2}), we have a tilting function that identifies $\color{teal}\odds(R_M=0\mid X,A,Y)$. 

\ref{S4-2} and \ref{S4-3} combined identify both tilting functions, as they allow (i) removing $Y$ from the conditioning sets on the RHS of (\ref{eq:odds-tilt1}), obtaining $\mfrac{\P(R_Y=1\mid X,A,R_M=1)}{\P(R_Y=1\mid X,A,R_M=0)}$, and (ii) adding $R_Y=1$ to the conditioning sets on the RHS of (\ref{eq:odds-tilt2}), obtaining  $\mfrac{\P(Y\mid X,A,R_M=0,R_Y=1)}{\P(Y\mid X,A,R_M=1,R_Y=1)}$.

\subsection{Comments on this class of models}
\label{ssec:discussion-separated}

This section covers five models, but the options depend on the temporal order setting. For the common \textit{in-time} setting there is a single option: S1. 
The other settings have more options thanks to the altered ordering of $R_M,R_Y$ vis-\`a-vis $M,Y$ and vis-\`a-vis each other, but 
the practical choices for application are likely limited.
Take model S3 for example, and consider it from a time perspective. In the \textit{delayed} setting ($M,Y,R_M,R_Y$), where $M$ and $Y$ are both measured retrospectively with $R_M$ preceding $R_Y$, the case where $M$ affects $R_Y$ but skips $R_M$ may be rare. Therefore, time-wise, S3 may be less reasonable for the \textit{delayed} setting than for the \textit{unordered} ($M,Y,\begin{matrix}R_M\\[-.4em]R_Y\end{matrix})$ and \textit{reverse} ($M,Y,R_M,R_Y$) settings. 
Yet for the latter two settings, a separate consideration for S3 is the plausibility of $M$ and $Y$ both affecting each other's missingness but not affecting their own missingness.

Now we put aside specific time considerations and just consider the reduced forms of the models. Three observations can be made. First, out of the five possible pairwise connections for the missingness ($M$-$R_M$, $Y$-$R_Y$, $M$-$R_Y$, $Y$-$R_M$ and $R_M$-$R_Y$), only two are present in each model. As these are self-separated missingness models, the two self-connections are absent, and additionally, one of the other three connections ($M$-$R_Y$, $Y$-$R_M$ and $R_M$-$R_Y$) is absent. Adding the absent third connection renders the models unidentified. Choosing among these models means choosing which of these three connections one is willing to assume away.

Second, only SIM models allow $R_M$ and $R_Y$ to share unobserved causes, a point to be considered. Models S4 and S5, for example, do not allow such unobserved common causes, even though they allow $R_M$ to affect $R_Y$ or vice versa. We note that the distinction between these two types of connection seems to be under-appreciated, and researchers may draw a DAG with a directed arrow from one missing indicator to another based on a simple argument that the two may be associated. Yet there are situations where $R_M$ and $R_Y$ may be associated more because they share causes than because one causes the other, e.g., where the measurements are far apart and the missingness is non-monotone and does not simply reflect study attrition. The key point is, these models do not hold if factors that may affect both $R_M$ and $R_Y$ (say, in some context, poor physical mobility, mental health struggles, or sporadic cell phone service access) are not captured in $X$.

Third, all non-SIM models imply testable conditions. If those conditions do not hold, then we are back to SIM as the only option(s) within this class of self-separated missingness models. (Note though that this does not mean SIM is correct, as the true model may not be in this class.)

\subsection{Relating to the examples}
\label{ssec:examples-separated}

As our two examples are both in the \textit{in-time} category, the question
is whether SIM, specifically S1, might hold. This model allows $R_M$-$R_Y$ connection (due to direct influence and/or unobserved common causes), but assumes away the self-connections and the $R_M$-$Y$ connection. 

For the latter, it may be reasonable to assume there is no $R_M\to Y$ path (it is much more plausible for the mediator than for its missing/observed status to affect the outcome), but whether there is a $R_M\leftarrow U_Y\to Y$ path requires careful thought. In the school intervention example, poverty and challenging home environment might increase school absence (causing mediator missingness) and might influence drinking behavior. If these factors are measured and included in the covariate set $X$, then it would be more plausible for this back-door path to be absent.

With the self-connections, there are also two possibilities: direct influence and back-door path. Back-door paths can be blocked if common causes of a variable and its missingness are measured (except some subtleties to be explained in section \ref{sssec:example-MY}). However, if there is concern that the outcome or mediator influences its missingness (students who drink and people who did not obtain an education credential or have zero earnings being less forthcoming with that information), such self-connections rule out the SIM model. This is a challenge with self-report of potentially stigmatizing behaviors or characteristics.

\section{Self-connected missingness models with shadow variables}
\label{sec:shadow}

This section concerns situations where, unlike in earlier models, $M$ and/or $Y$ are not separated from their missingness, but the densities of interest are identified by leveraging \textit{shadow variables} \cite[SV, term used in][]{kott2014CalibrationWeightingWhen,miao2015IdentificationDoublyRobusta} and certain distribution completeness conditions. Outside of the mediation setting, SVs have been used to handle non-ignorable missingness and biased selection by many authors \cite[e.g.,][]{dhaultfoeuille2010NewInstrumentalMethod,kott2014CalibrationWeightingWhen,wang2014InstrumentalVariableApproach,zhao2015SemiparametricPseudoLikelihoodsGeneralized,miao2024IdentificationSemiparametricEfficiency,miao2016VarietiesDoublyRobust,yang2019CausalInferenceConfounders,li2023SelfcensoringModelMultivariate}.
In the mediation setting, \cite{zuo2024MediationAnalysisMediator} study several models with SVs. \cite{li2017IdentifiabilityEstimationCausal} consider a SV model in handling outcome (but not mediator) missingness. We will build on this prior work in several important steps. First we synthesize a general SV theory that is simple and accommodates missingness in the SV, and apply it to provide a review with updates of models with mediator/outcome as SVs. Then we extend the theory to consider external SVs for missingness in a pair of variables, and apply it to uncover a range of SV models for mediator and outcome missingness, with auxiliary variables at different locations in the causal structure or with covariates serving as SVs.
Proofs of all results are provided in Appendix~C.

\subsection{A general theory}
\label{ssec:shadow-theory}

Let $V$ denote a variable of interest, with missing indicator $R_V$. The task is to identify the distribution of $V$ given a set of variables. For simple presentation, we keep this conditioning implicit.
The theory has two key components: a SV and a distribution completeness condition. An additional third component accommodates missingness in the SV.

\subsubsection{A shadow variable}

Suppose there exists a variable $Z$ that is statistically dependent on $V$, but is independent of $R_V$ given $V$. Such a variable is called a SV with respect to the missingness in $V$ (intuitively, a variable in the shadow of $V$ from the point of view of $R_V$). The top panel of Fig.~\ref{fig:shadow-theory} shows several cases that satisfy this condition. The dependence between $Z$ and $V$ may be due to one variable causing the other and/or the two sharing common causes.

This setup provides a key connection (see proof in the Appendix) between two odds of missingness:
\begin{align}
    \underbrace{\frac{\P(R_V=0\mid Z)}{\P(R_V=1\mid Z)}}_{\odds(R_V=0\mid Z)}=\E\bigg[\underbrace{\frac{\P(R_V=0\mid V)}{\P(R_V=1\mid V)}}_{\odds(R_V=0\mid V)}\mid Z,R_V=1\bigg].\label{eq:shadow}
\end{align}
That is, the \textit{odds of missingness given the SV} is equal to the expectation of the \textit{odds of missingness given the variable itself} where the expectation is taken among observed cases, conditional on the SV. 
As the LHS and the conditional distribution $\P(V\mid Z,R_V=1)$ (over which the expectation is taken) are observed, (\ref{eq:shadow}) provides an integral equation with the unknown inside the expectation. If this equation has a unique solution then $\P(R_V=1\mid V)$ is identified and the distribution of $V$ is identified.

We draw attention to the significance of (\ref{eq:shadow}). While this key connection may be expressed in other ways,
the expression we present here keeps identification simple and ties existing results together.
It is simpler than the identification theory in \cite{miao2024IdentificationSemiparametricEfficiency}, which involves identifying an odds ratio function which references a baseline odds function.  It is implicitly seen (with some re-expression) in the proofs of identification for specific cases in \cite{zuo2024MediationAnalysisMediator}.
(For details on these points, see the Appendix.)
Also, an advantage of this expression is that it combines well with odds-tilting techniques and thus is readily extendable to the case with missingness in the SV (see section~\ref{sssec:shadow-missing} below).%
\footnote{An alternative expression of the same connection, $\E\left[\mfrac{R_V}{\P(R_V=1\mid V)}\mid Z\right]=1$ \citep{dhaultfoeuille2010NewInstrumentalMethod}, is simpler but does not extend as easily to accommodate missingness in $Z$.}

When applying (\ref{eq:shadow}) in the next sections, we will often use the \textit{inverse response probability} version obtained by adding 1 to both sides  of (\ref{eq:shadow}), 
$\mfrac{1}{\P(R_V=1\mid Z)}=\E\left[\frac{1}{\P(R_V=1\mid V)}\mid Z,R_V=1\right]$.
With an abuse of labels, we refer to this equation also as (\ref{eq:shadow}).

\begin{figure*}[t!]
    \centering
    \caption{Top panel: three cases where $Z$ serves as a SV with respect to the missingness in $V$, including cases (a) and (b) where $V$ influences $Z$ and case (c) where $Z$ influences $V$. Bottom panel: same settings now with missingness in $Z$ that satisfies mSV-2, including cases (i) where $R_Z\independent(V,Z)\mid R_V$ (blue arrow allowed but not gray arrows) and cases (ii) where $R_Z\independent R_V\mid V,Z$ (gray arrows allowed but not blue arrow).}
    \label{fig:shadow-theory}
    \includegraphics[width=.75\textwidth,page=6]{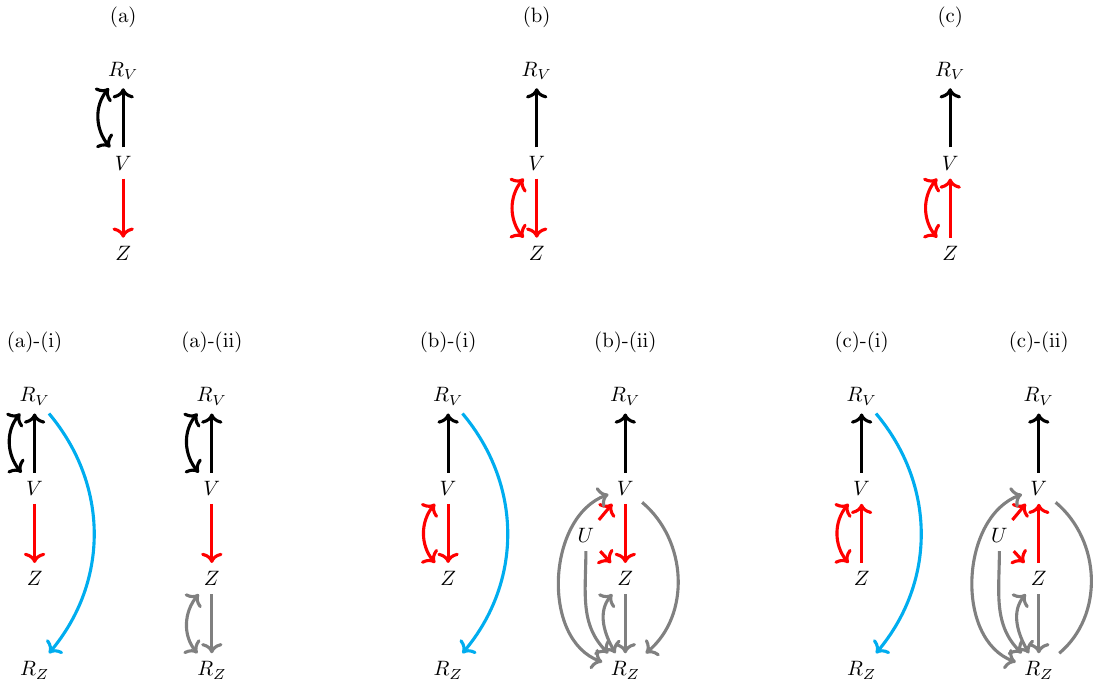}
\end{figure*}

\subsubsection{A distribution completeness condition}

The second ingredient of the theory is a completeness condition: 
$\P(V,Z\mid R_V=1)$ is complete in $Z$, which means for any function $t(V)$ with finite second moment, if $\E[t(V)\mid Z,R_V=1]\overset{\text{a.s.}}{=}0$ it must be the case that $t(V)\overset{\text{a.s.}}{=}0$  \citep{lehmann_completeness_1950}. If this holds, our integral equation has a unique solution.
Intuitively, this condition means there is more variation in $Z$ than in $V$; if the opposite is true then there is not enough information (fewer equations than unknowns in the discrete case) so the solution is not unique.
This condition is about the observed data distribution, 
and does not involve the missing data.
As completeness implies $Z\not\independent V$, the theory only requires $Z\independent R_V\mid V$ and the completeness condition.

The completeness condition is key in this identification theory. There has been some work aiming to relax this requirement, however, for missing data and also for proximal causal inference \citep[see][]{li2023NonparametricInferenceMean,zhang2023ProximalCausalInferencea}.

\subsubsection{Missingness in the shadow variable}
\label{sssec:shadow-missing}

If the SV is subject to missingness, (\ref{eq:shadow}) cannot immediately be used as the integral equation to identify $\odds(R_V=0\mid V)$, so it helps to extend the theory to accommodate SV missingness. 
Let $R_Z$ be the response indicator for $Z$. Also, let $R_{VZ}:=R_V R_Z$ and $Z^\dagger:=(R_Z,R_ZZ)$.
Identification is possible under some restriction on the $R_Z$ model, formalized in assumption mSV.

\paragraph*{Assumption mSV:}
\begin{align}
    \P(&R_Z=1,R_V=1\mid V,Z)>0\tag{mSV-1},\label{mSV-1}
    \\
    \text{and either}~
    &R_Z\independent(V,Z)\mid R_V\tag{mSV-2i}\label{mSV-2i}
    \\
    \text{or}~
    &R_Z\independent R_V\mid V,Z.\tag{mSV-2ii}\label{mSV-2ii}
\end{align}

Loosely speaking, this assumption says that $R_Z$ is either separated from $(V,Z)$ or separated from $R_V$, or both. (This restriction is used for a specific case in the appendix of \citeauthor{zuo2024MediationAnalysisMediator}, \citeyear{zuo2024MediationAnalysisMediator}.) The bottom panel of Fig.~\ref{fig:shadow-theory} illustrates this restriction. In this figure, we use an undirected edge ($V-R_V$ in (c)-(ii)) to indicate an arrow that can be in either direction (either $V\to R_Z$ or $R_Z\to V$). Three points to note from this figure: 
(i) $R_Z$ is not allowed to influence $R_V$ or share unobserved causes with $R_V$; 
(ii) blue and gray arrows are not simultaneously allowed; 
(iii) in the absence the blue arrow, $Z^\dagger$ is a SV.

\paragraph*{Potential integral equations.}

If  \ref{mSV-2i} holds ($R_Z$ is separated from $(V,Z)$), (\ref{eq:shadow}) implies
\begin{align}
    \odds(R_V=0\mid Z,R_Z=1)\mfrac{\P(R_Z=1\mid R_V=1)}{\P(R_Z=1\mid R_V=0)}=\E[\odds(R_V=0\mid V)\mid Z,R_{VZ}=1],\label{eq:shadow-blue}
\end{align}
which can serve as the integral equation for $\odds(R_V=0\mid V)$.
(The derivation of this result uses the first odds-tilting technique in section~\ref{ssec:tilt-the-odds}.)

If \ref{mSV-2ii} holds ($R_Z$ is separated from $R_V$), then we have two potential integral equations:
\begin{align}
    \label{eq:shadow-red}\odds(R_V=0\mid Z,R_Z=1)
    &=\E[\odds(R_V=0\mid V)\mid Z,R_{VZ}=1],
    \\
    \label{eq:shadow-red-dagger}\odds(R_V=0\mid Z^\dagger)
    &=\E[\odds(R_V=0\mid V)\mid Z^\dagger,R_V=1]
\end{align}
(based on $Z$ being a SV among those with $Z=1$ and based on $Z^\dagger$ being a SV). Which one should be used depends on the completeness condition (explained shortly).

Like with (\ref{eq:shadow}), inverse response probability versions of (\ref{eq:shadow-blue}), (\ref{eq:shadow-red}), (\ref{eq:shadow-red-dagger}) are obtained by adding 1 to both sides of each equation. We will refer to those by the same labels, (\ref{eq:shadow-blue}), (\ref{eq:shadow-red}), (\ref{eq:shadow-red-dagger}).

\paragraph*{Completeness conditions.}

Under \ref{mSV-1}, the original completeness condition implies that $\P(V,Z\mid R_{VZ}=1)$ is complete in $Z$. This can be used as the completeness condition to ensure the integral equation have a unique solution in both the \ref{mSV-2i} and \ref{mSV-2ii} cases. 

The \ref{mSV-2ii} ($R_V$-$R_Z$ separation) case is interesting with two options for the integral equation. (\ref{eq:shadow-red-dagger}) has more data points than (\ref{eq:shadow-red}), as it is the combination of (\ref{eq:shadow-red}) and
\begin{align}
    \label{eq:shadow-red-extra}\odds(R_V=0\mid R_Z=0)=
    \E[\odds(R_V=0\mid V)\mid R_Z=0,R_V=1].
\end{align}
Hence (\ref{eq:shadow-red-dagger}) has a unique solution under a weaker completeness condition: $\P(V,Z^\dagger\mid R_V=1)$ complete in $Z^\dagger$. This condition was used by \cite{zuo2024MediationAnalysisMediator} for a model discussed in section~\ref{sssec:models8-9} below. To differentiate the two completeness conditions in this case, we refer to them as \textit{strong} and \textit{weak completeness}. (A setting where strong completeness is not satisfied but weak completeness may be satisfied is where $V$ and $Z$ are both discrete, $V$ has $k$ values and $Z$ has $k-1$ values.) 
If strong completeness holds, $\odds(R_V=0\mid V)$ is identified by the unique solution to (\ref{eq:shadow-red}); and if strong completeness does not hold but weak completeness does, $\odds(R_V=0\mid V)$ is identified by the unique solution to (\ref{eq:shadow-red-dagger}). 

\paragraph*{A testable implication under strong completeness in the $R_V$-$R_Z$ separation case.}

Note that (\ref{eq:shadow-red-extra}) is used for identification (as part of (\ref{eq:shadow-red-dagger})) under weak completeness, but not under strong completeness. This means under strong completeness, (\ref{eq:shadow-red-extra}) serves as a testable condition, reflecting the assumed conditional independence of $R_V$ and $R_Z$.

\subsection{Models with mediator and/or outcome as shadow variables}
\label{ssec:shadow-MY}

In applying the theory above, we start with a subclass of SV models where $M$ and/or $Y$ serve as SVs (shown in Fig.~\ref{fig:shadow-MY}). This section is a review and extension of models in \cite{zuo2024MediationAnalysisMediator}. This prior work proves identifiability for three models -- one where $M$ serves as a SV, one where $Y$ serves as a SV, and one where both serve as SVs. Here we present those models in more general graphical forms that allow the presence of some unobserved common causes, and add two models that apply to the \textit{reverse} and \textit{unordered} settings. Where \citeauthor{zuo2024MediationAnalysisMediator} use \textit{weak completeness}, we add the option of \textit{strong completeness}. We point out the models' testable implications,
which previously were not considered.
We make explicit the tilting functions $h$ (or $h_b$) and $k$ for all models.

We will briefly present the models (labeled Z1 to Z5 to recognize \citeauthor{zuo2024MediationAnalysisMediator}'s work on this model class), and will discuss them all at the end of the section.
The text mainly highlights key points; full derivations are relegated to the Appendix. 
Results are summarized in Table~\ref{tab:shadow-MY}.

To connect to the broader literature, \citeauthor{zuo2024MediationAnalysisMediator}'s (\citeyear{zuo2024MediationAnalysisMediator}) DAGs coincide with some DAGs in \cite{ma2003IdentificationGraphicalModels}. This prior work considers binary variables, does not explicitly discuss completeness, and focuses on a general identification algorithm not limited to mediation.

\begin{figure*}[tp!]
    \centering
    \caption{DAGs of self-connected missingness models where $M$ and/or $Y$ serve as shadow variables. For completeness conditions, see model assumptions in text.}
    \label{fig:shadow-MY}
    \includegraphics[width=.9\textwidth, page=7]{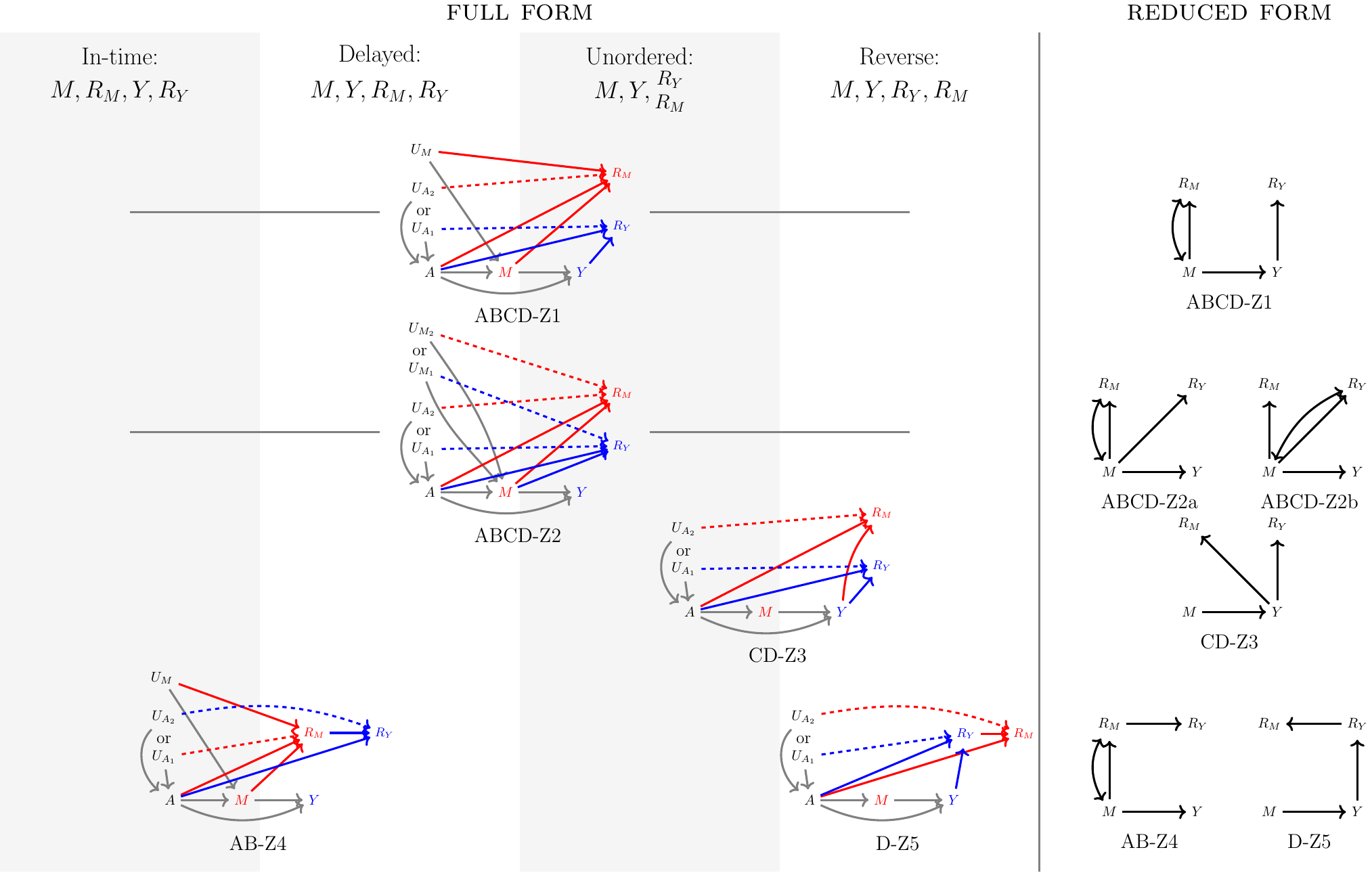}
    \bigskip
    
    \captionof{table}{Identification results and testable implications of models Z1-Z5}
    \label{tab:shadow-MY}
    \resizebox{\linewidth}{!}{%
\begin{tabular}{cl@{}c}
    Model
    & Tilting functions*
    & \begin{tabular}{@{}l}
        Testable\\implications
    \end{tabular}
    \\\hline
    \\[-.5em]
    Z1
    &
    \begin{tabular}[t]{@{}l@{}}
        $h=\P(R_M=1\mid X,A)q_{1a}(X,A,M)$,
        \\
        $k=\P(R_Y=1\mid X,A,M,R_M=1)q_{2a}(X,A,Y)$.
    \end{tabular}
    & 
    \ref{T4}, \ref{T5}
    \\[2em]
    Z2
    &
    \begin{tabular}[t]{@{}l@{}}
        $h=\P(R_M=1\mid X,A)\times\begin{cases}
            q_{1a}(X,A,M)~\text{if strong completeness}
            \\
            q_{1b}(X,A,M)~\text{if weak (not strong) completeness}
        \end{cases}$\!\!\!\!,
        \\
        $k=1$.
    \end{tabular}
    &
    \begin{tabular}[t]{@{}l@{}}
        \ref{T4} if strong
        \\
        completeness
    \end{tabular}
    \\[3em]
    Z3
    &
    \begin{tabular}[t]{@{}l@{}}
    $h_b=\E[g_b\mid X,A,M,R_{MY}=1]$, 
    ~~~$k=g_b/h_b$, where
    \\
    $g_b=\mfrac{\P(R_{MY}=1\mid X,A)}{\P(R_M=1\mid X,A,Y,R_Y=1)}\times
    \begin{cases}
        q_{2a}(X,A,Y)~\text{if strong completeness}
        \\
        q_{2b}(X,A,Y)~\text{if weak (not strong) completeness}
    \end{cases}$\!\!\!\!.
    \end{tabular}
    &
    \begin{tabular}[t]{@{}l@{}}
        \ref{T5} if strong
        \\
        completeness
    \end{tabular}
    \\[4em]
    Z4
    & 
    $h=\P(R_M=1\mid X,A)q_{1c}(X,A,M)$,~~~$k=1$.
    &
    \ref{T1}
    \\[1em]
    Z5
    & 
    \begin{tabular}[t]{@{}l@{}}
        $h_b=\P(R_Y=1\mid X,A)\left[1+\text{odds}(R_Y=0\mid X,A,M,R_M=1)\mfrac{\P(R_M=1\mid X,A,R_Y=1)}{\P(R_M=1\mid X,A,R_Y=0)}\right]$,
        \\[.3em]
        $k=q_{2c}(X,A,Y)\Big/\left[1+\text{odds}(R_Y=0\mid X,A,M,R_M=1)\mfrac{\P(R_M=1\mid X,A,R_Y=1)}{\P(R_M=1\mid X,A,R_Y=0)}\right]$.
    \end{tabular}
    &
    \ref{T2}
    \\[-.5em]
    \\\hline
    \multicolumn{3}{l}{*~$q_{1a}$, $q_{2a}$, $q_{1b}$, $q_{2b}$, $q_{1c}$, $q_{2c}$ are respectively defined in (\ref{eq:o1a-defn}), (\ref{eq:o2a-defn}), (\ref{eq:o1b-defn}), (\ref{eq:o2b-defn}), (\ref{eq:o1c-defn}), (\ref{eq:o2c-defn}).}
\end{tabular}%
    }
\end{figure*}

\subsubsection{Model Z1}

This model is essentially \citeauthor{zuo2024MediationAnalysisMediator}'s model 2(d) (with the same model assumption), but the graphical model (see ABCD-Z1 in Fig.~\ref{fig:shadow-MY}) is relaxed to allow $M$ and $R_M$ to share unobserved causes. (\cite{li2023SelfcensoringModelMultivariate} consider the same model as \citeauthor{zuo2024MediationAnalysisMediator}'s model, but for multiple outcomes rather than mediator and outcome.)

Model Z1 is relevant to all four temporal orders because 
it assumes that $R_M$ is not influenced by $Y,R_Y$, and $R_Y$ is not influenced by $M,R_M$, thus putting no restriction on the order of $R_M,R_Y$.

In this model, $M$ and $Y$ are both subject to self-connected missingness, and both serve as SVs. 
(This separates Z1 out from the other models in this class, where only one of variables $M,Y$ is subject to self-connected missingness, and the other variable serves as a SV.)

\paragraph*{Assumption Z1:}
\begin{align}
    &R_M\independent(Y,R_Y)\mid X,A,M,\tag{Z1-1}\label{Z1-1}
    \\
    &R_Y\independent(M,R_M)\mid X,A,Y,\tag{Z1-2}\label{Z1-2}
    \\
    &\P(M,Y\mid X,A,R_{MY}=1)~\text{complete in $M$ and in $Y$}.\tag{Z1-3}\label{Z1-3}
\end{align}

\ref{Z1-3} is a combination of two completeness conditions, as $M$ and $Y$ both serve as SVs. Thus, $M$ and $Y$ must each have sufficient variation for the other within $R_{MY}=1$ and levels of $A,X$. We will revisit this in section \ref{sssec:example-MY}.


\paragraph*{Identification.}

With this model, we consider
\begin{align*}
    h=\mfrac{\P(R_M=1\mid X,A)}{\color{teal}\P(R_M=1\mid X,A,M)},~k=\mfrac{\P(R_Y=1\mid X,A,M,R_M=1)}{\color{purple}\P(R_Y=1\mid X,A,Y,R_M=1)}
\end{align*}
($k$ simplified by \ref{Z1-1}). Here $Y$ is a SV with respect to the missingness in $M$ (within levels of $X,A$) , $R_Y$ is conditionally independent of $R_M$, and the relevant completeness condition holds; this justifies an identifying equation of type~(\ref{eq:shadow-red}). $\color{teal}1/\P(R_M=1\mid X,A,M)$ is thus identified by function $q_{1a}(X,A,M)\geq1$ that solves
\begin{align}
    \label{eq:o1a-defn}1/\P(R_M=1&\mid X,A,Y,R_Y=1)=
    \E[q_{1a}(X,A,M)\mid X,A,Y,R_{MY}=1].
\end{align}
Also, $M$ is a SV with respect to the missingness in $Y$ (conditional on $X,A,R_M=1$), with the relevant completeness condition; this justifies an identifying equation of type~(\ref{eq:shadow}). $\color{purple}1/\P(R_Y=1\mid X,A,Y,R_M=1)$ is thus identified by function $q_{2a}(X,A,Y)\geq1$ that solves
\begin{align}
    \label{eq:o2a-defn}1/\P(R_Y=1&\mid X,A,M,R_M=1)=
    \E[q_{2a}(X,A,Y)\mid X,A,M,R_{MY}=1].
\end{align}

\paragraph*{Testable implications.}
This model has two testable implications based on the assumed conditional independence of $R_M$ and $R_Y$:
\begin{align}
    \tag{T4}\label{T4}1/\P(&R_M=1\mid X,A,R_Y=0)=
    \E[q_{1a}(X,A,M)\mid X,A,R_Y=0,R_M=1],
    \\
    \tag{T5}\label{T5}1/\P(&R_Y=1\mid X,A,R_M=0)=
    \E[q_{2a}(X,A,Y)\mid X,A,R_M=0,R_Y=1].
\end{align}

\subsubsection{Models Z2 and Z3}
\label{sssec:models8-9}

Model Z2, based on \citeauthor{zuo2024MediationAnalysisMediator}'s model 2(e), applies to all four temporal order settings. The graphical model (see ABCD-Z2 in Fig.~\ref{fig:shadow-MY}) is relaxed (compared to \citeauthor{zuo2024MediationAnalysisMediator}'s original) to allow $M$ to share unobserved causes with either $R_M$ or $R_Y$ (but not both), hence the two reduced forms. In this model, $M$ -- in addition to $X,A$ -- drives the missingness of both $M$ and $Y$, $R_M$ and $R_Y$ are conditionally independent, and $Y$ serves as a SV. The completeness condition is expanded as either strong or weak completeness could be used.

\paragraph*{Assumption Z2:}
\begin{align}
    \tag{Z2-1}\label{Z2-1}&R_M,R_Y,Y~\text{mutually independent given}~X,A,M,
    \\
    \text{and either}~&\P(M,Y\mid X,\!A,\!R_{MY}\!=\!1)~\text{complete in $Y$}\tag{Z2-2i}\label{Z2-2i}
    \\
    \tag{Z2-2ii}\label{Z2-2ii}\text{or}~&\P(M,Y^\dagger\mid X,\!A,\!R_M\!=\!1)~\text{complete in $Y^\dagger$},
\end{align}
where $Y^\dagger:=(R_Y,R_YY)$.

We add model Z3 for the \textit{unordered} and \textit{reversed} settings (see CD-Z3 in Fig.~\ref{fig:shadow-MY}) with a mirror assumption. In this model, $Y$ -- in addition to $X,A$ -- drives the missingness of both $M$ and $Y$, $R_M$ and $R_Y$ are conditionally independent, and $M$ serves as a SV.

\paragraph*{Assumption Z3:}
\begin{align}
    \tag{Z3-1}\label{Z3-1}&R_M,R_Y,M~\text{mutually independent given}~X,A,Y,
    \\
    \tag{Z3-2i}\label{Z3-2i}\text{and either}~&\P(M,\!Y\!\mid\! X,\!A,\!R_{MY}\!=\!1)~\text{complete in $M$}
    \\
    \tag{Z3-2ii}\label{Z3-2ii}\text{or}~&\P(M^\dagger,\!Y\mid\! X,\!A,\!R_Y\!=\!1)~\text{complete in $M^\dagger$},
\end{align}
where $M^\dagger:=(R_M,R_MM)$.

Note that while Z2 allows $M$ to share unobserved causes with one of the missing indicators, Z3 does not similarly allow $Y$ to share unobserved causes with missing indicators, because that would make $Y$ a collider between the missing indicator and $M$, thus violating \ref{Z3-1}.

\paragraph*{Identification.}

First, consider model Z2.
\ref{Z2-1} identifies $\P(Y\mid X,A,M)$ by complete cases, i.e., $k=1$. The distribution completeness condition is used to identify $\P(M\mid X,A)$. 

(As pointed out by \citeauthor{zuo2024MediationAnalysisMediator}, if $M\independent Y\mid X,A$ (which under \ref{Z2-1} can be verified via condition $M\independent Y\mid X,A,R_{MY}=1$), completeness fails and $\P(M\mid X,A)$ is not identified, but target effects are still identified; this also applies to model Z4. We put this special case aside.)

For the completeness condition, \citeauthor{zuo2024MediationAnalysisMediator} use \ref{Z2-2ii} (\textit{weak completeness}). We add \ref{Z2-2i} (\textit{strong completeness}) as an option.
If strong completeness holds, 
$h=\P(R_M=1\mid X,A)q_{1a}(X,A,M)$
(like in model Z1), with $q_{1a}(X,A,M)$ already defined in (\ref{eq:o1a-defn}).
If strong completeness does not hold but weak completeness does,
$h=\P(R_M=1\mid X,A)q_{1b}(X,A,M)$, 
where $q_{1b}(X,A,M)\geq1$ and solves the type~(\ref{eq:shadow-red-dagger}) equation
\begin{align}
    \label{eq:o1b-defn}1/\P(R_M&=1\mid X,A,Y^\dagger)=\E[q_{1b}(X,A,M)\mid X,A,Y^\dagger,R_M=1].
\end{align}

For model Z3, identification is via tilting functions $h_b,k$ based on $g_b$.
If strong completeness (\ref{Z3-2i}) holds, then
$g_b=\mfrac{\P(R_{MY}=1\mid X,A)}{\P(R_M=1\mid X,A,Y,R_Y=1)}q_{2a}(X,A,Y)$,
with $q_{2a}(X,A,Y)$ already defined in (\ref{eq:o2a-defn}). If strong completeness does not hold but weak completeness (\ref{Z3-2ii}) does, then $g_b=\mfrac{\P(R_{MY}=1\mid X,A)}{\P(R_M=1\mid X,A,Y,R_Y=1)}q_{2b}(X,A,Y)$, where function $q_{2b}(X,A,Y)\geq1$ and solves
\begin{align}
    \label{eq:o2b-defn}1/\P(R_Y=1&\mid X,A,M^\dagger)=
    \E[q_{2b}(X,A,Y)\mid X,A,M^\dagger,R_Y=1].
\end{align}

\paragraph*{Testable implications.} Under strong completeness, models Z2 and Z3 imply the testable conditions \ref{T4} and \ref{T5}, respectively. These testable conditions are not available under weak completeness.

\subsubsection{Models Z4 and Z5}

Model Z4, which applies to the \textit{in-time} and \textit{delayed} settings, is essentially \citeauthor{zuo2024MediationAnalysisMediator}'s model 2(c) (same model assumption), but the graphical model (see BC-Z4 in Fig.~\ref{fig:shadow-MY}) is relaxed to allow $M$ and $R_M$ to share unobserved common causes. In this model (within levels of $X,A$) $M$ drives its missingness and this missingness drives the missingness in $Y$; and $Y$ serves as a shadow of $M$.

We add model Z5, which applies to the \textit{reverse} setting (see E-Z4 in Fig.~\ref{fig:shadow-MY}), with a mirror assumption. In this model (within levels of $X,A$) $Y$ drives its missingness and this missingness drives the missingness in $M$; and $M$ serves as a shadow of $Y$.

\paragraph*{Assumption Z4:}
\begin{align}
    &(R_M,R_Y)\independent Y\mid X,A,M,\tag{Z4-1}\label{Z4-1}
    \\
    &R_Y\independent M\mid X,A,R_M,\tag{Z4-2}
    \\
    &\P(R_Y=1\mid X,A,R_M=0)>0,\tag{Z4-3}
    \\
    &\P(M,Y\mid X,A,R_{MY}=1)~\text{complete in $Y$}.\tag{Z4-4}
\end{align}

\paragraph*{Assumption Z5:}
\begin{align}
    &(R_M,R_Y)\independent M\mid X,A,Y,\tag{Z5-1}\label{Z5-1}
    \\
    &R_M\independent Y\mid X,A,R_Y,\tag{Z5-2}\label{Z5-2}
    \\
    &\P(R_M=1\mid X,A,R_Y=0)>0,\tag{Z5-3}
    \\
    &\P(M,Y\mid X,A,R_{MY}=1)~\text{complete in $M$}.\tag{Z5-4}
\end{align}

\paragraph*{Testable implications.}
Model Z4 (like S5) implies condition \ref{T1}. Model Z5 (like S6) implies \ref{T2}.

\paragraph*{Identification.}

First consider model Z4. Under \ref{Z4-1}, $k=1$. $h=\mfrac{\P(R_M=1\mid X,A)}{\color{teal}\P(R_M=1\mid X,A,M)}$ is identified because with $Y$ being a SV with respect to the missingness of $M$ (conditional on $X,A$), $R_Y$ separated from $(M,Y)$ by $R_M$ and the relevant completeness condition, $\color{teal}1/\P(R_M=1\mid X,A,M)$ is identified by function $q_{1c}(X,A,M)\geq1$ that solves the type~(\ref{eq:shadow-blue}) equation

\begin{align}
    \label{eq:o1c-defn}
    1\!+\!\mathrm{odds}(R_M\!=\!0\mid X,\!A,\!Y,\!R_Y\!=\!1)\mfrac{\P(R_Y\!=\!1\mid X,\!A,\!R_M\!=\!1)}{\P(R_Y\!=\!1\mid X,\!A,\!R_M\!=\!0)}&
    =\E[{\color{teal}q_{1c}(X,A,M)}\mid X,A,Y,R_{MY}=1].
\end{align}

In model Z5 (specifically due to \ref{Z5-1} and \ref{Z5-2}),
\begin{align*}
    h_b=\mfrac{\P(R_Y=1\mid X,A)}{\color{teal}\P(R_Y=1\mid X,A,M)},~k=\mfrac{\color{teal}\P(R_Y=1\mid X,A,M)}{\color{violet}\P(R_Y=1\mid X,A,Y)}.
\end{align*}
Similar reasoning to above identifies $\color{violet}1/\P(R_Y=1\mid X,A,Y)$ by $q_{2c}(X,A,Y)\geq1$ that solves
\begin{align}
    \label{eq:o2c-defn}
    1\!+\!\odds(R_Y\!=\!0\mid X,\!A,\!M,\!R_M\!=\!1)\mfrac{\P(R_M\!=\!1\mid X,\!A,\!R_Y\!=\!1)}{\P(R_M\!=\!1\mid X,\!A,\!R_Y\!=\!0)}&
    =\E[q_{2c}(X,A,Y)\mid X,A,M,R_{MY}=1].
\end{align}
The LHS of (\ref{eq:o2c-defn}) identifies $\color{teal}1/\P(R_Y=1\mid X,A,M)$ -- by the odds-tilting technique in (\ref{eq:odds-tilt1}).

\subsubsection{Comments on this subclass of models}
\label{sssec:discussion-shadow-MY}

This class provides additional model options: three for the \textit{in-time} and \textit{delayed} settings (Z1, Z2, Z4), three for the \textit{unordered} setting (Z1, Z2, Z3), and four for the \textit{reverse} setting (Z1, Z2, Z3, Z5). 

The reduced forms of these models show a feature shared with the self-separated models: only two connections are present out of the five possible connections for the missingness model. And because these are self-connected missingness models, one (or both) of these connections is a self-connection. This reveals that these models are not relaxations of self-separation. Rather, to allow self-connection, these models restrict some other connections to not exist.

Among those restrictions, one (shared by self-separated models) is that a missing indicator cannot be influenced by both $M$ and $Y$. Earlier that was because the missingness of these variables is self-separated. Here this is to allow one or both variables to serve as SVs.

Other restrictions to be noted concern the connection between $R_M$ and $R_Y$. Importantly, models Z1, Z2, Z3 all require that connection to not exist. In addition, even though Z4 and Z5 have $R_M$ and $R_Y$ causally connected, they do not allow these two variables to share unobserved causes (because such unobserved causes would turn one of these two variables into a collider, breaking a required conditional independence). As mentioned in section~\ref{ssec:discussion-separated}, this can be unrealistic in applications.

In brief, the models just discussed and the self-separated models from section~\ref{sec:separated} represent trade-offs where allowing certain features requires disallowing others. It would be ideal if we could do away with some of those trade-offs. That may be possible if we have additional information from elsewhere, e.g., in the form of SVs that are not $M,Y$. This will be our focus in  sections \ref{ssec:shadow-theory2} to \ref{ssec:shadow-missing}.

\subsubsection{Relating to the school intervention example}
\label{sssec:example-MY}

For this class of models, we focus on this example, as \cite{zuo2024MediationAnalysisMediator} have discussed the other example extensively. As mentioned earlier, a concern is that students who drink may be hesitant to report on drinking, so we need to allow the $Y\to R_Y$ path. In this \textit{in-time} setting, this narrows down to model Z1 (top of Fig.~\ref{fig:shadow-MY}) as the only option. Conveniently, this model also allows self-connected mediator missingness -- due to either direct influence or shared causes -- so this element requires no further consideration.

One restriction of Z1 is that $R_Y$ not share unobserved causes with $R_M$, $M$ and $Y$. (In fact, the whole model class disallows unobserved causes shared by $R_Y$ and either $R_M$ or $Y$.) As the missing indicators $R_M$ and $R_Y$ are months apart, their most important common causes may be stable factors that influence school attendance or absence, such as student mental health, family socio-economic status, family functioning, etc., and perhaps a baseline measure of absenteeism if available. Thus, these variables should be included in covariate set $X$. Common causes of $M$ and $R_Y$ may be a different but likely overlapping set of variables, as they influence the mediator (student attitude and parental rules) in addition to $R_Y$; and common causes of $Y$ and $R_Y$ may be yet another set. Hence an important takeaway is that studies should collect data on things that likely influence missingness of key variables. This has long been considered good practice to help the handling of missing data. What it achieves precisely is to block certain undesirable back-door paths.

A disclaimer: There may be cases where it is determined that $Y$ and $R_Y$ are likely both influenced by certain factors that occur later in the study, and those may be on the causal pathways from $X$, $A$ or $M$ to $Y$. Similarly, $M$ (or $R_M$) and $R_Y$ may share causes that are on the causal pathways from $X$ or $A$ to $M$. Recall from section \ref{subsec:saturated} that these cases fall outside the scope of models considered in this paper. Such causes, especially those that are post-treatment, require separate investigation.

The other restriction of Z1 is that $R_M$ does not influence $R_Y$ and $M$ does not influence $R_Y$ directly. While the former is defensible (the two missing indicators are more likely associated due to common causes), the latter requires careful consideration. If students who drink are more likely to skip the questions about drinking especially if they had previously reported healthy attitudes and strong self-control (i.e., $M$ and $Y$ interact in influencing $R_Y$) then the no $M\to R_Y$ assumption does not hold.

Note that we have focused on model Z1 because it is the only option with self-connected outcome missingness for the in-time setting. If the study asked sensitive questions via an electronic device in privacy, there may still be missingness but there may be less concern that it is driven by the variable itself. That would open up two model options, Z2 and Z4, each with some other flexibility.

Once a model is picked based on the above reasoning about conditional independences, its completeness condition needs to be considered. Meaningful guidance on practical consideration of completeness can be found elsewhere \citep[appendix 2 of][]{miao2023IdentifyingEffectsMultiple,ringlein2025DemystifyingProximalCausal}. Here we just note a rough minimum condition to be considered in practice: the SV should have the same or higher dimensionality than the variable of interest. For model Z1, because $M$ and $Y$ are both SVs for each other, this means $M$ and $Y$ need to have the same dimensionality. If these are categorical variables, they need to have the same number of categories. If $M$ is, say, student attitudes on a Likert scale with three levels, a binary drinking $Y$ does not have sufficient variation, so one could expand to a three-level measurement (e.g., never, occasionally, and frequently drink) if available. If one variable is continuous, the other needs to be continuous. If $Y$ is continuous but $M$ is categorical, $Y$ may have sufficient variation relative to $M$ but $M$ does not have sufficient variation relative to $Y$ and completeness fails.

\subsection{A specific theory: External shadow for a pair of variables}
\label{ssec:shadow-theory2}

The next sections will consider models with SVs that are downstream of $M,Y$ in the causal structure, or on the $M$ to $Y$ causal pathway, or upstream of $M,Y$. To help streamline the presentation of those models, we first extend the theory of section~\ref{ssec:shadow-theory} (which concerns self-connected missingness of a single variable) to address the case where two variables of interest, $V_1,V_2$, are subject to self-connected missingness, and a third variable $Z$ seves as the SV. We refer to $Z$ in this case as an \textit{external SV}. Let $R_1,R_2$ denote the missing indicators and let $R_{12}:=R_1R_2$. Again, we keep any conditioning on covariates implicit. The goal is to identify $\P(R_{12}=1\mid V_1,V_2)$. 

Here we treat $Z$ as fully observed, but missingness in $Z$ that satisfies an assumption similar to mSV can be accommodated (a topic we will revisit later). Also, we talk about $Z$ as a single variable here, but $Z$ can be multivariate. (Using multivariate $Z$ may help with the completeness conditions, but may exacerbate SV missingness.) 

We seek SV models that minimize restrictions on the model for $(R_1,R_2)$, aka \textit{the core missingness model}. As the use of SVs is to deal with self-connected missingness, we will insist on allowing $M$ and $Y$ to influence their own missingness. Restriction on other elements of the core missingness model is allowed if needed to make a SV model work, but is kept to a minimum.

\subsubsection{Three external SV models}

Recall that the recipe for SV-based identification includes a SV independence assumption and a completeness condition. We find three SV models (see Table~\ref{tab:xSV-models}), which we label the \textit{combined}, \textit{parallel} and \textit{sequential} external SV models.
These models have similar SV independence assumptions but different completeness conditions. They therefore differ in terms of their \textit{explicit} restrictions on the core missingness model -- see last column of the table. (There may be \textit{implicit} restrictions required for the SV independence assumption to hold, which we will address later when applying these models to the mediation context with different SV types and different temporal order settings.)

\begin{table*}[]
    \caption{Three external shadow variable ($Z$) models for a pair of variables ($V_1,V_2$) with self-connected missingness ($R_1,R_2$)}
    \label{tab:xSV-models}
    \centering
    \resizebox{\linewidth}{!}{%
    \begin{tabular}{llll}
        & SV independence & completeness & \begin{tabular}{@{}l@{}}
            explicit restriction \\ on core model
        \end{tabular}
        \\\hline
        Model 1
        & $Z\independent(R_1,R_2)\mid V_1,V_2$ 
        &  $\P(V_1,V_2,Z\mid R_{12}=1)$ complete in $Z$ 
        & none
        \\
        (\textit{combined}) & {\footnotesize(weaker version: $Z\independent R_{12}\mid V_1,V_2$)}
        \\\hline
        Model 2
        & $Z\independent(R_1,R_2)\mid V_1,V_2$ 
        & $\P(V_1,Z\mid V_2,R_{12}=1)$ complete in $Z$,
        & $R_1\independent R_2\mid V_1,V_2$
        \\
        (\textit{parallel}) & {\footnotesize(weaker version: {\scriptsize$\begin{matrix}
            Z\independent R_1\mid V_1,V_2,R_2=1\\[-.2em]
            Z\independent R_2\mid V_1,V_2,R_1=1
        \end{matrix}$})} 
        & $\P(V_2,Z\mid V_1,R_{12}=1)$ complete in $Z$
        \\\hline
        Model 3
        & $Z\independent(R_1,R_2)\mid V_1,V_2$
        & $\P(V_1,Z\mid R_1=1)$ complete in $Z$,
        & $R_1\independent V_2\mid V_1$
        \\
        (\textit{sequential}) & {\footnotesize(weaker version: {\scriptsize$\begin{matrix}
            Z\independent R_1\mid V_1,V_2\\[-.2em]
            Z\independent R_2\mid V_1,V_2,R_1=1
        \end{matrix}$})} 
        & $\P(V_2,Z\mid V_1,R_{12}=1)$ complete in $Z$
        \\\hline
    \end{tabular}%
    }
\end{table*}

The technical SV independence assumptions required (labeled \textit{weaker version} in Table~\ref{tab:xSV-models}) vary across these models, but are all implied by the assumption that $Z\independent(R_1,R_2)\mid V_1,V_2$. As this assumption is simpler (and perhaps easier to consider in applications), we put aside the varying weaker versions here. Instead, we focus on the varying completeness conditions, which imply variation in the amount of information contained in the SV.

\subsubsection{Identification under these external SV models}

\paragraph*{Model 1 (combined external SV model).} 
This model is in a sense the simplest: it helps identify the combined missingness probability $\P(R_{12}=1\mid V_1,V_2)$ without factorization and does not require any restriction on the core missingness model. We treat $R_{12}$ as the response indicator of interest and treat $(V_1,V_2)$ combined as the ``variable'' with missingness. Then $Z$ is a SV with the relevant completeness condition, so equality (\ref{eq:shadow}) holds and serves as the identifying equation, which here takes the form
\begin{align}
    \label{eq:extSV1}\odds(R_{12}&=0\mid Z)=
    \E[{\color{purple}\odds(R_{12}=0\mid V_1,V_2)}\mid Z,R_{12}=1].
\end{align}
What allows this simplicity is a strong completeness condition, which requires that there is as much variation in $Z$ as there is in both of $V_1,V_2$ combined (rather than just $V_1$ and/or just $V_2$).

\paragraph*{Model 2 (parallel external SV model).} 
This model places a restriction on the core missingness model: the missingness in $M$ and the missingness in $Y$ are conditionally independent. This gives a symmetric factorization of the response probability:
\begin{align*}
    \P(R_{12}=1\mid V_1,V_2)=
    \P(R_1=1\mid V_1,V_2,R_2=1)\P(R_2=1\mid V_1,V_2,R_1=1).
\end{align*}
This model handles the missingness in the two variables separately, in parallel -- leveraging $Z$ both as a SV with respect to the missingness in $V_1$ (conditional on $V_2,R_2=1$) and as a SV with respect to the missingness in $V_2$ (conditional $V_1,R_1=1$). The completeness condition required is weaker than that in model 1, with two parts: there is as much variation in $Z$ as there is in $V_1$ within levels of $V_2$, and as much variation in $Z$ as in $V_2$ within levels of $V_1$. Then the relevant missingness odds are identified as solutions to the two equations:
\begin{align}
    \label{eq:extSV2-V}\odds(R_1=0\mid Z,V_2,R_2=1)
    &=\E[{\color{purple}\odds(R_1=0\mid V_1,V_2,R_2=1)}\mid Z,V_2,R_{12}=1],
    \\
    \label{eq:extSV2-W}\odds(R_2=0\mid Z,V_1,R_1=1)
    &=\E[{\color{violet}\odds(R_2=0\mid V_1,V_2,R_1=1)}\mid Z,V_1,R_{12}=1].
\end{align}

\paragraph*{Model 3 (sequential external SV model).} 
This model places a different restriction of the core missingness model: $R_1$ and $V_2$ are conditionally independent. This factors the response probability as
\begin{align*}
    \P(R_{12}=1\mid V_1,V_2)=\P(R_1=1\mid V_1)\P(R_2=1\mid V_1,V_2,R_1=1).
\end{align*}
The second probability is identified the same way as in model 2, using the same $V_2$-related completeness condition. The first probability does not condition on $V_2$, so the $V_1$-related completeness condition of model 3 is weaker than that in model 2. It only requires that there is as much variation in $Z$ as in $V_1$ marginally, not within levels of $V_2$. Under this model, the identifying equations are:
\begin{align}
    \odds(R_1=0\mid Z)
    &=\E[{\color{purple}\odds(R_1=0\mid V_1)}\mid Z,R_1=1],
    \\
    \tag{\ref{eq:extSV2-W}}\odds(R_2=0\mid Z,V_1,R_1=1)
    &=
    \E[{\color{violet}\odds(R_2=0\mid V_1,V_2,R_1=1)}\mid Z,V_1,R_{12}=1].
\end{align}

\subsubsection{Testable implication (of model 2)}

The $R_1$-$R_2$ conditional independence in model 2 implies
\begin{align}
    \mfrac{\P(R_1=0,R_2=0)}{\P(R_1=1,R_2=1)}
    &=\E\left[\text{o}_1(V_1,V_2)\,\text{o}_2(V_1,V_2)\mid R_{12}=1\right],
\end{align}
where functions $\text{o}_1(V_1,V_2)$ and $\text{o}_2(V_1,V_2)$ are solutions for $\color{purple}\odds(R_1=0\mid V_1,V_2,R_2=1)$ and $\color{violet}\odds(R_2=0\mid V_1,V_2,R_1=1)$ based on (\ref{eq:extSV2-V}) and (\ref{eq:extSV2-W}).
This testable condition is in the same spirit as \ref{T3} from model S3.

\subsection{Models with auxiliary variables as shadow variables}

We can now use the theory just developed to explore models with SVs that are not $M,Y$. The options are: shadow variables being covariates (and possibly treatment) and shadow variables being auxiliary variables (not $X,A,M,Y$). It is easier to discuss the latter first.

There is a range of auxiliary SV models, where the SV is downstream or upstream of $M,Y$ in the causal structure, or on the $M$-to-$Y$ path. As model assumptions are shared across SV locations, we state them upfront. 

Denote the auxiliary SV by $Z$. There are four different assumptions A1-A4 below (A for \textit{auxiliary}). A1 and A2 are straightforward applications of the combined and parallel external SV models discussed above. A3 and A4 are applications of the sequential external SV model with $(V_1,V_2):=(M,Y)$ and $(V_1,V_2):=(Y,M)$, respectively.

\paragraph*{Assumption A1~~(combined auxiliary SV)}
\begin{align}
    &Z\independent(R_M,R_Y)\mid X,A,M,Y,\tag{Ax-1}\label{A-1}
    \\
    &\P(M,Y,Z\mid X,A,R_{MY}=1)~\text{complete in $Z$}.\tag{A1-2}\label{A1-2}
\end{align}

\paragraph*{Assumption A2~~(parallel auxiliary SV)}
\begin{align}
    & Z\independent(R_M,R_Y)\mid X,A,M,Y,\tag{\ref{A-1}}
    \\
    &\P(M,Z\mid X,A,Y,R_{MY}=1)~\text{complete in $Z$},\tag{A2-2a}\label{A2-2a}
    \\
    &\P(Y,Z\mid X,A,M,R_{MY}=1)~\text{complete in $Z$},\tag{A2-2b}\label{A2-2b}
    \\
    &R_M\independent R_Y\mid X,A,M,Y.\tag{A2-3}\label{A2-0}
\end{align}

\paragraph*{Assumption A3~~(\textsc{my}-sequential auxiliary SV)}
\begin{align}
    &Z\independent (R_M,R_Y)\mid X,A,M,Y,\tag{\ref{A-1}}
    \\
    &\P(M,Z\mid X,A,R_M=1)~\text{complete in $Z$},\tag{A3-2a}\label{A3-2a}
    \\
    &\P(Y,Z\mid X,A,M,R_{MY}=1)~\text{complete in $Z$},\tag{A3-2b}\label{A3-2b}
    \\
    &R_M\independent Y\mid X,A,M.\tag{A3-3}\label{A3-0}
\end{align}

\paragraph*{Assumption A4~~(\textsc{ym}-sequential auxiliary SV)}
\begin{align}
    &Z\independent (R_M,R_Y)\mid X,A,M,Y,\tag{\ref{A-1}}
    \\
    &\P(Y,Z\mid X,A,R_Y=1)~\text{complete in $Z$},\tag{A4-2a}\label{A4-2a}
    \\
    &\P(M,Z\mid X,A,Y,R_{MY}=1)~\text{complete in $Z$},\tag{A4-2b}\label{A4-2b}
    \\
    &R_Y\independent M\mid X,A,Y.\tag{A4-3}\label{A4-0}
\end{align}

We leave out the weaker versions of the SV independence assumption to simplify presentation, so A1-A4 here all share the same first component, \ref{A-1}. The second component is the completeness condition, and the third component is the explicit restriction on the core missingness model.
We treat $Z$ as fully observed here, and will address missingness in $Z$ in section~\ref{ssec:shadow-missing}.

\begin{figure*}[t!]
    \caption{Reduced-form DAGs of models with auxiliary SVs. For completeness conditions, see text.}
    \label{fig:shadow-auxiliary}
    \centering
    \includegraphics[width=.85\textwidth, page=8]{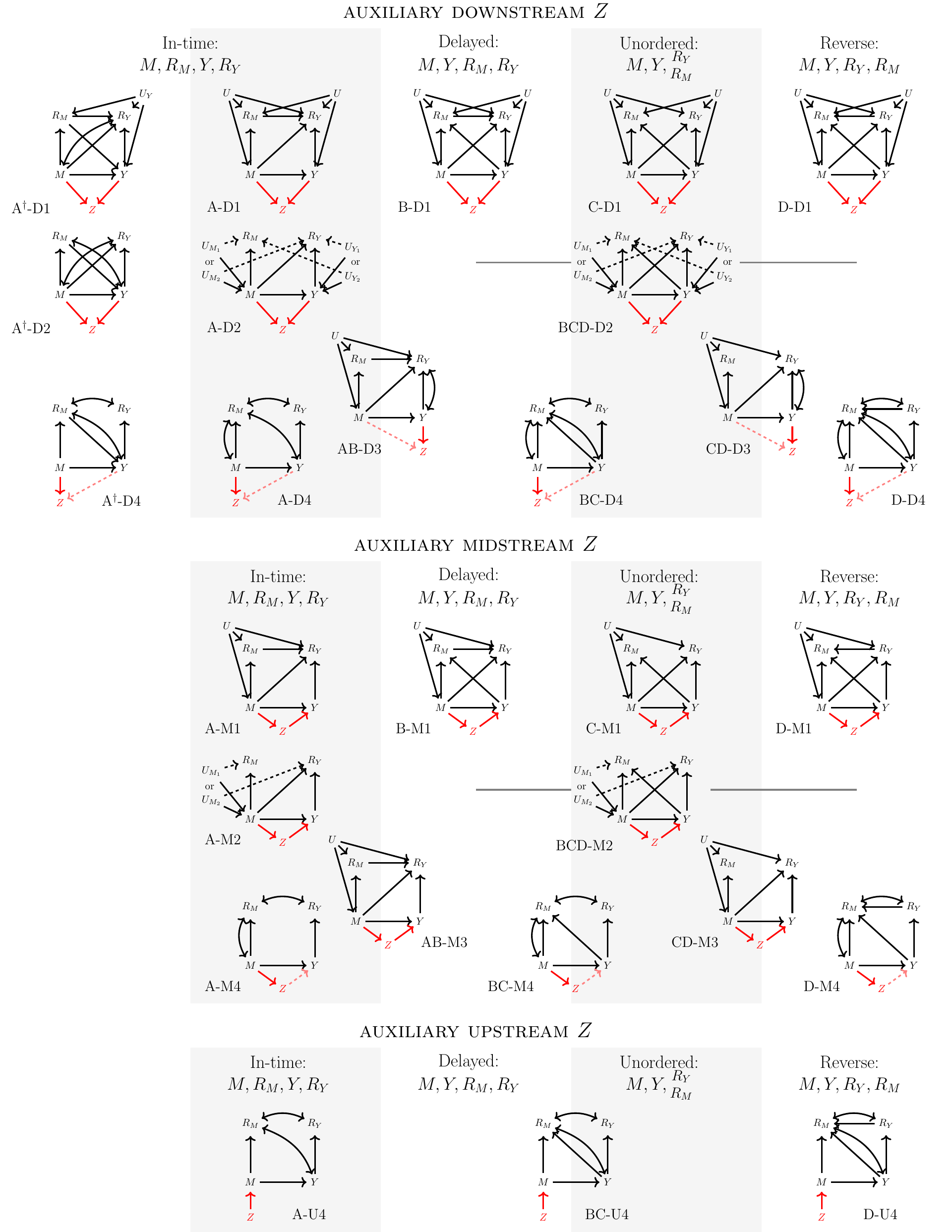}
\end{figure*}

For identification results under these assumptions, see Table~\ref{tab:shadow-auxiliary}.
Assumption A2 implies the testable condition:
\begin{align}
    \tag{T6}\label{T6}\mfrac{\P(R_M=0,R_Y=0\mid X,A)}{\P(R_M=1,R_Y=1\mid X,A)}=\E\Big\{[q_{1e}(X,A,M,Y)\!-\!1]\cdot[q_{2e}(X,A,M,Y)\!-\!1]\mid X,A,R_{MY}=1\Big\},
\end{align}
where functions $q_{1e}$ and $q_{2e}$ are defined in the second row of Table~\ref{tab:shadow-auxiliary}. $(q_{1e}-1)$ and $(q_{2e}-1)$ respectively identify $\odds(R_M=0\mid X,A,M,Y,R_Y=1)$ and $\odds(R_Y=0\mid X,A,M,Y,R_M=1)$.

Before using A1-A4 with different types of auxiliary SVs, it helps to retain a conceptual distinction of them. Here we keep implicit the conditioning on $X,A$ in all these models and zoom in on where there is conditioning on $M,R_M$ or $Y,R_Y$. The distinction is: the \textit{combined} SV model A1 requires $Z$ to be the SV for $M,Y$ combined (unconditional); the \textit{parallel} SV model A2 requires $Z$ to be the SV for $M$ within levels of observed $Y$ and for $Y$ within levels of observed $M$ (both conditional); the \textit{sequential} SV  model A3 requires $Z$ to first be the SV for $M$ agnostic of $Y$ and then to be the SV for $Y$ within levels of observed $M$ (unconditional then conditional); and the other \textit{sequential} model A4 requires $Z$ to first be the SV for $Y$ agnostic of $M$ and then to be the SV for $M$ within levels of observed $Y$ (unconditional then conditional).

\subsubsection{Models with auxiliary downstream \texorpdfstring{$Z$}{Z}}
\label{ssec:shadow-downstream}

We label models that satisfy A1-A4 where $Z$ is downstream of $M,Y$ in the causal structure D1-D4 (D for \textit{downstream}). The relevant DAGs specific to the different temporal order settings are shown in the top panel of Fig.~\ref{fig:shadow-auxiliary}.

\begin{table*}[t!]
    \centering
    \caption{Identification results and testable implications of models with auxiliary shadow variables.}
    \label{tab:shadow-auxiliary}
    \resizebox{\linewidth}{!}{%
\begin{tabular}[b]{@{}cl@{}c}
    Model & Tilting functions & \begin{tabular}{@{}c@{}}
        Testable\\implications
    \end{tabular}
    \\\hline
    \\[-.5em]
    A1
    & $h_b=\E[g_b\mid X,A,M,R_{MY}=1]$,~~$k=g_b/h_b$,~~where
    \\[.3em]
    (D1, M1)
    &~~~$g_b=\P(R_{MY}=1\mid X,A)\,q_{12d}(X,A,M,Y)$,
    \\[.3em]
    &~~~$q_{12d}(X,A,M,Y)\geq1$ and solves
    \\[.3em]
    &~~~~~~~$[\P(R_{MY}=1\mid X,A,Z)]^{-1}=\E\left[q_{12d}(X,A,M,Y)\mid X,A,Z,R_{MY}=1\right]$.
    \\[1em]
    A2 
    & $h_b=\E[g_b\mid X,A,M,R_{MY}=1]$,~~$k=g_b/h_b$,~~where
    & \ref{T6}
    \\[.3em]
    (D2, M2)
    &~~~$g_b=\P(R_{MY}=1\mid X,A)\,q_{1e}(X,A,M,Y)q_{2e}(X,A,M,Y)$, 
    \\[.3em]
    &~~~$q_{1e}(X,A,M,Y)\geq1$ and $q_{2e}(X,A,M,Y)\geq1$ and respectively solve
    \\[.3em]
    &~~~~~~~$[\P(R_M=1\mid X,A,Y,Z,R_Y=1)]^{-1}=\E\left[q_{1e}(X,A,M,Y)\mid X,A,Y,Z,R_{MY}=1\right]$ and
    \\[.3em]
    &~~~~~~~$[\P(R_Y=1\mid X,A,M,Z,R_M=1)]^{-1}=\E\left[q_{2e}(X,A,M,Y)\mid X,A,M,Z,R_{MY}=1\right]$.
    \\[1em]
    A3 
    & $h=\E[g\mid X,A,M,R_{MY}=1]$,~~$k=g/h$,~~where
    \\[.3em]
    (D3, M3)
    &~~~$g=\P(R_M=1\mid X,A)q_{1f}(X,A,M)\P(R_Y=1\mid X,A,M,R_M=1)q_{2e}(X,A,M,Y)$,
    \\[.3em]
    &~~~$q_{1f}(X,A,M)\geq1$ and solves
    \\[.3em]
    &~~~~~~~$[\P(R_M=1\mid X,A,Z)]^{-1}=\E[q_{1f}(X,A,M)\mid X,A,Z,R_M=1]$
    \\[.3em]
    &~~~and $q_{2e}$ is defined as above.
    \\[1em]
    A4 
    & $h_b=\E[g_b\mid X,A,M,R_{MY}=1]$,~~$k=g_b/h_b$,~~where
    \\[.3em]
    (D4, M4, U4)
    &~~~$g=\P(R_{MY}=1\mid X,A)q_{1e}(X,A,M,Y)q_{2f}(X,A,Y)$,
    \\[.3em]
    &~~~$q_{1e}$ is defined as above,
    \\[.3em]
    &~~~and $q_{2f}(X,A,Y)\geq1$ and solves
    \\[.3em]
    &~~~~~~~$[\P(R_Y=1\mid X,A,Z)]^{-1}=\E[q_{2f}(X,A,Y)\mid X,A,Z,R_Y=1]$.
    \\[-.7em]
    \\\hline
    \multicolumn{3}{l}{\small Notes: $q$ functions identify inverse response probabilities. Specifically, $q_{12d}=[\P(R_{MY}=1\mid X,A,M,Y)]^{-1}$ under A1;}
    \\
    \multicolumn{3}{l}{\small~~~~~~~~~~~$q_{1e}=[\P(R_M=1\mid X,A,M,Y,R_Y=1)]^{-1}$ under A2/A4; $q_{2e}=[\P(R_Y=1\mid X,A,M,Y,R_M=1)]^{-1}$ under A2/A3;}
    \\
    \multicolumn{3}{l}{\small~~~~~~~~~~~$q_{1f}=[\P(R_M=1\mid X,A,M)]^{-1}$ under A3; $q_{2f}=[\P(R_Y=1\mid X,A,Y)]^{-1}$ under A4.}
\end{tabular}%
    }
\end{table*}

\paragraph*{Model D1.}

This model places no restriction on the core missingness model. The DAGs A$^\dagger$-D1 to D-D1 are our starting models from Fig.~\ref{fig:saturated} augmented with the auxiliary downstream SV. Here we split the \textit{in-time} setting model into two versions, one without and one with the $R_M\to Y$ path (labeled A-D1 and A$^\dagger$-D1) to show how the model looks different depending on this path. A$^\dagger$-D1 is the first option so far for the case where $R_M$ influences $Y$.


\paragraph*{Model D2.}

The $R_M\independent R_Y\mid X,A,M,Y$ assumption of this model implies several things about the core missingness model. Across all settings, it means that $R_M$ and $R_Y$ do not influence each other and do not share unobserved common causes; and that $M$ and $Y$ each can share unobserved causes with either $R_M$ or $R_Y$ but not with both. For the case where $R_M$ influences $Y$ (see A$^\dagger$-D2) specifically, this assumption also means that $Y$ can influence $R_Y$ but does not share unobserved causes with $R_Y$.


\paragraph*{Models D3 and D4.}

Each of these models places a restriction on the core missingness model.
In D3, the $R_M\independent Y\mid X,A,M$ assumption rules out any direct influence and unobserved common causes between $R_M$ and $Y$; this means there is no A$^\dagger$-D3 model. This assumption also means that for the \textit{reverse setting} $R_Y$ is not allowed to influence $R_M$.
In D4, the $R_Y\independent M\mid X,A,Y$ assumption rules out any direct influence and unobserved common causes between $R_Y$ and $M$; and for the \textit{in-time} and \textit{delayed} settings disallows causal influence of $R_M$ on $R_Y$. 

Unlike models D1 and D2 (which require the dependence of $Z$ with $M$ and with $Y$ through the $M\to Z$ and $Y\to Z$ paths to be both strong), there is asymmetry with respect to these two paths in models D3 and D4. Model D3 can work with a weak or even null $M\to Z$ path, as long as the $M\to Y\to Z$ path is strong enough to satisfy the $M$-related completeness condition. Similarly, model D4 can work with a weak or null $Y\to Z$ path. (To signal this asymmetry, for both of these models, we indicate the less essential path by a dashed arrow.)

Related to this last point, \cite{zuo2024MediationAnalysisMediator} consider a model where the SV is a future outcome influenced by $Y$ but not directly by $M$ (see their appendix). That model is a special case of AB-D3 where the $M\to Z$ path is absent (and there are also no unobserved common causes of any variable pair). We note here that there is no need for $Z$ to be conditionally independent of $M$ as depicted in that model (and that AB-D3 accommodates certain unobserved common causes among the variables).
Also, D3 (and M3, C3, introduced shortly) minus unobserved common causes coincides with a DAG considered by \cite{ma2003IdentificationGraphicalModels} in a setting with binary variables.

\subsubsection{Models with auxiliary midstream \texorpdfstring{$Z$}{Z}}
\label{ssec:shadow-midstream}

We label models that satisfy A1-A4 where $Z$ in on the $M$-to-$Y$ causal path M1-M4 (M for \textit{midstream}). The relevant DAGs specific to the different temporal order settings are shown in the middle panel of Fig.~\ref{fig:shadow-auxiliary}.

A key difference between models M1-M4 and the downstream SV models above is that with M1-M4, \textit{there are no arrows into $Y$ in the core missingness model}. This is an \textit{implicit} restriction required for the SV independence assumption \ref{A-1} to hold in the current context where $Z$ emits an arrow into $Y$. This means there are no A$^\dagger$- models in this subclass, and $Y$ does not share unobserved causes with the missing indicators.

\subsubsection{Model with auxiliary upstream \texorpdfstring{$Z$}{Z}}
\label{ssec:shadow-upstream}

An auxiliary upstream SV cannot be a direct cause of both $M$ and $Y$, otherwise it would belong in $(X,A)$ and thus not be auxiliary. Also, an auxiliary variable that is a direct cause of $Y$ but not $M$ would not be useful as a SV for the missingness in $M$. Therefore we have only one model with auxiliary upstream $Z$, U4 (U for \textit{upstream}) satisfying assumption A4 (see the bottom panel of Fig.~\ref{fig:shadow-auxiliary}), where $Z$ is a direct cause of $M$ but not $Y$. ($Z$ can be influenced by $X,A$.) With $Z$ emitting an arrow into $M$, U4 places an implicit no-arrows-into-$M$ restriction on the core missingness model.

\subsubsection{When the auxiliary \texorpdfstring{$Z$}{Z} shares unobserved causes with \texorpdfstring{$M,Y$}{M, Y}}
\label{ssec:shadow-U}

In the models presented above, the auxiliary SV is connected to $M$ and $Y$ via direct causal influence only. We should also consider cases where $Z$ shares unobserved causes with $M,Y$. In these cases, even though model assumptions stay the same, the models have more implicit restrictions on the core missingness model. Specifically, if $Z$ shares unobserved causes with $Y$, the implicit restriction is that there are no arrows into $Y$ in the core missingness model. If $Z$ shares unobserved causes with $M$, the implicit restriction is that there are no arrows into $M$ in the core missingness model.

\subsubsection{Relating to the examples}

When self-separated models and models with mediator or outcome as SVs are ruled out, we ask whether there are auxiliary variables somewhere that might serve as SVs. It may help to think about SVs for each variable with self-connected missingness separately. Candidate SVs should be (i) connected to that variable (and the stronger they reflect that variable, the better),%
\footnote{Related to this point, strategies for selecting proxies in proximal causal inference may be useful \citep[see e.g.,][]{ringlein2025DemystifyingProximalCausal}.}
 but (ii) separate from the missingness.

In a setting like the job training example, where participants are recruited and followed and they are (we assume) the only source of data, we look for SVs from the same and any other of the study's data collection waves. If there are additional follow-up waves beyond 48 months, those might contain candidate downstream SVs for the outcome (earnings), such as earnings at a later time or measures of wealth and expenditure (home ownership or vacations). These same SVs may also work for the mediator (education at 30 months), and education in these later waves may also be used.
A convenience in using downstream SVs is that there is no worry about them influencing $R_Y$ and $R_M$, but a challenge is that they are often subject to a higher level of missingness due to attrition over time (see section~\ref{ssec:shadow-missing} on SV missingness). Any data collection between 30 and 48 months may contain mid-stream SVs, for example education and/or earnings at such in-between time, provided one can argue that they do not influence the missingness.
If there are just the two follow-up waves at 30 and 48 months (when the mediator and outcome are measured), then concurrent data can also be used as SVs, although the causal structure needs to be very carefully considered. Regarding baseline data, again, the key is to find variables that likely do not influence $R_M$ and $R_Y$ but carry information about $M$ and $Y$. In some settings, especially where follow-ups are close in time to baseline, there may be baseline measures of the mediator and/or outcome that might fit this bill. (Note that if these variables are already included in $X$, then they would be covariate SVs instead of auxiliary SVs -- see section~\ref{ssec:shadow-covariate}.) Our job training example is a different case without such baseline measures. Yet one might consider the amount of schooling attained before dropout as a candidate SV.

The school intervention setting is different as it likely has less of an attrition problem, compared to the previous example. Also, in addition to data collected for the study, investigators may seek access to routine school records -- under proper human subjects protection protocols. Grades and teacher ratings of student performance and behavior (before, during and after the study) may be available and useful as auxiliary SVs. This example is also interesting in that there are data from both students and parents; this is a case where measures of the same construct from different sources can serve as SVs for each other. For example, if we are interested specifically in parent-reported parental rules as a mediator (this was a target of the intervention), one strategy is to use student-reported parental rules as an auxiliary mid-stream SV. Similarly, if there is a survey of parents at the end of the study about their children's attitudes and behaviors, those variables may serve as SVs for the student-reported outcome (drinking). And of course, just like in the other example, any data provided by the students about essentially the same variables but at other time points (e.g., baseline attitudes and drinking behavior), or about related behaviors and attitudes at the same time points, can potentially serve as SVs.

Once certain variables have been tentatively chosen as SVs based on the kind of reasoning above, we need to consider whether they fit with one of the auxiliary SV models in Fig.~\ref{fig:shadow-auxiliary} (or Fig.~\ref{fig:shadow-covariate} if the SV is part of $X$ -- see section~\ref{ssec:shadow-covariate}). For example, with a tentative midstream SV, the question is whether it fits models M1 (combined SV), M2 (parallel SV), or M3 or M4 (sequential SV); for the \textit{in-time} setting, see the DAGs in Fig.~\ref{fig:shadow-auxiliary}, middle panel, first column. The plausibility of these models' conditional independence restrictions needs to be evaluated, in the same way we did with earlier models in sections~\ref{ssec:examples-separated} and \ref{sssec:example-MY}.

In addition, the completeness conditions required by these models need to be judged. For example, M1 has the strongest completeness condition (there is as much variation in $Z$ as in the combination of $M,Y$) to compensate for its lighter conditional independence restrictions; to satisfy this completeness may require a set of several auxiliary variables to serve as the SV. The other models, M2-M4, have weaker completeness conditions, but are more restrictive in terms of conditional independences.


\subsection{Models with covariates serving as shadow variables}
\label{ssec:shadow-covariate}

Now we turn to the last model subclass where the SV is contained in the $\{X,A\}$ set. We refer to $Z$ in this case as a \textit{covariate SV} (because it needs to include elements of $X$ and cannot be just $A$, to carry enough information), and label these models C1-C4 (C for \textit{covariate}). 
Denote the remainder of $\{X,A\}$ by $W$. That is, $\{X,A\}$ and $\{W,Z\}$ are the same set of variables. 
The assumptions of these models remind of A1-A4, but while A1-A4 condition $X,A$, C1-C4 instead condition on $W$.

\begin{figure*}[hp!]
    \caption{DAGs of models where baseline covariates $Z$ serve as shadow variables (shown conditional on $W:=(X,A)\setminus Z$). For completeness conditions, see model assumptions in text.}
    \label{fig:shadow-covariate}
    \centering
    \includegraphics[width=.7\textwidth, page=9]{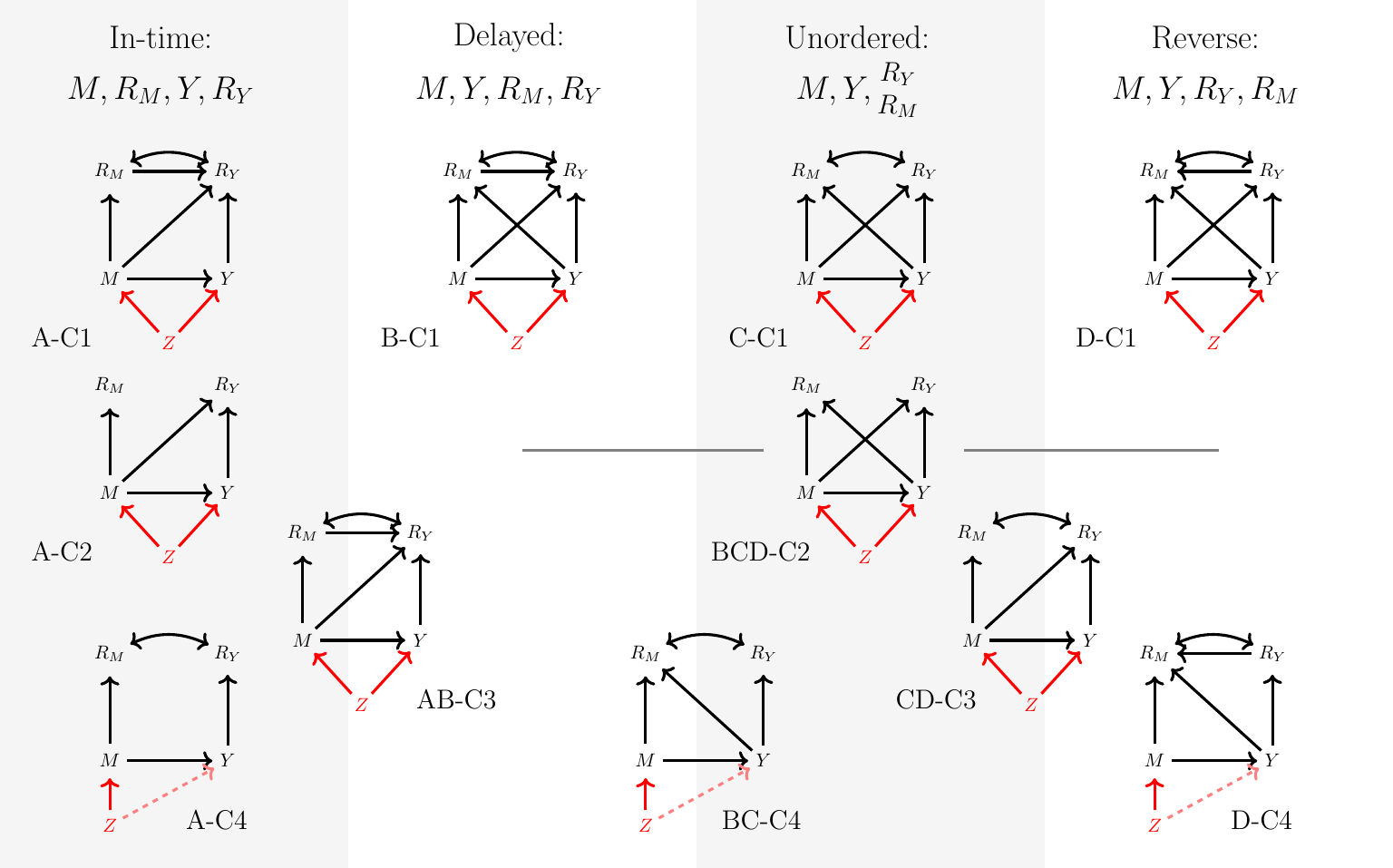}

    \bigskip
    
    \captionof{table}{Identification results and testable implications of models with covariate SVs}
    \label{tab:shadow-covariate}
    \resizebox{\linewidth}{!}{%
\begin{tabular}[b]{cl@{}c}
    Model 
    & Tilting functions 
    & \begin{tabular}{@{}c@{}}
        Testable\\implications
    \end{tabular}
    \\\hline
    \\[-.5em]
    C1
    & $h_b=\E[g_b\mid W,Z,M,R_{MY}=1]$,~~$k=g_b/h_b$,~~where
    \\[.3em]
    &~~~$g_b=\P(R_{MY}=1\mid W,Z)\,q_{12d}^*(W,M,Y)$,
    \\[.3em]
    &~~~$q_{12d}^*(W,M,Y)\geq1$ and solves 
    \\[.3em]
    &~~~~~~~$[\P(R_{MY}=1\mid W,Z)]^{-1}=\E\left[q_{12d}^*(W,M,Y)\mid W,Z,R_{MY}=1\right]$.
    \\[1em]
    C2 
    & $h_b=\E[g_b\mid W,Z,M,R_{MY}=1]$,~~$k=g_b/h_b$,~~where
    & \ref{T7}
    \\[.3em]
    &~~~$g_b=\P(R_{MY}=1\mid W,Z)\,q_{1e}^*(W,M,Y)q_{2e}^*(W,M,Y)$, 
    \\[.3em]
    &~~~$q_{1e}^*(W,M,Y)\geq1$ and $q_{2e}^*(W,M,Y)\geq1$ and respectively solve
    \\[.3em]
    &~~~~~~~$[\P(R_M=1\mid W,Z,Y,R_Y=1)]^{-1}=\E\left[q_{1e}^*(W,M,Y)\mid W,Z,Y,R_{MY}=1\right]$ and~~~~~~~~
    \\[.3em]
    &~~~~~~~$[\P(R_Y=1\mid W,Z,M,R_M=1)]^{-1}=\E\left[q_{2e}^*(W,M,Y)\mid W,Z,M,R_{MY}=1\right]$.
    \\[1em]
    C3 
    & $h=\E[g\mid W,Z,M,R_{MY}=1]$,~~$k=g/h$,~~where
    \\[.3em]
    &~~~$g=\P(R_M=1\mid W,Z)q_{1f}^*(W,M)\P(R_Y=1\mid W,Z,M,R_M=1)q_{2e}^*(W,M,Y)$,
    \\[.3em]
    &~~~$q_{2e}^*$ is defined as above, and $q_{1f}^*(W,M)\geq1$ and solves
    \\[.3em]
    &~~~~~~~$[\P(R_M=1\mid W,Z)]^{-1}=\E[q_{1f}^*(W,M)\mid W,Z,R_M=1]$.
    \\[1em]
    C4 
    & $h_b=\E[g_b\mid W,Z,M,R_{MY}=1]$,~~$k=g_b/h_b$,~~where
    \\[.3em]
    &~~~$g=\P(R_{MY}=1\mid W,Z)q_{1e}^*(W,M,Y)q_{2f}^*(W,Y)$,
    \\[.3em]
    &~~~$q_{1e}^*$ is defined as above, and $q_{2f}^*(W,Y)\geq1$ and solves
    \\[.3em]
    &~~~~~~~$[\P(R_Y=1\mid W,Z)]^{-1}=\E[q_{2f}^*(W,Y)\mid W,Z,R_Y=1]$.
    \\[-.7em]
    \\\hline
    \multicolumn{3}{l}{\small Notes: $q^*$ functions identify inverse response probabilities. Specifically, $q_{12d}^*=[\P(R_{MY}=1\mid W,M,Y)]^{-1}$ under C1;}
    \\
    \multicolumn{3}{l}{\small~~~~~~~~~~$q_{1e}^*=[\P(R_M=1\mid W,M,Y,R_Y=1)]^{-1}$ under C2/C4; $q_{2e}^*=[\P(R_Y=1\mid W,M,Y,R_M=1)]^{-1}$ under C2/C3;}
    \\
    \multicolumn{3}{l}{\small~~~~~~~~~~$q_{1f}^*=[\P(R_M=1\mid W,M)]^{-1}$ under C3; $q_{2f}^*=[\P(R_Y=1\mid W,Y)]^{-1}$ under C4.}
\end{tabular}%
    }
\end{figure*}

\paragraph*{Assumption C1~~(combined covariate SV)}
\begin{align}
    &Z\independent(R_M,R_Y)\mid W,M,Y,\tag{Cx-1}\label{C-1}
    \\
    &\P(M,Y,Z\mid W,R_{MY}=1)~\text{complete in $Z$}.\tag{C1-2}
\end{align}

\paragraph*{Assumption C2~~(parallel covariate SV)}
\begin{align}
    & Z\independent(R_M,R_Y)\mid W,M,Y,\tag{\ref{C-1}}
    \\
    &\P(M,Z\mid W,Y,R_{MY}=1)~\text{complete in $Z$},\tag{C2-2a}
    \\
    &\P(Y,Z\mid W,M,R_{MY}=1)~\text{complete in $Z$},\tag{C2-2b}
    \\
    &R_M\independent R_Y\mid W,M,Y.\tag{C2-3}
\end{align}

\paragraph*{Assumption C3~~(\textsc{my}-sequential covariate SV)}
\begin{align}
    &Z\independent (R_M,R_Y)\mid W,M,Y,\tag{\ref{C-1}}
    \\
    &\P(M,Z\mid W,R_M=1)~\text{complete in $Z$},\tag{C3-2a}
    \\
    &\P(Y,Z\mid W,M,R_{MY}=1)~\text{complete in $Z$},\tag{C3-2b}
    \\
    &R_M\independent Y\mid W,M.\tag{C3-3}
\end{align}

\paragraph*{Assumption C4~~(\textsc{ym}-sequential covariate SV)}
\begin{align}
    &Z\independent (R_M,R_Y)\mid W,M,Y,\tag{\ref{C-1}}
    \\
    &\P(Y,Z\mid W,R_Y=1)~\text{complete in $Z$},\tag{C4-2a}
    \\
    &\P(M,Z\mid W,Y,R_{MY}=1)~\text{complete in $Z$},\tag{C4-2b}
    \\
    &R_Y\independent M\mid W,Y.\tag{C4-3}
\end{align}

DAGs for these models are shown in Fig.~\ref{fig:shadow-covariate}. As $Z$ influences $M$ and $Y$, models C1-C4 inherit both of the above-mentioned implicit restrictions on the core missingness model: no arrows in to $Y$ and no arrows into $M$.

Identification results are provided in Table~\ref{tab:shadow-covariate}.
Model C2 has the testable implication:
\begin{align*}
    \tag{T7}\label{T7}\frac{\P(R_M=0,R_Y=0\mid W)}{\P(R_M=1,R_Y=1\mid W)}=
    \E\Big\{[q_{1e}^*(W,M,Y)-1]\cdot&[q_{2e}^*(W,M,Y)-1]
    \mid W,R_{MY}=1\Big\},
\end{align*}
where functions $q_{1e}^*$ and $q_{2e}^*$ are defined in the second row of Table~\ref{tab:shadow-covariate}. $(q_{1e}^*-1)$ and $(q_{2e}^*-1)$ respectively identify $\odds(R_M=0\mid W,M,Y,R_Y=1)$ and $\odds(R_Y=0\mid W,M,Y,R_M=1)$.

\begin{figure*}[t!]
    \centering
    \caption{Auxiliary/covariate shadow variable models extended to accommodate missingness in the shadow variable. Each graph represents two models: one with blue but not gray arrows and one with gray but not blue arrows. The gray square containing $M,Y,R_M,R_Y$ represents the core missingness model from D1-D4, M1-M4, U4, C1-C4 in Fig.~\ref{fig:shadow-auxiliary} and Fig.~\ref{fig:shadow-covariate}. Auxiliary $Z$ models are shown conditional on $X,A$; covariate $Z$ models are shown conditional on $W:=(X,A)\setminus Z$.}
    \label{fig:shadow-missing}
    \includegraphics[width=.9\textwidth,page=10]{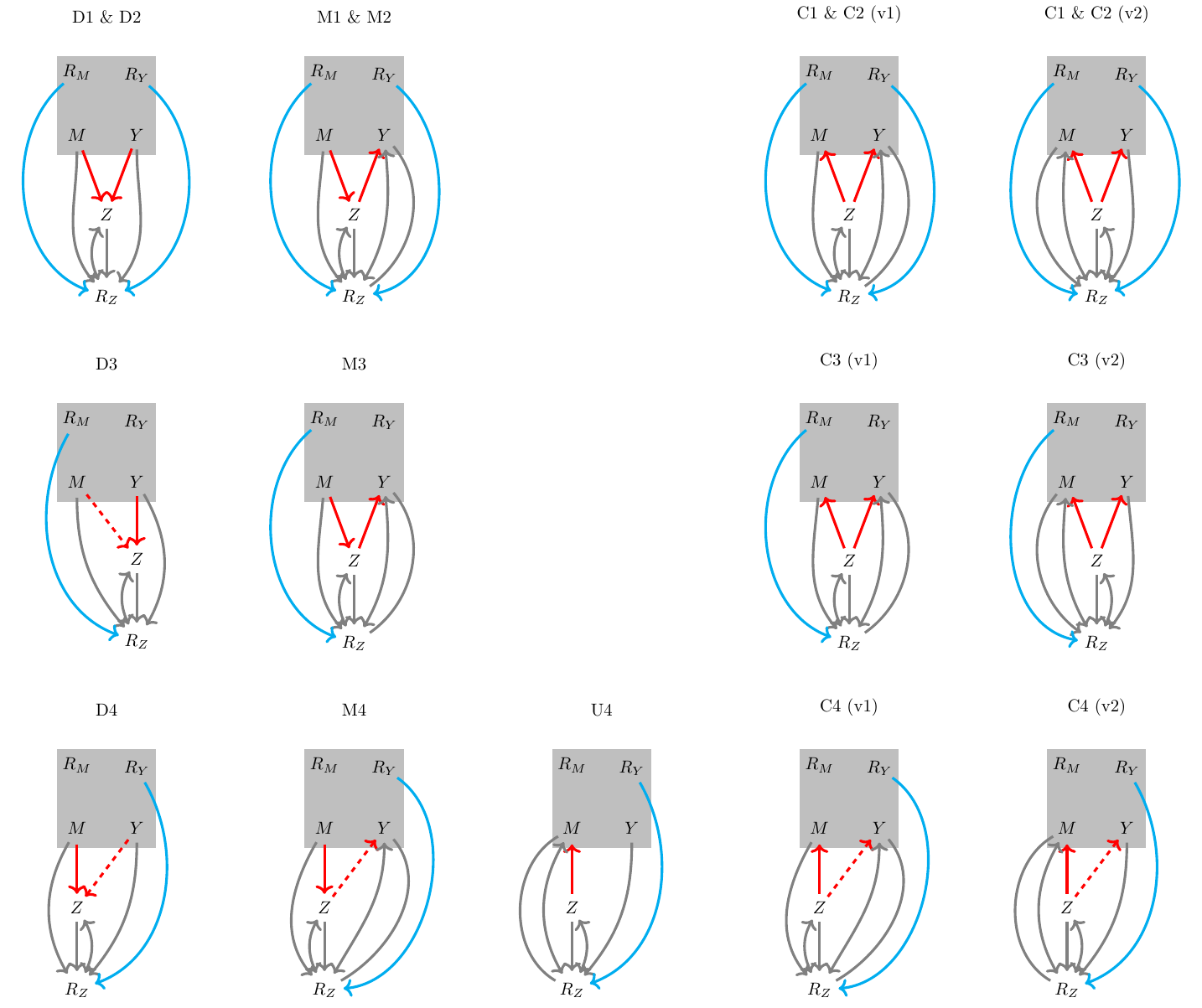}
\end{figure*}

In all these results, we use $W,Z$ (rather than $X,A$) language, as it better differentiates variables that serve as the SV ($Z$) and remaining variables ($W$). Note though that conditioning on $W,Z$ is the same as conditioning on $X,A$, i.e., $\P(R_{MY}=1\mid W,Z)$ is $\P(R_{MY}=1\mid X,A)$.

\subsection{Missingness in auxiliary/covariate shadow variable}
\label{ssec:shadow-missing}

To complete the story, we now address missingness in the auxiliary or covariate SV. Above we mention that such missingness that satisfies a similar assumption to mSV can be accommodated. To be more precise, identifiability is preserved (although identification result needs to be modified) when missingness in $Z$ satisfies assumption mSV for each instance where $Z$ is used as a SV. For example, with models D1 and M1 (which share assumption A1), $Z$ is used as a SV with respect to the missingness in $M,Y$ combined ($R_{MY}$) conditional on $X,A$. The missingness in $Z$ in this case needs to satisfy the following version of assumption mSV:
\begin{align*}
    \P(&R_Z=1,R_{MY}=1\mid X,A,M,Y,Z)>0,
    \\
    \text{and either}~&R_Z\independent(M,Y,Z)\mid X,A,R_{MY}
    \\
    \text{or}~&R_Z\independent R_{MY}\mid X,A,M,Y,Z.
\end{align*}
With models D3 and M3 (which share assumption A3), $Z$ is used first as a SV with respect to the missingness in $M$ ($R_M$) conditional on $X,A$, and second as a SV with respect to the missingness in $Y$ ($R_Y$) conditional on $X,A,M,R_M=1$. The missingness in $Z$ in this case needs to satisfy a double version of assumption mSV:
\begin{align*}
    \P\big(&R_Z=1,R_M=1\mid M,Z,(X,A)\big)>0,
    \\
    \P\big(&R_Z=1,R_Y=1\mid Y,Z,(X,A,M,R_M=1)\big)>0,
    \\
    \text{either}~&R_Z\independent(M,Z)\mid R_M,(X,A)
    \\
    \text{or}~&R_Z\independent R_M\mid M,Z,(X,A),
    \\
    \text{and either}~&R_Z\independent(Y,Z)\mid R_Y,(X,A,M,R_M=1)
    \\
    \text{or}~&R_Z\independent R_Y\mid Y,Z,(X,A,M,R_M=1).
\end{align*}
For the sake of length, we list all the mSV-type assumptions to be paired with model assumptions A1-A4 and C1-C4 in the Appendix together with the modified identification results. Here we comment on some features of these SV missingness models.

Fig.~\ref{fig:shadow-missing} shows the extension of the auxiliary/covariate shadow variable models from Fig.~\ref{fig:shadow-auxiliary} and Fig.~\ref{fig:shadow-covariate} to accommodate missingness in $Z$. To avoid visual clutter and focus on the SV missingness component, we represent the core missingness model component (which varies across models, see Fig.~\ref{fig:shadow-auxiliary} and Fig.~\ref{fig:shadow-covariate}) by a gray square containing $M,Y,R_M,R_Y$. Each diagram in Fig.~\ref{fig:shadow-missing} is intended to convey two models: one including the blue but not gray arrows, the other including the gray but not blue arrows; these reflect the either-or part in assumption mSV. 

Several points to note about connections between $R_Z$ and $R_M,R_Y$ in these models. First, there are no double-headed arrows between $R_Z$ and $R_M,R_Y$, meaning no unobserved cause of mediator or outcome missingness is also a cause of SV missingness. Second, all the blue arrows go in one direction, into $R_Z$, which means $R_Z$ does not influence $R_M,R_Y$. Third, while models in the top row allow two blue arrows from both $R_M$ and $R_Y$, those in the lower two rows only allow one of these two blue arrows. This is because the latter are sequential SV models. Specifically, models in the second row rely on $Z$ being a SV for $M$ (unconditional on $Y,R_Y$), so a blue $R_Y\to R_Z$ arrow if present would result in $R_Z$ being connected to both $(Z,M)$ and $R_M$, which is not allowed. The explanation is similar for models in the third row.
Fourth, the temporal location of $R_Z$ may rule out some connections and simplify the story. For example, in covariate SV models where $Z$ precedes $M,Y$ in time, if (as typically is the case) the measurement of $Z$ also precedes the measurement of $M,Y$, then blue arrows from $R_M,R_Y$ to $R_Z$ do not exist; however, the model still caries the important constraint that the missingness in $Z$ does not cause and does not share unobserved causes with the missingness in $M$ and $Y$. Similarly, in auxiliary midstream SV models, if the measurement of $Z$ precedes the measurement of $Y$, then the blue arrow from $R_Y$ to $R_Z$ is ruled out.

A couple of points about the gray arrows. First, some models include $R_Z$-$Y$ connections with arrows into $Y$; others include $R_Z$-$M$ connections with arrows into $M$; and other models include neither. The absence of such arrows simply reflects that the core missingness model includes arrows into that variable ($Y$ or $M$). If such arrows do not exist in the core missingness model, the SV missingness model can include arrows into the variable ($Y$ or $M$). Second, a related point is that for C1-C4, there are two options with arrows into either $Y$ or $M$; this is because the core missingness model does not include arrows into either of these variables. Third, if $Z$ is measured retrospectively after the occurrence of $M$ or $Y$, then the arrow from $R_Z$ to $M$ or $Y$ does not exist.


\section{Discussion}
\label{sec:conclusion}

For a setting with mediator and outcome missingness, we have explored model options within two missingness model classes: self-separated missingness models (S1-S5), and self-connected missingness models with built-in (Z1-Z5, C1-C4) or auxiliary (U3-U4, M1-M4, D1-D4) SVs. 
Under these models, the two conditional distributions that play an important role in many causal mediation estimands, $\P(M\mid X,A)$ and $\P(Y\mid X,A,M)$ are identified. We have also pointed out conditions on the missingness in the SV that preserve identifiability.

With the two examples, we walked through considerations of the plausibility of the restrictions of different models and potential SVs for handling self-connected missingness. This thought exercise obtained many possibilities under different scenarios of data availability. It makes clear that the utility of these models in practice depends on domain-specific subject-matter knowledge and intimate familiarity with study design and data content, and would benefit from awareness of auxiliary data sources. Also, these considerations should start in the study planning stage, so that data collection can be designed such that useful variables are available to help handle self-connected missingness in key variables (here mediators and outcomes), either through blocking problematic back-door paths or serving as shadow variables.

A next task is to develop methods for estimation and inference.
To support this, we have presented all identification results in the form of tilting functions, which are likely useful for non-response weighting strategies. (We resisted the temptation to simplify expressions in results when that would obscure the tilting functions.) Much remains to be investigated, however. As the identification result does not touch certain data components (a simple example under S1 is outcome data from study participants whose outcome is observed but mediator is not), a question is how to make use of all the information in the data, or even what other information is available at all, and under what conditions. Answering this question may also shed light on how to use an imputation approach where all observed data is used and missing data is to be imputed. Aside from these general questions, there are questions about the statistical models to be used and their varying properties, spanning parametric strategies \citep[as used in][]{zuo2024MediationAnalysisMediator} and semi-parametric strategies \citep[as used in][]{miao2024IdentificationSemiparametricEfficiency}. While there is an immediate need for methods that target the estimands (mediation effects) in this setting, there is an appeal for tools that estimate self-connected response probabilities specifically, as those would be useful in different use cases.

Theory-wise, we believe that our advancement of shadow variable theory (covering two variables with missingness and allowing missingness in the shadow variable) is generally beneficial and not limited to the current mediation problem, as it can be extended easily to settings with more than two time points. The odds-tilting techniques, which seem to be small tricks, are quite useful. It is interesting to note that our derivations for both self-separated and self-connected missingness models center on the missingness odds. It is possible that this is just a random result of us finding that path among many other possible paths. Or it might mean that the missingness odds has a special place in the problem of missing data.

\bibliography{refs}

\clearpage
\pagenumbering{arabic}
\renewcommand*{\thepage}{A\arabic{page}}
\appendix
\noindent{\Large \textbf{APPENDIX}}

\vspace{2em}

\vspace{-4em}

\part{}
\parttoc

\renewcommand{\thefigure}{A\arabic{figure}}

\setcounter{figure}{0}

\section{Appendix to Section~\ref{sec:setting} (Setting)}
\label{asec:setting}

\subsection{Disjoint MAR model - numerical example}

The point made in the paper about the disjoint MAR model in Fig.~\ref{fig:badMAR} is that under this model $\P(Y\mid X,A,M)$ is not identified. Here we illustrate it with a numerical example, specifically showing that more than one full-data distribution implies the same observed-data distribution.

To keep things simple, we leave $X,A$ out to zoom into the main part of the story. One can think about all quantities here as implicitly conditioning on a specific level ($X=x,A=a$) of $X,A$. Also for simplicity, consider binary $M$ and $Y$. 

We will (1) assume a true full-data distribution -- for $(M,R_M,Y,R_Y)$ -- and pin down the density of interest $\P(Y\mid M)$ and derive the observed-data distribution -- for $(R_M,R_Y,R_MM,R_YY)$; and then (2) come up with a different full-data distribution with a different density for $\P(Y\mid M)$ that implies the same observed-data distribution.

Assume the full-data distribution:
\begin{align*}
    \P(M=1)&=0.5,
    \\
    \P(R_M=1\mid M)&=0.5,
    \\
    \P(Y=1\mid M,R_M)&={\color{blue}0.3+0.3M+0.2R_M},
    \\
    \P(R_Y=1\mid M,Y,R_M)&=0.3+0.1R_M.
\end{align*}
With this full-data distribution, the density of interest $\P(Y\mid M)$ is given by:
\begin{align*}
    \P(Y=1\mid M)
    &=\E[\P(Y=1\mid M,R_M)\mid M]
    ={\color{blue}0.3+0.3M+0.2}\,\underbrace{\color{blue}\E[R_M\mid M]}_{0.5}
    \\
    &={\color{blue}0.4+0.3M}.
\end{align*}
Also with this full-data distribution, the observed-data distribution is:
\begin{alignat*}{2}
    &\P(R_M=1)&&=0.5,
    \\
    &\P(R_Y=1\mid R_M)&&=0.3+0.1R_M,
    \\
    &\P(M=1\mid R_Y,R_M=1)&&=0.5,
    \\
    &\P(Y=1\mid R_M,R_Y=1)&&=\P(Y=1\mid R_M)=\E[\P(Y=1\mid M,R_M)\mid R_M]
    ={\color{blue}0.3+0.3}\underbrace{\color{blue}\E[M\mid R_M]}_{0.5}{\color{blue}+0.2R_M}
    \\
    &&&={\color{blue}0.45+0.2R_M},
    \\
    &\P(Y=1\mid M,R_{MY}=1)&&=\P(Y=1\mid M,R_M=1)={\color{blue}0.3+0.3M+0.2}
    \\
    &&&={\color{blue}0.5+0.3M}.
\end{alignat*}

Now consider a modified full-data distribution:
\begin{align*}
    \P(M=1)&=0.5,
    \\
    \P(R_M=1\mid M)&=0.5,
    \\
    \P(Y=1\mid M,R_M)&={\color{purple}0.25+0.4M+(0.25-0.1M)R_M},
    \\
    \P(R_Y=1\mid M,Y,R_M)&=0.3+0.1R_M.
\end{align*}
In this case, the density of interest is different:
\begin{align*}
    \P(Y=1\mid M)
    &=\E[\P(Y=1\mid M,R_M)\mid M]
    ={\color{purple}0.25+0.4M+(0.25-0.1M)}\underbrace{\color{purple}\E[R_M\mid M]}_{0.5}
    \\
    &={\color{purple}0.375+0.35M}.
\end{align*}
But the observed-data distribution is the same:
\begin{alignat*}{2}
    &\P(R_M=1)&&=0.5,
    \\
    &\P(R_Y=1\mid R_M)&&=0.3+0.1R_M,
    \\
    &\P(M=1\mid R_Y,R_M=1)&&=0.5,
    \\
    &\P(Y=1\mid R_M,R_Y=1)&&=\P(Y=1\mid R_M)=\E[\P(Y=1\mid M,R_M)\mid R_M]
    \\
    &&&={\color{purple}0.25+0.4}\underbrace{\color{purple}\E[M\mid R_M]}_{0.5}+{\color{purple}(0.25-0.1}\underbrace{\color{purple}\E[M\mid R_M]}_{0.5}{\color{purple})R_M}
    \\
    &&&={\color{purple}0.45+0.2R_M},
    \\
    &\P(Y=1\mid M,R_{MY}=1)&&=\P(Y=1\mid M,R_M=1)={\color{purple}0.25+0.4M+0.25-0.1M}
    \\
    &&&={\color{purple}0.5+0.3M}.
\end{alignat*}

In fact, there are infinitely many full-data distributions (with varying $\P(Y\mid M)$) that imply the same observed-data distribution. They belong to a family with
$$\P(Y=1\mid M,R_M)=a+bM+(c+dM)R_M,$$
where $(a,b,c,d)$ solves
\begin{align*}
    \begin{pmatrix}
        a+0.5b+(c+0.5d)R_M\\
        (a+c)+(b+d)M
    \end{pmatrix}
    =
    \begin{pmatrix}
        0.45+0.2R_M\\0.5+0.3M
    \end{pmatrix}.
\end{align*}

\section{Appendix to Section~\ref{sec:separated} (Self-separated missingness models)}
\label{asec:separated}

\begin{lemma}\label{lm:extraction}
    If $A\independent(B,C)$ then $A\independent B\mid C$.
\end{lemma}

\begin{proof}[Proof of Lemma~\ref{lm:extraction}]
    First, note that if $A\independent(B,C)$ then $A\independent C$. This is because
    \begin{align*}
        \P(A\mid C)
        &=\E[\P(A\mid B,C)\mid C] && (\text{law of total probability})
        \\
        &=\E[\P(A)\mid C] && (A\independent(B,C))
        \\
        &=\P(A).
    \end{align*}
    That $A\independent B\mid C$ follows,
    \begin{align*}
        \P(B\mid C)
        &=\frac{\P(B,C)}{\P(C)} && (\text{Bayes' rule})
        \\
        &=\frac{\P(B,C\mid A)}{\P(C\mid A)} && (A\independent(B,C)~\text{and}~A\independent C)
        \\
        &=\P(B\mid A,C). && (\text{Bayes' rule})
    \end{align*}
\end{proof}

\subsection{SIM models}
\label{assec:SIM}

\subsubsection{Models S1}
\label{asssec:S1}

\paragraph{Assumption S1:}\hfill

\begin{minipage}{.35\textwidth}
    \small First statement:
    \begin{align*}
    R_M&\independent(M,Y)\mid X,A,\tag{\ref{S1-1}}
    \\
    R_Y&\independent Y\mid X,A,M,R_M.\tag{\ref{S1-2}}
    \end{align*}
\end{minipage}
\hspace{.1\textwidth}
\begin{minipage}{.35\textwidth}
    \small Second statement:
    \begin{align*}
    \Rm&\independent M\mid X,A,\tag{\ref{S1-1b}}
    \\
    (R_M\!,R_Y)&\independent Y\mid X,A,M.\tag{\ref{S1-2b}}
    \end{align*}
\end{minipage}
\hspace{.1\textwidth}
\begin{minipage}{.1\textwidth}
    \includegraphics[width=\textwidth, page=1]{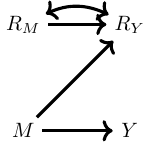}
\end{minipage}

\paragraph{Proof of equivalence of two statements}

\begin{itemize}
    \item[($\Rightarrow$)]
    Assume the first statement of S1: \ref{S1-1} and \ref{S1-2}. As \ref{S1-1} contains \ref{S1-1b}, we only need to show that the first statement of S1 implies \ref{S1-2b}. This is true because
    \begin{align*}
        \P(Y\mid X,A,M)
        &=\P(Y\mid X,A,M,R_M)=\P(Y\mid X,A,M,R_M,R_Y),
    \end{align*}
    where the first equality is due to $R_M\independent Y\mid X,A,M$ (implied by \ref{S1-1} under Lemma~\ref{lm:extraction}) and the second equality is due to \ref{S1-2}.

    \item[($\Leftarrow$)]
    Assume the second statement of S1: \ref{S1-1b} and \ref{S1-2b}. As \ref{S1-2b} implies \ref{S1-2} by Lemma~\ref{lm:extraction}, we only need to show that the second statement of S1 implies \ref{S1-1}. This is true because
    \begin{align*}
        \P(R_M\mid X,A,M,Y)=\P(R_M\mid X,A,M)=\P(R_M\mid X,A),
    \end{align*}
    where the first equality is due to \ref{S1-2b} and the second is due to \ref{S1-1b}.
\end{itemize}






\subsubsection{Models S2}
\label{asssec:S2}

\paragraph{Assumption S2:}\hfill

\begin{minipage}{.35\textwidth}
    \small First statement:
    \begin{align*}
    \Ry&\independent(M,Y)\mid X,A,\tag{\ref{S2-1}}
    \\
    \Rm&\independent M\mid X,A,Y,R_Y.\tag{\ref{S2-2}}
    \end{align*}
\end{minipage}
\hspace{.1\textwidth}
\begin{minipage}{.35\textwidth}
    \small Second statement:
    \begin{align*}
    \Ry&\independent Y\mid X,A,\tag{\ref{S2-1b}}
    \\
    (R_Y\!,R_M)&\independent M\mid X,A,Y.\tag{\ref{S2-2b}}
    \end{align*}
\end{minipage}
\hspace{.1\textwidth}
\begin{minipage}{.1\textwidth}
    \includegraphics[width=\textwidth, page=2]{standalone-appendix.pdf}
\end{minipage}

\paragraph{Identification}

\begin{align*}
    \P(R_{MY}=1\mid X,A,M,Y)
    &\overset{\text{\ref{S2-2b}}}{=}\P(R_{MY}=1\mid X,A,Y)
    \\
    &=\P(R_Y=1\mid X,A,Y)\P(R_M=1\mid X,A,Y,R_Y=1)
    \\
    &\overset{\text{\ref{S2-1b}}}{=}\P(R_Y=1\mid X,A)\P(R_M=1\mid X,A,Y,R_Y=1)
    \\
    &\overset{\text{\ref{S2-2}}}{=}\P(R_Y=1\mid X,A)\P(R_M=1\mid X,A,M,Y,R_Y=1).
\end{align*}
It follows that
\begin{align*}
    g_b
    :=h_bk
    =\frac{\P(R_{MY}=1\mid X,A)}{\P(R_{MY}=1\mid X,A,M,Y)}
    =\frac{\P(R_M=1\mid X,A,R_Y=1)}{\P(R_M=1\mid X,A,Y,R_Y=1)},
\end{align*}
which identifies the tilting functions $h_b=\E[g_b\mid X,A,M,R_{MY}=1]$ and $k=g_b/h_b$.

\subsection{Other models}

\subsubsection{Model S3}
\label{asssec:S3}

\begin{minipage}{.7\textwidth}
    \paragraph{Assumption S3:}\hfill

    \begin{align*}
    \Rm&\independent (M,R_Y)\mid X,A,Y,\tag{\ref{S3-1}}
    \\
    \Ry&\independent (Y,R_M)\mid X,A,M,\tag{\ref{S3-2}}
    \end{align*}
\end{minipage}
\hspace{.18\textwidth}
\begin{minipage}{.12\textwidth}
    \includegraphics[width=\textwidth, page=3]{standalone-appendix.pdf}
\end{minipage}

\paragraph{Identification}

\begin{align*}
    \P(R_{MY}=1\mid X,A,M,Y)
    &=\P(R_M=1\mid X,A,M,Y)\P(R_Y=1\mid X,A,M,Y,R_M=1)
    \\
    &\overset{\text{\ref{S3-1}}}{=}\P(R_M=1\mid X,A,Y,R_Y=1)\P(R_Y=1\mid X,A,M,Y,R_M=1)
    \\
    &\overset{\text{\ref{S3-2}}}{=}\P(R_M=1\mid X,A,Y,R_Y=1)\P(R_Y=1\mid X,A,M,R_M=1).
\end{align*}
It follows that
\begin{align*}
    g:=hk
    &=\frac{\P(R_M=1\mid X,A)\P(R_Y=1\mid X,A,M,R_M=1)}{\P(R_{MY}=1\mid X,A,M,Y)}
    \\
    &=\frac{\P(R_M=1\mid X,A)}{\P(R_M=1\mid X,A,Y,R_Y=1)},
\end{align*}
which identifies the tilting functions $h=\E[g\mid X,A,M,R_{MY}=1]$ and $k=g/h$.

\paragraph{Testable implication}

The testable implication of this model was discovered by \cite{nabi2023TestabilityGoodnessFit}, using a special odds-ratio parameterization \citep{chen2007SemiparametricOddsRatio} for the missingness model. Here we skip this topic and work directly toward a simple expression of the testable implication. (This derivation is based on the key idea from \cite{nabi2023TestabilityGoodnessFit} that the model implies conditional independence of $R_M$ and $R_Y$ so an odds ratio relating them is 1, and uses hints from \citeauthor{malinsky2022SemiparametricInferenceNonmonotone}'s (\citeyear{malinsky2022SemiparametricInferenceNonmonotone}) method for estimating this odds ratio under the model if we ignore this specific independence implication.)

First we show a general identity
\begin{align}
    \frac{\P(R_M,R_Y\mid X,A)}{\P(R_M=1,R_Y=1\mid X,A)}=\E\left[\frac{\P(R_M,R_Y\mid X,A,M,Y)}{\P(R_M=1,R_Y=1\mid X,A,M,Y)}\mid X,A,R_{MY}=1\right]\label{eq:ini1}
\end{align}
by starting with the RHS and arriving at the LHS:
\begin{align*}
    \text{RHS}
    &=\E\left[\frac{\P(R_M,R_Y,M,Y\mid X,A)}{\P(R_M=1,R_Y=1,M,Y\mid X,A)}\mid X,A,R_{MY}=1\right]
    \\
    &=\E\left[\frac{\P(R_M,R_Y\mid X,A)\P(M,Y\mid X,A,R_M,R_Y)}{\P(R_M=1,R_Y=1\mid X,A)\P(M,Y\mid X,A,R_{MY}=1)}\mid X,A,R_{MY}=1\right]
    \\
    &=\frac{\P(R_M,R_Y\mid X,A)}{\P(R_M=1,R_Y=1\mid X,A)}\E\left[\frac{\P(M,Y\mid X,A,R_M,R_Y)}{\P(M,Y\mid X,A,R_{MY}=1)}\mid X,A,R_{MY}=1\right]
    \\
    &=\frac{\P(R_M,R_Y\mid X,A)}{\P(R_M=1,R_Y=1\mid X,A)}=\text{LHS}.
\end{align*}

Side note: It is the single missing indicator version of (\ref{eq:ini1}) that gives us the key result (\ref{eq:shadow}) for the shadow variable case in Section~\ref{sec:shadow}. (\ref{eq:ini1}) is the general version that allows multiple missing indicators. The generic statement of (\ref{eq:ini1}) involves two conditioning sets, $S_1\subset S_2$; in the current context $S_1=\{X,A\}$, $S_2=\{X,A,M,Y\}$.

We can keep working with (\ref{eq:ini1}), but for the current task, what will be useful is where we plug 0 in for $R_M$ and $R_Y$ in the numerators:
\begin{align}
    \frac{\P(R_M=0,R_Y=0\mid X,A)}{\P(R_M=1,R_Y=1\mid X,A)}=\E\left[\frac{\P(R_M=0,R_Y=0\mid X,A,M,Y)}{\P(R_M=1,R_Y=1\mid X,A,M,Y)}\mid X,A,R_{MY}=1\right].\label{eq:ini2}
\end{align}
We re-express the integrand on the RHS,
\begin{align*}
    &\frac{\P(R_M=0,R_Y=0\mid X,A,M,Y)}{\P(R_M=1,R_Y=1\mid X,A,M,Y)}=
    \\
    &=\frac{\P(R_M=0,R_Y=1\mid X,A,M,Y)}{\P(R_M=1,R_Y=1\mid X,A,M,Y)}\cdot\frac{\P(R_M=0,R_Y=0\mid X,A,M,Y)}{\P(R_M=0,R_Y=1\mid X,A,M,Y)}
    \\
    &=\frac{\P(R_M=0\mid X,A,M,Y,R_Y=1)}{\P(R_M=1\mid X,A,M,Y,R_Y=1)}\cdot\frac{\P(R_Y=0\mid X,A,M,Y,R_M=0)}{\P(R_Y=1\mid X,A,M,Y,R_M=0)}
    \\
    &=\frac{\P(R_M=0\mid X,A,M,Y,R_Y=1)}{\P(R_M=1\mid X,A,M,Y,R_Y=1)}\cdot\frac{\P(R_Y=0\mid X,A,M,Y,R_M=1)}{\P(R_Y=1\mid X,A,M,Y,R_M=1)}\times
    \\
    &~~~~~\left[\frac{\P(R_Y=0\mid X,A,M,Y,R_M=0)}{\P(R_Y=1\mid X,A,M,Y,R_M=0)}\Big/\frac{\P(R_Y=0\mid X,A,M,Y,R_M=1)}{\P(R_Y=1\mid X,A,M,Y,R_M=1)}\right].
\end{align*}
Therefore we exan re-express (\ref{eq:ini2}) as
\begin{align}
    \frac{\P(R_M=0,R_Y=0\mid X,A)}{\P(R_M=1,R_Y=1\mid X,A)}
    =\E\bigg[&\text{odds}(R_M=0\mid X,A,M,Y,R_Y=1)\,\text{odds}(R_Y=0\mid X,A,M,Y,R_M=1)\times\nonumber
    \\
    &\frac{\text{odds}(R_Y=0\mid X,A,M,Y,R_M=0)}{\text{odds}(R_Y=0\mid X,A,M,Y,R_M=1)}\mid X,A,R_{MY}=1\bigg].\label{eq:ini3}
\end{align}

Under S3, the two missingness odds and the missingness odds ratio inside the expectation above are equal to $\text{odds}(R_M=0\mid X,A,Y,R_Y=1)$, $\text{odds}(R_Y=0\mid X,A,M,R_M=1)$ and 1, respectively.
This means, under S3, (\ref{eq:ini3}) becomes
\begin{align*}
    &\frac{\P(R_M=0,R_Y=0\mid X,A)}{\P(R_M=1,R_Y=1\mid X,A)}=
    \\
    &~~~~~~\E\left[\text{odds}(R_M=0\mid X,A,M,Y,R_Y=1)\,\text{odds}(R_Y=0\mid X,A,M,Y,R_M=1)\mid X,A,R_{MY}=1\right].\tag{T3}
\end{align*}

\subsubsection{Model S4}

\begin{minipage}{.7\textwidth}
    \paragraph{Assumption S4:}
    \begin{align*}
    \Rm&\independent M\mid X,A,Y,\tag{\ref{S4-1}}
    \\
    \Ry&\independent(M,Y)\mid X,A,R_M.\tag{\ref{S4-2}}
    \end{align*}
\end{minipage}
\hspace{.18\textwidth}
\begin{minipage}{.12\textwidth}
    \includegraphics[width=\textwidth, page=4]{standalone-appendix.pdf}
\end{minipage}

\paragraph{Identification}

\begin{align*}
    g:=hk
    &=\frac{\P(R_M=1\mid X,A)\P(R_Y=1\mid X,A,M,R_M=1)}{\P(R_{MY}=1\mid X,A,M,Y)}
    \\
    &=\frac{\P(R_M=1\mid X,A)\P(R_Y=1\mid X,A,M,R_M=1)}{\P(R_M=1\mid X,A,M,Y)\P(R_Y=1\mid X,A,M,Y,R_M=1)}
    \\
    &\overset{\text{\ref{S4-2}}}{=}\frac{\P(R_M=1\mid X,A)}{\P(R_M=1\mid X,A,M,Y)}
    \\
    &\overset{\text{\ref{S4-1}}}{=}\frac{\P(R_M=1\mid X,A)}{\P(R_M=1\mid X,A,Y)}.
\end{align*}

\begin{align*}
    \frac{\text{odds}(R_M=0\mid X,A,Y)}{\text{odds}(R_M=0\mid X,A,Y,R_Y=1)}=\frac{\P(R_Y=1\mid X,A,Y,R_M=1)}{\P(R_Y=1\mid X,A,Y,R_M=0)}=\frac{\P(R_Y=1\mid X,A,R_M=1)}{\P(R_Y=1\mid X,A,R_M=0)},
\end{align*}
where the first equal sign is obtained by the first odds-tilting technique (\ref{eq:odds-tilt1}) described in Section~\ref{ssec:tilt-the-odds}, and the second equal sign is given by \ref{S4-2}. This obtains
\begin{align*}
    \frac{1}{\P(R_M=1\mid X,A,Y)}=1+\text{odds}(R_M=0\mid X,A,Y,R_Y=1)\frac{\P(R_Y=1\mid X,A,R_M=1)}{\P(R_Y=1\mid X,A,R_M=0)},
\end{align*}
and
\begin{align*}
    g=\P(R_M=1\mid X,A)\left[1+\text{odds}(R_M=0\mid X,A,Y,R_Y=1)\frac{\P(R_Y=1\mid X,A,R_M=1)}{\P(R_Y=1\mid X,A,R_M=0)}\right],
\end{align*}
which identifies the tilting functions $h=\E[g\mid X,A,M,R_{MY}=1]$ and $k=g/h$.

\subsubsection{Model S5}

\begin{minipage}{.7\textwidth}
    \paragraph{Assumption S5:}
    \begin{align}
    &(R_M,R_Y)\independent Y\mid X,A,M,\tag{\ref{S5-1}}
    \\*
    &\Rm\independent(M,Y)\mid X,A,R_Y,\tag{\ref{S5-2}}
    \\
    &\P(R_M=1\mid X,A,R_Y=0)>0.\tag{\ref{S5-3}}
    \end{align}
\end{minipage}
\hspace{.18\textwidth}
\begin{minipage}{.12\textwidth}
    \includegraphics[width=\textwidth, page=5]{standalone-appendix.pdf}
\end{minipage}

\subsubsection*{Identification}

\ref{S5-1} identifies $\P(Y\mid X,A,M)$ by complete cases, i.e., $k=1$. To identify $\P(M\mid X,A)$, we first identify $\P(R_Y,M,\mid X,A)$ and then marginalize over $R_Y$.

Again, we use density tilting,
\begin{align*}
    \frac{\P(R_Y,M\mid X,A)}{\P(R_Y,M\mid X,A,R_M=1)}
    =\frac{\P(R_M=1\mid X,A)}{\P(R_M=1\mid X,A,M,R_Y)}
    =\frac{\P(R_M=1\mid X,A)}{\P(R_M=1\mid X,A,R_Y)},
\end{align*}
where the first equal sign is due to Bayes' rule and the second is due to \ref{S5-2}, to obtain
\begin{align*}
    \P(R_Y,M\mid X,A)
    &=\P(R_Y,M\mid X,A,R_M=1)\frac{\P(R_M=1\mid X,A)}{\P(R_M=1\mid X,A,R_Y)}
    \\
    &=\P(M\mid X,A,R_M=1)\P(R_Y\mid X,A,M,R_M=1)\frac{\P(R_M=1\mid X,A)}{\P(R_M=1\mid X,A,R_Y)}.
\end{align*}

Marginalizing over $R_Y$ obtains
\begin{align*}
    \P(M\mid X,A)
    =\P(M\mid X,A,R_M=1)\E\left[\frac{\P(R_M=1\mid X,A)}{\P(R_M=1\mid X,A,R_Y)}\mid X,A,M,R_M=1\right],
\end{align*}
or
\begin{align*}
    h=\E\left[\frac{\P(R_M=1\mid X,A)}{\P(R_M=1\mid X,A,R_Y)}\mid X,A,M,R_M=1\right],
\end{align*}
which can also be expressed as
\begin{align*}
    h=\P(R_M=1\mid X,A)\left[\frac{\P(R_Y=1\mid X,A,M,R_M=1)}{\P(R_M=1\mid X,A,R_Y=1)}+\frac{\P(R_Y=0\mid X,A,M,R_M=1)}{\P(R_M=1\mid X,A,R_Y=0)}\right].
\end{align*}


\subsection{Odds tilting}
\label{assec:odds-tilt}

Here we derive the two odds tilting functions presented in Section~\ref{ssec:tilt-the-odds}:
\begin{align}
    \frac{\text{odds}(R_M=0\mid X,A,Y)}{\text{odds}(R_M=0\mid X,A,Y,R_Y=1)}&=\frac{\P(R_Y=1\mid X,A,Y,R_M=1)}{\P(R_Y=1\mid X,A,Y,R_M=0)},\tag{\ref{eq:odds-tilt1}}
    \\
    \frac{\text{odds}(R_M=0\mid X,A,Y)}{\text{odds}(R_M=0\mid X,A)}&=\frac{\P(Y\mid X,A,R_M=0)}{\P(Y\mid X,A,R_M=1)}.\tag{\ref{eq:odds-tilt2}}
\end{align}
We start with (\ref{eq:odds-tilt1}):
\begin{align*}
    \text{LHS}
    &:=\frac{\P(R_M=0\mid X,A,Y)}{\P(R_M=1\mid X,A,Y)}\Big/\frac{\P(R_M=0\mid X,A,Y,R_Y=1)}{\P(R_M=1\mid X,A,Y,R_Y=1)}
    \\
    &=\frac{\P(R_M=0\mid X,A,Y)}{\P(R_M=0\mid X,A,Y,R_Y=1)}\Big/\frac{\P(R_M=1\mid X,A,Y)}{\P(R_M=1\mid X,A,Y,R_Y=1)}
    \\
    &\overset{\text{Bayes' rule}}{=}
    \frac{\P(R_Y=1\mid X,A,Y)}{\P(R_Y=1\mid X,A,Y,R_M=0)}\Big/
    \frac{\P(R_Y=1\mid X,A,Y)}{\P(R_Y=1\mid X,A,Y,R_M=1)}
    \\
    &=\frac{\P(R_Y=1\mid X,A,Y,R_M=1)}{\P(R_Y=1\mid X,A,Y,R_M=0)}
    =\text{RHS}.
\end{align*}
And now (\ref{eq:odds-tilt2}):
\begin{align*}
    \text{LHS}
    &:=\frac{\P(R_M=0\mid X,A,Y)}{\P(R_M=1\mid X,A,Y)}\Big/\frac{\P(R_M=0\mid X,A)}{\P(R_M=1\mid X,A)}
    \\
    &=\frac{\P(R_M=1\mid X,A)}{\P(R_M=1\mid X,A,Y)}\Big/\frac{\P(R_M=0\mid X,A)}{\P(R_M=0\mid X,A,Y)}
    \\
    &\overset{\text{Bayes' rule}}{=}
    \frac{\P(Y\mid X,A)}{\P(Y\mid X,A,R_M=1)}\Big/\frac{\P(Y\mid X,A)}{\P(Y\mid X,A,R_M=0)}
    \\
    &=\frac{\P(Y\mid X,A,R_M=0)}{\P(Y\mid X,A,R_M=1)}
    =\text{RHS}.
\end{align*}

\section{Appendix to Section~\ref{sec:shadow} (Self-connected missingness models with shadow variables)}
\label{asec:shadow}

\subsection{General theory}
\label{asec:shadow-theory1}

\subsubsection{The key result (\ref{eq:shadow})}
\label{asubsec:shadow-theory}

\begin{align*}
    \frac{\P(R_V=0\mid Z)}{\P(R_V=1\mid Z)}=\E\bigg[\frac{\P(R_V=0\mid V)}{\P(R_V=1\mid V)}\mid Z,R_V=1\bigg].\tag{\ref{eq:shadow}}
\end{align*}

\paragraph{Proof of (\ref{eq:shadow}).}

Under the assumption $Z\independent R_V\mid V$, the odds inside the expectation on the RHS and (\ref{eq:shadow}) is equal to $\mfrac{\P(R_V=0\mid Z,V)}{\P(R_V=1\mid Z,V)}$. This means (\ref{eq:shadow}) holds if 
\begin{align*}
    \frac{\P(R_V=0\mid Z)}{\P(R_V=1\mid Z)}=\E\bigg[\frac{\P(R_V=0\mid V,Z)}{\P(R_V=1\mid V,Z)}\mid Z,R_V=1\bigg].
\end{align*}
This is a general result that is easy to show:
\begin{align*}
    \text{RHS}
    &=\E\bigg[\frac{\P(R_V=0,V\mid Z)/\P(V\mid Z)}{\P(R_V=1,V\mid Z)/\P(V\mid Z)}\mid Z,R_V=1\bigg]
    \\
    &=\E\bigg[\frac{\P(R_V=0\mid Z)\P(V\mid Z,R_V=0)}{\P(R_V=1\mid Z)\P(V\mid Z,R_V=1)}\mid Z,R_V=1\bigg]
    \\
    &=\frac{\P(R_V=0\mid Z)}{\P(R_V=1\mid Z)}\E\bigg[\frac{\P(V\mid Z,R_V=0)}{\P(V\mid Z,R_V=1)}\mid Z,R_V=1\bigg]
    \\
    &=\frac{\P(R_V=0\mid Z)}{\P(R_V=1\mid Z)}\E[1\mid Z,R_V=0]
    \\
    &=\text{LHS}.
\end{align*}

\subsubsection{Connections with existing results}

\paragraph{Connections with results in \cite{zuo2024MediationAnalysisMediator}.}

\citeauthor{zuo2024MediationAnalysisMediator} show identifiability of the densities of interest through identifiability of the self-connected missingness odds. There are different cases but techniques tend to repeat, so here we take as example the simplest case where the mediator is missing but the outcome is observed and serves as a shadow variable. The result for the self-connected mediator missingness odds in the proof of their Theorem 1 is
\begin{align*}
    \P(Y,\Rm=0\mid A,X)=\int \P(M=m,Y,\Rm=1\mid A,X)\underbrace{\frac{\P(\Rm=0\mid M=m,A,X)}{\P(\Rm=1\mid M=m,A,X)}}_{\text{odds}(R_M=0\mid M=m,A,X)}dm.
\end{align*}
It helps to translate this to our current notation, so we replace $M$ with $V$, $Y$ with $Z$ and make $A,X$ implicit, and write this as
\begin{align*}
    \P(Z,R_V=0)=\int \P(V=v,Z,R_V=1)\,\text{odds}(R_V=0\mid V=v)dv.
\end{align*}
Connecting this to (\ref{eq:shadow}) requires two simple manipulations. First, dividing both sides by $\P(Z)$ obtains a conditional on $Z$ version of the equation:
\begin{align*}
    \P(R_V=0\mid Z)=\int \P(V=v,R_V=1\mid Z)\,\text{odds}(R_V=0\mid V=v)dv.
\end{align*}
Then dividing both sides by $\P(R_V=1\mid Z)$ puts the LHS in odds terms:
\begin{align*}
    \frac{\P(R_V=0\mid Z)}{\P(R_V=1\mid Z)}=\int \P(V=v\mid Z,R_V=1)\,\text{odds}(R_V=0\mid V=v)dv.
\end{align*}
And this is exactly (\ref{eq:shadow}): $\text{odds}(R_V=0\mid Z)=\E[\text{odds}(R_V=0\mid V)\mid Z,R_V=1]$.

\paragraph{Connections with results in \cite{miao2024IdentificationSemiparametricEfficiency}.}

The identification theory in \cite{miao2024IdentificationSemiparametricEfficiency} is complicated. Here we summarize the key points leading to identification, leaving out side results. We use our notation $V$ for the variable with self-connected missingness and $Z$ for the shadow variable, and keep covariates implicit. The theory is based on the definition of an odds ratio function (see definition of $\text{OR}(X,Y)$ on the RHS of equation 2 in Proposition 1 in \citeauthor{miao2024IdentificationSemiparametricEfficiency}). In our notation, this odds ratio function is
\begin{align*}
    \text{OR}(V):=\frac{\text{odds}(R_V=0\mid V)}{\text{odds}(R_V=0\mid V=v)}
\end{align*}
for $v$ being a specific value in the support of $V$ (any value is fine). (An alternative label such as $\text{OR}_v(V)$ would make this reference level explicit, but we will continue with the implicit label.) 
A normalized version of this function is then defined (inside Proposition 2) as 
$$\widetilde{\text{OR}}(V):=\text{OR}(V)/\E[\text{OR}(V)\mid R_V=1].$$
$\text{OR}(V)$ can be recovered from $\widetilde{\text{OR}}(V)$ as
\begin{align*}
    \text{OR}(V)
    &=\widetilde{\text{OR}}(V)/\widetilde{\text{OR}}(V=v).
\end{align*}
The theory states that due to $Z$ being a shadow variable, $\widetilde{\text{OR}}(V)$ satisfies
\begin{align*}
    \E[\widetilde{\text{OR}}(V)\mid R_V=1,Z]=\frac{\P(Z\mid R_V=0)}{\P(Z\mid R_V=1)}
\end{align*}
(equation 5 in Proposition 2). This provides an integral equation that identifies $\widetilde{\text{OR}}(V)$ under the completeness condition. Then as a result of $\widetilde{\text{OR}}(V)$ identification, $\text{OR}(V)$ is identified, therefore
the distribution of the missing values conditional on the shadow variable is recovered:
\begin{align*}
    \P(V\mid R_V=0,Z)=\P(V\mid R_V=1,Z)\frac{\text{OR}(V)}{\E[\text{OR}(V)\mid R_V=1,Z]}
\end{align*}
(equation 7 in Proposition 2). This, combined with the fact that $\P(Z,R_V)$ and $\P(V\mid R_V=1,Z)$ are observed data functionals, implies that the joint distribution $\P(V,Z,R_V)$ is identified (Theorem 1). 

Now we examine this theory with a view to connect it to our simple theory with the odds connection (\ref{eq:shadow}). One element that makes this theory complex is that it is built on the odds ratio function $\text{OR}(V)$, which contains a reference odds, $\text{odds}(R_V=0\mid V=v)$. Identification of $\P(V,Z,R_V)$ hangs on identification of $\widetilde{\text{OR}}(V)$ using an integral equation, however, so we examine this function first. Note that in the expression that defines $\widetilde{\text{OR}}(V)$, the reference odds appears in both the numerator and denominator, so it is canceled out, and $\widetilde{\text{OR}}(V)$ does not involve $v$. $\widetilde{\text{OR}}(V)$ can thus be thought of as a normalized odds rather than a normalized odds ratio:
\begin{align*}
    \widetilde{\text{OR}}(V)
    &=\text{odds}(R_V=0\mid V)/\E\left[\text{odds}(R_V=0\mid V)\mid R_V=1\right].
\end{align*}
Also, the denominator of this function can be simplified,
\begin{align*}
    \E\left[\frac{\P(R_V=0\mid V)}{\P(R_V=1\mid V)}\mid R_V=1\right]
    &=\E\left[\frac{\P(R_V=0,V)}{\P(R_V=1,V)}\mid R_V=1\right]
    \\
    &=\E\left[\frac{\P(R_V=0)\P(V\mid R_V=0)}{\P(R_V=1)\P(V\mid R_V=1)}\mid R_V=1\right]
    \\
    &=\frac{\P(R_V=0)}{\P(R_V=1)}\E\left[\frac{\P(V\mid R_V=0)}{\P(V\mid R_V=1)}\mid R_V=1\right]
    \\
    &=\frac{\P(R_V=0)}{\P(R_V=1)},
\end{align*}
which means $\widetilde{\text{OR}}(V)$ is the \textit{stabilized} version of the self-conditional missingness odds,
\begin{align*}
    \widetilde{\text{OR}}(V)
    &=\text{odds}(R_V=0\mid V)/\text{odds}(R_V=0).
\end{align*}
This means identification of $\widetilde{\text{OR}}(V)$ (in this theory) is equivalent to identification of the missingness odds (in our theory). Plugging this new expression of $\widetilde{\text{OR}}(V)$ into the earlier mentioned equality used as the integral equation for $\widetilde{\text{OR}}(V)$ identification obtains
\begin{align*}
    \E\left[\text{odds}(R_V=0\mid V)\mid R_V=1,Z\right]/\text{odds}(R_V=0)=\frac{\P(Z\mid R_V=0)}{\P(Z\mid R_V=1)}.
\end{align*}
Moving the denominator on the LHS to the RHS, then dividing both the numerator and denominator of the RHS by $\P(Z)$, obtains our key result
\begin{align*}
    \E\left[\text{odds}(R_V=0\mid V)\mid R_V=1,Z\right]=\text{odds}(R_V=0\mid Z).\tag{\ref{eq:shadow}}
\end{align*}
This means (\ref{eq:shadow}) is basically the underlying connection that helps identify $\widetilde{\text{OR}}(V)$, which is at the heart of \citeauthor{miao2024IdentificationSemiparametricEfficiency}'s theory. This is the connection we were looking for. The job is done!

To complete the story, though, we revisit the role of $\text{OR}(V)$, which we skipped earlier. Our theory does not use this function or any other function that involves a reference level $v$. \citeauthor{miao2024IdentificationSemiparametricEfficiency}, on the other hand, consider this function essential to identification. We note that this function appears in the identification result for $\P(V\mid R_V=0,Z)$ above in the specific form $\text{OR}(V)/\E[\text{OR}(V)\mid R_V=1,Z]$, so now we pay attention to this form rather than $\text{OR}(V)$ itself. This form, however, does not involve the reference level $v$:
\begin{align*}
    \frac{\text{OR}(V)}{\E[\text{OR}(V)\mid R_V=1,Z]}&=\frac{\text{odds}(R_V=0\mid V)}{\E\left[\text{odds}(R_V=0\mid V)\mid R_V=1,Z\right]}.
\end{align*}

Note that $\text{odds}(R_V=0\mid V)$ plays the central role in both this expression and the last expression of $\widetilde{\text{OR}}(V)$. This means that the theory can be simplified by making identification of $\text{odds}(R_V=0\mid V)$ the central task. This is exactly what we do in our simple theory.

\paragraph{Connections with \cite{dhaultfoeuille2010NewInstrumentalMethod}.}

Adding 1 to both sides of (\ref{eq:shadow}) and the multiplying both sides by $\P(R_V=1\mid Z)$ obtains
\begin{align*}
    1=\E\left[\frac{\P(R_V=1\mid Z)}{\P(R_V=1\mid V)}\mid Z,R_V=1\right].
\end{align*}
The RHS can be re-expressed as
\begin{align*}
    \text{RHS}
    &=\E\left[\frac{\P(R_V=1\mid Z)}{\P(R_V=1\mid V)}\frac{\P(V\mid Z,R_V=1)}{\P(V\mid Z)}\mid Z\right]
    \\
    &=\E\left[\frac{\P(R_V=1\mid Z)}{\P(R_V=1\mid V)}\frac{\P(R_V=1\mid V,Z)}{\P(R_V=1\mid Z)}\mid Z\right]
    \\
    &=\E\left[\frac{\P(R_V=1\mid V,Z)}{\P(R_V=1\mid V)}\mid Z\right]
    \\
    &=\E\left\{\frac{\E[R_V\mid V,Z]}{\P(R_V=1\mid V)}\mid Z\right\}
    \\
    &=\E\left\{\E\left[\frac{R_V}{\P(R_V=1\mid V)}\mid V,Z\right]\mid Z\right\}
    \\
    &=\E\left[\frac{R_V}{\P(R_V=1\mid V)}\mid Z\right],
\end{align*}
which is equation (2.4) in \cite{dhaultfoeuille2010NewInstrumentalMethod}.

\subsubsection{Missingness in the shadow variable}
\label{asssec:shadow-missing}
 
\paragraph{The \ref{mSV-2i} case: proof of (\ref{eq:shadow-blue}).}

Recall the \ref{mSV-2i} assumption: $R_Z\independent(V,Z)\mid R_V$.

We start with the original result,
\begin{align}
    \text{odds}(R_V=0\mid Z)=\E[\text{odds}(R_V=0\mid V)\mid Z,R_V=1].\tag{\ref{eq:shadow}}
\end{align}
Using the odds-tilting technique (\ref{eq:odds-tilt1}) and then applying \ref{mSV-2i}, we have
\begin{align*}
    \frac{\text{odds}(R_V=0\mid Z)}{\text{odds}(R_V=0\mid Z,R_Z=1)}
    \overset{\text{(\ref{eq:odds-tilt1})}}{=}\frac{\P(R_Z=1\mid Z,R_V=1)}{\P(R_Z=1\mid Z,R_V=0)}
    \overset{\text{\ref{mSV-2i}}}{=}\frac{\P(R_Z=1\mid R_V=1)}{\P(R_Z=1\mid R_V=0)},
\end{align*}
which identifies the LHS of (\ref{eq:shadow}).

The RHS of (\ref{eq:shadow}) is an expectation over $\P(V\mid Z,R_V=1)$. \ref{mSV-2i} implies $R_Z\independent W\mid Z,R_V=1$, so $\P(V\mid Z,R_V=1)=\P(V\mid Z,R_{ZV}=1)$. Combining results for the LHS and RHS of (\ref{eq:shadow}), we have
\begin{align}
    \text{odds}(R_V=0\mid Z,R_Z=1)\frac{\P(R_Z=1\mid R_V=1)}{\P(R_Z=1\mid R_V=0)}=\E[\text{odds}(R_V=0\mid V)\mid Z,R_{ZV}=1].\tag{\ref{eq:shadow-blue}}
\end{align}

\paragraph{The \ref{mSV-2ii} case: proof of (\ref{eq:shadow-red}) and (\ref{eq:shadow-red-dagger}).}

Recall the \ref{mSV-2ii} assumption: $R_Z\independent R_V\mid V,Z$.

\begin{align*}
    \P(R_V\mid V)\overset{Z\independent R_V\mid V}{=}\P(R_V\mid V,Z)\overset{\text{\ref{mSV-2ii}}}{=}\P(R_V\mid V,Z,R_Z).
\end{align*}
This means 
\begin{align}
    (Z,R_Z)\independent R_V\mid V.\label{eq:fish}
\end{align}
First, this implies $Z\independent R_V\mid V,R_Z=1$, i.e., $Z$ is a SV w.r.t. the missingness in $V$ in the $R_Z=1$ population. Applying (\ref{eq:shadow}) in this context, we have
\begin{align*}
    \text{odds}(R_V=0\mid Z,R_Z=1)=\E[\text{odds}(R_V=0\mid V,R_Z=1)\mid Z,R_Z=1,R_V=1].
\end{align*}
The odds inside the expectation can be replaced with $\text{odds}(R_V=0\mid V,R_Z=1)$ due to (\ref{eq:fish}), so this equation becomes
\begin{align*}
    \text{odds}(R_V=0\mid Z,R_Z=1)=\E[\text{odds}(R_V=0\mid V)\mid Z,R_{ZV}=1].\tag{\ref{eq:shadow-red}}
\end{align*}

Second, (\ref{eq:fish}) implies $Z^\dagger\independent R_V\mid V$, i.e., $Z^\dagger$ is a SV w.r.t. the missingness in $V$. Applying (\ref{eq:shadow}) in this context, we have
\begin{align*}
    \text{odds}(R_V=0\mid Z^\dagger)=\E[\text{odds}(R_V=0\mid V)\mid Z^\dagger,R_V=1].\tag{\ref{eq:shadow-red-dagger}}
\end{align*}

\subsection{Specific theory: external shadow for a pair of variables}
\label{assec:shadow-theory2}

Here we lay out mSV-type assumptions (SV missingness assumptions) that pair with SV model assumptions A1-A4 and C1-C4, and the corresponding identification results.

\subsubsection{mSV-type assumptions to pair with SV model assumptions A1-A4}

\paragraph{Assumption mSV(A1)}\hfill

\begin{align*}
    \P(&R_Z=1,R_{MY}=1\mid M,Y,Z,(X,A))>0,
    \\
    \text{and either}~&R_Z\independent(M,Y,Z)\mid R_{MY},(X,A)
    \\
    \text{or}~&R_Z\independent R_{MY}\mid M,Y,Z,(X,A).
\end{align*}

\paragraph{Assumption mSV(A2)}\hfill

{\small
\begin{minipage}{.48\textwidth}
    (M-part, conditional)
    \begin{align*}
    \P(&R_Z=1\mid M,Z,R_M=1,(X,A,Y,R_Y=1))>0,
    \\
    \text{either}~&R_Z\independent(M,Z)\mid R_M,(X,A,Y,R_Y=1)
    \\
    \text{or}~&R_Z\independent R_M\mid M,Z,(X,A,Y,R_Y=1).
    \end{align*}
\end{minipage}
\hspace{.02\textwidth}
\begin{minipage}{.48\textwidth}
    (Y-part, conditional)
    \begin{align*}
    \P(&R_Z=1\mid Y,Z,R_Y=1,(X,A,M,R_M=1))>0,
    \\
    \text{either}~&R_Z\independent(Y,Z)\mid R_Y,(X,A,M,R_M=1)
    \\
    \text{or}~&R_Z\independent R_Y\mid Y,Z,(X,A,M,R_M=1).
    \end{align*}
\end{minipage}
}

\paragraph{Assumption mSV(A3)}\hfill

{\small
\begin{minipage}{.48\textwidth}
    (M-part, unconditional)
    \begin{align*}
    \P(&R_Z=1\mid M,Z,R_M=1,(X,A))>0,
    \\
    \text{either}~&R_Z\independent(M,Z)\mid R_M,(X,A)
    \\
    \text{or}~&R_Z\independent R_M\mid M,Z,(X,A).
    \end{align*}
\end{minipage}
\hspace{.02\textwidth}
\begin{minipage}{.48\textwidth}
    (Y-part, conditional)
    \begin{align*}
    \P(&R_Z=1\mid Y,Z,R_Y=1,(X,A,M,R_M=1))>0,
    \\
    \text{either}~&R_Z\independent(Y,Z)\mid R_Y,(X,A,M,R_M=1)
    \\
    \text{or}~&R_Z\independent R_Y\mid Y,Z,(X,A,M,R_M=1).
    \end{align*}
\end{minipage}
}

\paragraph{Assumption mSV(A4)}\hfill

{\small
\begin{minipage}{.48\textwidth}
    (M-part, conditional)
    \begin{align*}
    \P(&R_Z=1\mid M,Z,R_M=1,(X,A,Y,R_Y=1))>0,
    \\
    \text{either}~&R_Z\independent(M,Z)\mid R_M,(X,A,Y,R_Y=1)
    \\
    \text{or}~&R_Z\independent R_M\mid M,Z,(X,A,Y,R_Y=1).
    \end{align*}
\end{minipage}
\hspace{.02\textwidth}
\begin{minipage}{.48\textwidth}
    (Y-part, unconditional)
    \begin{align*}
    \P(&R_Z=1\mid Y,Z,R_Y=1,(X,A))>0,
    \\
    \text{either}~&R_Z\independent(Y,Z)\mid R_Y,(X,A)
    \\
    \text{or}~&R_Z\independent R_Y\mid Y,Z,(X,A).
    \end{align*}
\end{minipage}
}

We have stated the four assumptions separately above and labeled them to indicate the SV model assumptions clearly. It is important to note that mSV(A1) implies mSV(A2), which explains why Figure~\ref{fig:shadow-missing} in the paper combines A1 and A2 in the DAGs in the top row.

\subsubsection{Modification of identification results under assumptions mSV(A1) to mSV(A4) -- supplement to Table~\ref{tab:shadow-auxiliary}}

Each of these assumptions has two versions, one where $R_Z$ is separated from $Z$ and the variables with missingness (indicated by the first element in the either-or component of the assumption), and one where $R_Z$ is separated from the missingness indicator of such variables (the second element in the either-or component). Below we will refer to these two versions as (i) and (ii).

\paragraph{Under version (i) of the assumptions}\hfill

The identification results in Table~\ref{tab:shadow-auxiliary} keep the same form, but with redefined $q$ functions. Specifically, the functions $q_{12d},q_{1e},q_{2e},q_{1f},q_{2f}$ functions in the table are to be replaced, respectively, with the solutions to
\begin{align}
    1+\text{odds}(R_{MY}=0\mid X,A,Z,R_Z=1)&\frac{\P(R_Z=1\mid X,A,R_{MY}=1)}{\P(R_Z=1\mid X,A,R_{MY}=0)}=\nonumber
    \\
    &~~~~\E[{\color{purple}q_{12d}(X,A,M,Y)}\mid X,A,Z,R_{MY}=1,R_Z=1],
    \\
    1+\text{odds}(R_M=0\mid X,A,Y,Z,R_Y=1,&R_Z=1)\frac{\P(R_Z=1\mid X,A,Y,R_Y=1,R_M=1)}{\P(R_Z=1\mid X,A,Y,R_Y=1,R_M=0)}=\nonumber
    \\
    &~~\E[{\color{purple}q_{1e}(X,A,M,Y)}\mid X,A,Y,Z,R_{MY}=1,R_Z=1],
    \\
    1+\text{odds}(R_Y=0\mid X,A,M,Z,R_M=1,&R_Z=1)\frac{\P(R_Z=1\mid X,A,M,R_M=1,R_Y=1)}{\P(R_Z=1\mid X,A,M,R_M=1,R_Y=0)}=\nonumber
    \\
    &~~\E[{\color{purple}q_{2e}(X,A,M,Y)}\mid X,A,M,Z,R_{MY}=1,R_Z=1],
    \\
    1+\text{odds}(R_M=0\mid X,A,Z,R_Z=1)&\frac{\P(R_Z=1\mid X,A,R_M=1)}{\P(R_Z=1\mid X,A,R_M=0)}=\nonumber
    \\
    &~~~~~\E[{\color{purple}q_{1f}(X,A,M)}\mid X,A,Z,R_M=1,R_Z=1],
    \\
    1+\text{odds}(R_Y=0\mid X,A,Z,R_Z=1)&\frac{\P(R_Z=1\mid X,A,R_Y=1)}{\P(R_Z=1\mid X,A,R_Y=0)}=\nonumber
    \\
    &~~~~~\E[{\color{purple}q_{2f}(X,A,Y)}\mid X,A,Z,R_Y=1,R_Z=1].
\end{align}

These equations are applications of (\ref{eq:shadow-blue}).

\paragraph{Under version (ii) of the assumptions}\hfill

The functions $q_{12d},q_{1e},q_{2e},q_{1f},q_{2f}$ functions in Table~\ref{tab:shadow-auxiliary} are to be replaced, respectively, with the solutions to
\begin{align}
    [\P(R_{MY}=1\mid X,A,Z,R_Z=1)]^{-1}&=\E[{\color{purple}q_{12d}(X,A,M,Y)}\mid X,A,Z,R_{MY}=1,R_Z=1],
    \\
    [\P(R_M=1\mid X,A,Y,Z,R_Y=1,R_Z=1)]^{-1}&=\E[{\color{purple}q_{1e}(X,A,M,Y)}\mid X,A,Y,Z,R_{MY}=1,R_Z=1],
    \\
    [\P(R_Y=1\mid X,A,M,Z,R_M=1,R_Z=1)]^{-1}&=\E[{\color{purple}q_{2e}(X,A,M,Y)}\mid X,A,M,Z,R_{MY}=1,R_Z=1],
    \\
    [\P(R_M=1\mid X,A,Z,R_Z=1)]^{-1}&=\E[{\color{purple}q_{1f}(X,A,M)}\mid X,A,Z,R_M=1,R_Z=1],
    \\
    [\P(R_Y=1\mid X,A,Z,R_Z=1)]^{-1}&=\E[{\color{purple}q_{2f}(X,A,Y)}\mid X,A,Z,R_Y=1,R_Z=1].
\end{align}

These equations are applications of (\ref{eq:shadow-red}).




\subsubsection{mSV-type assumptions to pair with SV model assumptions C1-C4}

\paragraph{Assumption mSV(C1)}\hfill

\begin{align*}
    \P(&R_Z=1,R_{MY}=1\mid M,Y,Z,W)>0,
    \\
    \text{and either}~&R_Z\independent(M,Y,Z)\mid R_{MY},W
    \\
    \text{or}~&R_Z\independent R_{MY}\mid M,Y,Z,W.
\end{align*}

\paragraph{Assumption mSV(C2)}\hfill

{\small
\begin{minipage}{.48\textwidth}
    (M-part, conditional)
    \begin{align*}
    \P(&R_Z=1\mid M,Z,R_M=1,(W,Y,R_Y=1))>0,
    \\
    \text{either}~&R_Z\independent(M,Z)\mid R_M,(W,Y,R_Y=1)
    \\
    \text{or}~&R_Z\independent R_M\mid M,Z,(W,Y,R_Y=1).
    \end{align*}
\end{minipage}
\hspace{.02\textwidth}
\begin{minipage}{.48\textwidth}
    (Y-part, conditional)
    \begin{align*}
    \P(&R_Z=1\mid Y,Z,R_Y=1,(W,M,R_M=1))>0,
    \\
    \text{either}~&R_Z\independent(Y,Z)\mid R_Y,(W,M,R_M=1)
    \\
    \text{or}~&R_Z\independent R_Y\mid Y,Z,(W,M,R_M=1).
    \end{align*}
\end{minipage}
}

\paragraph{Assumption mSV(C3)}\hfill

{\small
\begin{minipage}{.48\textwidth}
    (M-part, unconditional)
    \begin{align*}
    \P(&R_Z=1\mid M,Z,R_M=1,W)>0,
    \\
    \text{either}~&R_Z\independent(M,Z)\mid R_M,W
    \\
    \text{or}~&R_Z\independent R_M\mid M,Z,W.
    \end{align*}
\end{minipage}
\hspace{.02\textwidth}
\begin{minipage}{.48\textwidth}
    (Y-part, conditional)
    \begin{align*}
    \P(&R_Z=1\mid Y,Z,R_Y=1,(W,M,R_M=1))>0,
    \\
    \text{either}~&R_Z\independent(Y,Z)\mid R_Y,(W,M,R_M=1)
    \\
    \text{or}~&R_Z\independent R_Y\mid Y,Z,(W,M,R_M=1).
    \end{align*}
\end{minipage}
}

\paragraph{Assumption mSV(C4)}\hfill

{\small
\begin{minipage}{.48\textwidth}
    (M-part, conditional)
    \begin{align*}
    \P(&R_Z=1\mid M,Z,R_M=1,(W,Y,R_Y=1))>0,
    \\
    \text{either}~&R_Z\independent(M,Z)\mid R_M,(W,Y,R_Y=1)
    \\
    \text{or}~&R_Z\independent R_M\mid M,Z,(W,Y,R_Y=1).
    \end{align*}
\end{minipage}
\hspace{.02\textwidth}
\begin{minipage}{.48\textwidth}
    (Y-part, unconditional)
    \begin{align*}
    \P(&R_Z=1\mid Y,Z,R_Y=1,W)>0,
    \\
    \text{either}~&R_Z\independent(Y,Z)\mid R_Y,W
    \\
    \text{or}~&R_Z\independent R_Y\mid Y,Z,W.
    \end{align*}
\end{minipage}
}

Here also we have the connection mSV(C1) implies mSV(C2).

\subsubsection{Modification of identification results under assumptions mSV(C1) to mSV(C4) -- supplement to Table~\ref{tab:shadow-covariate}}

Each of these assumptions has two versions, one where $R_Z$ is separated from $Z$ and the variables with missingness (indicated by the first element in the either-or component of the assumption), and one where $R_Z$ is separated from the missingness indicator of such variables (the second element in the either-or component). Below we will refer to these two versions as (i) and (ii).

\paragraph{Under version (i) the assumptions}\hfill

The identification results in Table~\ref{tab:shadow-covariate} keep the same form, but with redefined $q^*$ functions. Specifically, the functions $q_{12d}^*,q_{1e}^*,q_{2e}^*,q_{1f}^*,q_{2f}^*$ functions in the table are to be replaced, respectively, with the solutions to
\begin{align}
    1+\text{odds}(R_{MY}=0\mid W,Z,R_Z=1)&\frac{\P(R_Z=1\mid W,R_{MY}=1)}{\P(R_Z=1\mid W,R_{MY}=0)}=\nonumber
    \\
    &~~~~\E[{\color{purple}q_{12d}^*(W,M,Y)}\mid W,Z,R_{MY}=1,R_Z=1],
    \\
    1+\text{odds}(R_M=0\mid W,Y,Z,R_Y=1,&R_Z=1)\frac{\P(R_Z=1\mid W,Y,R_Y=1,R_M=1)}{\P(R_Z=1\mid W,Y,R_Y=1,R_M=0)}=\nonumber
    \\
    &~~\E[{\color{purple}q_{1e}^*(W,M,Y)}\mid W,Y,Z,R_{MY}=1,R_Z=1],
    \\
    1+\text{odds}(R_Y=0\mid W,M,Z,R_M=1,&R_Z=1)\frac{\P(R_Z=1\mid W,M,R_M=1,R_Y=1)}{\P(R_Z=1\mid W,M,R_M=1,R_Y=0)}=\nonumber
    \\
    &~~\E[{\color{purple}q_{2e}^*(W,M,Y)}\mid W,M,Z,R_{MY}=1,R_Z=1],
    \\
    1+\text{odds}(R_M=0\mid W,Z,R_Z=1)&\frac{\P(R_Z=1\mid W,R_M=1)}{\P(R_Z=1\mid W,R_M=0)}=\nonumber
    \\
    &~~~~~\E[{\color{purple}q_{1f}^*(W,M)}\mid W,Z,R_M=1,R_Z=1],
    \\
    1+\text{odds}(R_Y=0\mid W,Z,R_Z=1)&\frac{\P(R_Z=1\mid W,R_Y=1)}{\P(R_Z=1\mid X,A,R_Y=0)}=\nonumber
    \\
    &~~~~~\E[{\color{purple}q_{2f}^*(W,Y)}\mid W,Z,R_Y=1,R_Z=1].
\end{align}

These equations are applications of (\ref{eq:shadow-blue}).

\paragraph{Under version (ii) of the assumptions}\hfill

The functions $q_{12d}^*,q_{1e}^*,q_{2e}^*,q_{1f}^*,q_{2f}^*$ functions in Table~\ref{tab:shadow-covariate} are to be replaced, respectively, with the solutions to
\begin{align}
    [\P(R_{MY}=1\mid W,Z,R_Z=1)]^{-1}&=\E[{\color{purple}q_{12d}^*(W,M,Y)}\mid W,Z,R_{MY}=1,R_Z=1],
    \\
    [\P(R_M=1\mid W,Y,Z,R_Y=1,R_Z=1)]^{-1}&=\E[{\color{purple}q_{1e}^*(W,M,Y)}\mid W,Y,Z,R_{MY}=1,R_Z=1],
    \\
    [\P(R_Y=1\mid W,M,Z,R_M=1,R_Z=1)]^{-1}&=\E[{\color{purple}q_{2e}^*(W,M,Y)}\mid W,M,Z,R_{MY}=1,R_Z=1],
    \\
    [\P(R_M=1\mid W,Z,R_Z=1)]^{-1}&=\E[{\color{purple}q_{1f}^*(W,M)}\mid W,Z,R_M=1,R_Z=1],
    \\
    [\P(R_Y=1\mid W,Z,R_Z=1)]^{-1}&=\E[{\color{purple}q_{2f}^*(W,Y)}\mid W,Z,R_Y=1,R_Z=1].
\end{align}

These equations are applications of (\ref{eq:shadow-red}).

\end{document}